%% file: thesis.tex
\documentclass[twoside,a4paper]{memoir}

\usepackage{amsmath,amsthm,amssymb,color,url,mathpazo,wrapfig}
\usepackage{graphicx}

\settocdepth{subsection}

\definecolor{grey}{gray}{.4}

\newcommand{\gap}{\hskip 0.4em plus 0.1em minus 0.2em}

\makeatletter
\def\l@figure{\@dottedtocline{1}{0 em}{7.5 em}}
\makeatother

\theoremstyle{plain}
\newtheorem{theorem}{Theorem}[chapter] 

\newtheorem{lemma}[theorem]{Lemma}

\theoremstyle{definition}

\newtheorem{definition}[theorem]{Definition}

\renewcommand{\vec}[1]{\mathbf{#1}}
\newcommand{\vc}[1]{\mathbf{#1}}

\newcommand{\mat}[1]{\mathsf{#1}}
\newcommand{\mt}[1]{\mathsf{#1}}

\renewcommand{\bar}{\overline}

\newcommand{\one}{\mathtt{1}}
\newcommand{\zero}{\mathtt{0}}
\newcommand{\FF}{\{\zero,\one\}}
\newcommand{\zo}{\{\zero,\one\}}
\newcommand{\q}{\mathtt{q}}

\newcommand{\I}{\mathbb{I}}
\newcommand{\HH}{\mathbb{H}}
\newcommand{\D}{\mathbb{D}}

\newcommand{\F}{\mathbb{F}}
\newcommand{\bP}{\mathbb{P}}
\newcommand{\PP}{\mathbb{P}}
\newcommand{\bN}{\mathbb{N}}
 
\newcommand{\bR}{\mathbb{R}}
\newcommand{\RR}{\mathbb{R}}

\newcommand{\EE}{\mathbb{E}}

\newcommand{\bZ}{\mathbb{Z}}

\newcommand{\CC}{\mathbb{C}}
\newcommand{\C}{\mathbb{C}}
\newcommand{\K}{\mathcal K}
\renewcommand{\L}{\mathcal L}
\newcommand{\J}{\mathcal J}
\newcommand{\Khat}{\hat \K}

\newcommand{\zeroone}{ \{\mathtt 0, \mathtt 1 \} }
\newcommand{\quest}{\text{\texttt{?}}}

\newcommand{\e}{\mathrm{e}}
\newcommand{\given}{\mid}
\newcommand{\Prob}{\mathbb{P}}
\DeclareMathOperator{\probab}{\bP}
\newcommand{\prob}[1]{\probab{\left(#1\right)}}

\DeclareMathOperator{\var}{Var}
\DeclareMathOperator{\Cov}{Cov}
\DeclareMathOperator{\CN}{\mathbb{C}N}

\newcommand{\N}{\mathcal{N}}
\newcommand{\link}{\!\!\to\!\!}

\renewcommand{\hat}{\widehat}
\renewcommand{\tilde}{\widetilde}
\newcommand{\ex}{\mathrm{e}}
\newcommand{\ud}{\mathrm{d}}
\newcommand{\ii}{\mathrm{i}}

\newcommand{\defn}[1]{\emph{#1}}

\newcommand{\sn}[1]{\textsc{#1}}
\newcommand{\jn}[1]{\emph{#1}}
\renewcommand{\tt}[1]{\emph{#1}}
\newcommand{\vol}[1]{\textbf{#1}}
\newcommand{\arx}[1]{\texttt{arXiv:#1}}

\newcommand{\SINR}{\mathtt{SINR}}
\newcommand{\DOF}{\mathtt{dof}}
\newcommand{\dof}{\mathtt{dof}}
\newcommand{\sinr}{\mathtt{sinr}}
\newcommand{\INR}{\mathtt{INR}}
\newcommand{\inr}{\mathtt{inr}}
\newcommand{\SNR}{\mathtt{SNR}}
\newcommand{\snr}{\mathtt{snr}}
\newcommand{\Out}{\mathrm{Out}}

\newcommand{\csum}{c_\Sigma}
\newcommand{\Csum}{C_\Sigma} 

\newcommand{\csep}{k_\text{sep}}
\newcommand{\cdec}{k_\text{att}}
\newcommand{\pout}{p_\text{out}}
\newcommand{\JAP}{\text{JAP}}
\newcommand{\JAPB}{\text{JAP-B}}

\renewcommand{\vec}{\mathbf}

\newcommand{\ee}{\mathrm{e}}
\DeclareMathOperator{\Var}{Var}
\DeclareMathOperator{\Ex}{\mathbb{E}}

\newcommand{\hbnew}{\bigskip}
\newcommand{\hbtitle}[1]{\noindent \textbf{#1}. }

\begin{document}

\setSingleSpace{1.2}
\SingleSpace

\nouppercaseheads

\include{chapters/frontmatter}
\include{chapters/introduction}

\include{chapters/interference}


\include{chapters/poisson}

\include{chapters/sumcapacity}

\include{chapters/delay}

\include{chapters/grouptesting}

\include{chapters/endmatter}

\end{document}

%% file: chapters/frontmatter.tex
\pagestyle{empty}

\bigskip\mbox{}\medskip

\begin{center}
{\HUGE \textbf{Interference Mitigation \\
       in Large Random \\
       Wireless Networks\\}}

\bigskip\mbox{}\bigskip\mbox{}\bigskip\mbox{}\bigskip\mbox{}\smallskip

{\Huge Matthew Aldridge}

\end{center}

\cleardoublepage

\pagestyle{empty}

\bigskip\mbox{}\medskip

\begin{center}
{\HUGE \textbf{Interference Mitigation \\
       in Large Random \\
       Wireless Networks\\}}

\bigskip\mbox{}\bigskip\mbox{}\bigskip\mbox{}\bigskip\mbox{}\smallskip

{\Huge Matthew Aldridge}

\bigskip\mbox{}\bigskip\mbox{}\bigskip

\bigskip\mbox{}\bigskip\mbox{}\bigskip

{\Large \emph{A dissertation submitted to \\
              the University of Bristol in \\
              accordance with the requirements \\
              for award of the degree of \\
              Doctor of Philosophy in \\
              the Faculty of Science \\}}

\bigskip\mbox{}\bigskip

{\LARGE School of Mathematics \\
\smallskip
2011 \\ }

\end{center}

\bigskip\mbox{}\bigskip\mbox{}\bigskip

\begin{flushright}
30,000 words \end{flushright}

\cleardoublepage

\pagestyle{headings}

\addcontentsline{toc}{section}{\protect\hspace{-1.5 em} \emph{Abstract}}
\markboth{Abstract}{Abstract}
\chapter*{Abstract}

A central problem in the operation of large wireless networks is 
how to deal with interference -- the unwanted signals being sent
by transmitters that a receiver is not interested in.  This thesis
looks at ways of combating such interference.

In Chapters 1 and 2, we outline the necessary information and
communication theory background.  We define the concept of capacity --
the highest rate at which information can be sent through a network
with arbitrarily low probability of error.  We also include an overview
of a new set of schemes for dealing with interference known as
interference alignment, paying special attention to a channel-state-based
strategy called ergodic interference alignment.

In Chapter 3, we consider the operation of large regular and random networks
by treating interference as background noise.  We consider the local
performance of a single node, and the global performance of a very
large network.

In Chapter 4, we use ergodic interference alignment to derive the
asymptotic sum-capacity of large random dense networks.  These
networks are derived from a physical model of node placement where
signal strength decays over the distance between transmitters and receivers.

In Chapter 5, we look at methods of reducing the long
time delays incurred by ergodic interference alignment.  We decrease
the delay for full performance of the scheme, and analyse
the tradeoff between reducing delay and lowering the communication rate.

In Chapter 6, we outline a problem of discovering which users interfere with which;
a situation that is equivalent to the problem of pooled group testing for defective items.
We then present some new work that uses information
theoretic techniques to attack group testing.
We introduce for the first time the concept of the group testing
channel, which allows for modelling of a wide range of statistical
error models for testing. We derive new results on the number
of tests required to accurately detect defective items, including
when using sequential `adaptive' tests.

Chapter 7 concludes and gives pointers for further work.


\cleardoublepage

\addcontentsline{toc}{section}{\protect\hspace{-1.5 em} \emph{Acknowledgments}}
\chapter*{Acknowledgments}

This thesis would not exist without the help of a great many people.

My supervisors are Oliver Johnson and Robert Piechocki -- a lot of the work
in this thesis is joint work with them.  Without Olly and Rob's wisdom, insight, knowledge, help, guidance,
encouragement, and gentle prodding, none of this would have been possible.
Thank you, Olly; thank you, Rob.

My research for this thesis was funded by Toshiba Research Europe Ltd.  Thanks to the Toshiba Telecommunication Research Laboratory in Bristol and its directors.  I have also received support from the Engineering and Physical Sciences Research Council, via the University of Bristol `Bridging the Gaps' cross-disciplinary feasibility account (\texttt{EP/H024786/1}). Thanks to Dino Sejdinovic and Olly, who got me involved.


This thesis has benefited, directly or indirectly, from conversations with Olly, Rob, Dino,
Henry Arnold, Kara Barwell, Lee Butler, Laura Childs, Justin Coon, Andreas M\"uller, Clare Raychaudhuri, Magnus Sandell, Andrew Smith,
and Will Thompson; conversations that were usually fun as well as useful.  It has been checked, in whole
or in part, by Olly, Rob, and Laura, and contains
considerably fewer errors (mathematical, grammatical, and typographical) than it would have done
without their care and attention.  My academic reviewers, Ayalvadi Ganesh and Feng Yu, checked three
precursors to this document, and had many helpful suggestions.  Thanks everyone. 

My examiners were Olivier L\'ev\^eque and Ganesh.  They both read this thesis thoroughly, and
made a large number of helpful comments.  Their suggestions for improvements and clarifications have made this
thesis better, a number of references they located have made it more complete,
and fixing the errors they spotted has made it more accurate.

Writing this thesis would not have been possible without free software.
Thanks to Donald Knuth, Leslie Lamport, and the \LaTeX 3 team for \LaTeX;
to the AMS, David Carlise, Lars Madsen, and Peter Wilson for various \LaTeX\ macros;
to Tino Weinkauf, Sven Wiegand, and the \TeX nicCenter team;
and to the Inkscape team.

Thanks Mum, thanks Dad, thanks Alice, for many, many things. And thanks Laura (for being awesome
like a pigeon).

\cleardoublepage

\addcontentsline{toc}{section}{\protect\hspace{-1.5 em} \emph{Author's declaration}}
\chapter*{Author's declaration}

I declare that the work in this dissertation was carried out in accordance with the requirements of the University's Regulations and Code of Practice for Research Degree Programmes and that it has not been submitted for any other academic award. Except where indicated by specific reference in the text, the work is the candidate's own work. Work done in collaboration with, or with the assistance of, others is indicated as such. Any views expressed in the dissertation are those of the author.

\bigskip

\begin{quote}
Signed: 

\medskip

Date:
\end{quote}

\cleardoublepage

\addcontentsline{toc}{section}{\protect\hspace{-1.5 em} \emph{Table of contents}}
\tableofcontents*

\cleardoublepage

\addcontentsline{toc}{section}{\protect\hspace{-1.5 em} \emph{List of definitions and theorems}}
\listoffigures*

\cleardoublepage

\addcontentsline{toc}{chapter}{Introduction}
\markboth{Introduction}{Introduction}
\chapter*{Introduction}

A central problem in the operation of large wireless networks is 
how to deal with \defn{interference} -- the unwanted signals being sent
by transmitters that a receiver is not interested in.  This thesis
looks at ways of combating such interference in large random wireless netoworks.

\bigskip

\noindent In \textbf{Chapter 1: Information}, we briefly summarise
information theory in the single user (point-to-point) case.

A \defn{channel} models how signals are corrupted by noise.
We pay particular attention to the \defn{Gaussian channel}, which is a good
model for real-world wireless communication, and the \defn{finite field
channel}, which can be thought of as a discretisation of the Gaussian channel

The \defn{capacity} of a channel tells us how much information we can send through
the channel for an arbitrarily low probability of error.  \defn{Shannon's channel coding theorem} tells
us how to calculate the capacity of a channel.  We also demonstrate
the capacity of the Gaussian channel under a power constraint.

\defn{Fading} models how signals can decay and distort when sent over long distances.
We investigate three types of fading -- \defn{fixed}, \defn{slow}, and \defn{fast} -- and show
how they affect the channel capacity.

\bigskip

\noindent In \textbf{Chapter 2: Interference}, we extend our study to multiuser
networks.

We look at information theoretic models of wireless networks, concentrating on
the \defn{interference network}, where many transmitter--receiver pairs want to communicate
through the same medium.  This network suffers from the problem of interference.

Weak interference can be ignored and treated as background noise, while strong
interference can be decoded and subtracted.  The main problem for networks
is interference of a similar strength to the desired signal.

We look at \defn{resource division} strategies, which share the channel resources between
the users.  While such schemes are simple to operate, they perform poorly
when the number of users is high.

Of more interest are new \defn{interference alignment} strategies.  These work by the
following idea: if transmitters plan their signals carefully, then for each receiver
the interfering signals can be aligned together, with the desired signal split off separately.
Interference alignment techniques offer potentially far higher performance than resource
division schemes.  We pay particular attention to a channel state-based strategy
called \defn{ergodic interference alignment}.

\bigskip

\noindent \textbf{Chapter 3: Regular and Poisson random networks}, shows how a simple
interference-as-noise technique can be useful when communicating over short hops
in well-structured networks.

In a $d$-dimensional \defn{regular network}, nodes are placed on the grid $\mathbb Z^d$.
We show that if signals decay like $\text{distance}^{-\alpha}$ for $\alpha > d$, then
all nodes can communicate at some fixed rate $r$.  We call this \defn{linear growth},
as the sum-rate of communication of a collection of nodes scales linearly with the
number of nodes.

We also look at nearest-neighbour communication in \defn{Poisson random networks}, where
nodes are placed at random like a Poisson point process.  We give bounds on the
\defn{outage probability}, the chance that a given link is unable to communicate at some
fixed rate.  We also show that linear growth occurs with probability tending to $1$.

\emph{This chapter is joint work with Oliver Johnson and Robert Piechocki.}

\bigskip

\noindent In \textbf{Chapter 4: Sum-capacity of random dense Gaussian interference networks}, we
consider \defn{spatially separated IID networks}  with \defn{power-law
attenuation}, a natural model for wireless networks.  We derive
the asymptotic sum-capacity of such networks by using ergodic interference
alignment to show achievability, and subtle probabilistic and counting techniques
to show the converse.

We also give an alternative proof (with an improved rate of convergence) to a recent
theorem of Jafar on the sum-capacity of large random networks.

\emph{This chapter is joint work with Oliver Johnson and Robert Piechocki.  This
research has been published in \emph{IEEE Transactions on Information Theory} \cite{JAP},
and in the \emph{Proceedings of the 2010 IEEE International Symposium on Information Theory} \cite{AJP}.}

\bigskip

\noindent \textbf{Chapter 5: Delay--rate tradeoff in ergodic interference alignment}
considers the long blocklengths required to perform ergodic interference alignment.
We outline a new scheme called $\JAP(\vec a)$ and study a \defn{beamforming} extension and
derived \defn{child schemes}.

We show how to reduce the time delay for full performance of ergodic interference alignment.
We also show how delay can be reduced even further for the tradeoff of a decrease
in communication rate.  We analyse the best schemes for small networks, and as the
size of the network tends to infinity.

\emph{This chapter is joint work with Oliver Johnson and Robert Piechocki.  This
research has been submiited to \emph{IEEE Transactions on Communication Theory} --
a preprint is available on the \texttt{arXiv} \cite{JAPdelay}.}

\bigskip

\noindent We begin \textbf{Chapter 6: Interference, group testing, and channel coding}
by considering a problem where receivers must aim to detect which transmitters
interfere with them.  We show that our formulation of this problem is equivalent to
the problem of combinatorial \defn{group testing}.

Recent work by Atia and Saligrama has shown how channel coding techniques can shed
light on the problem of group testing.  We extend their results, by defining
\defn{group testing channels}, and identifying the \defn{only-defects-matter property}, under
which an important theorem holds.

We give the first information theoretic analysis of \defn{adaptive group testings} --
where test pools can be constructed sequentially based on previous outcomes -- by drawing
a comparison with the problem of channel coding with feedback.

\bigskip

\noindent The thesis finishes with \textbf{Chapter 7: Conclusions and further work}.

\bigskip

  \begin{center}
    Schematic representation of chapter dependencies
    
    \bigskip

    \includegraphics[scale=0.59]{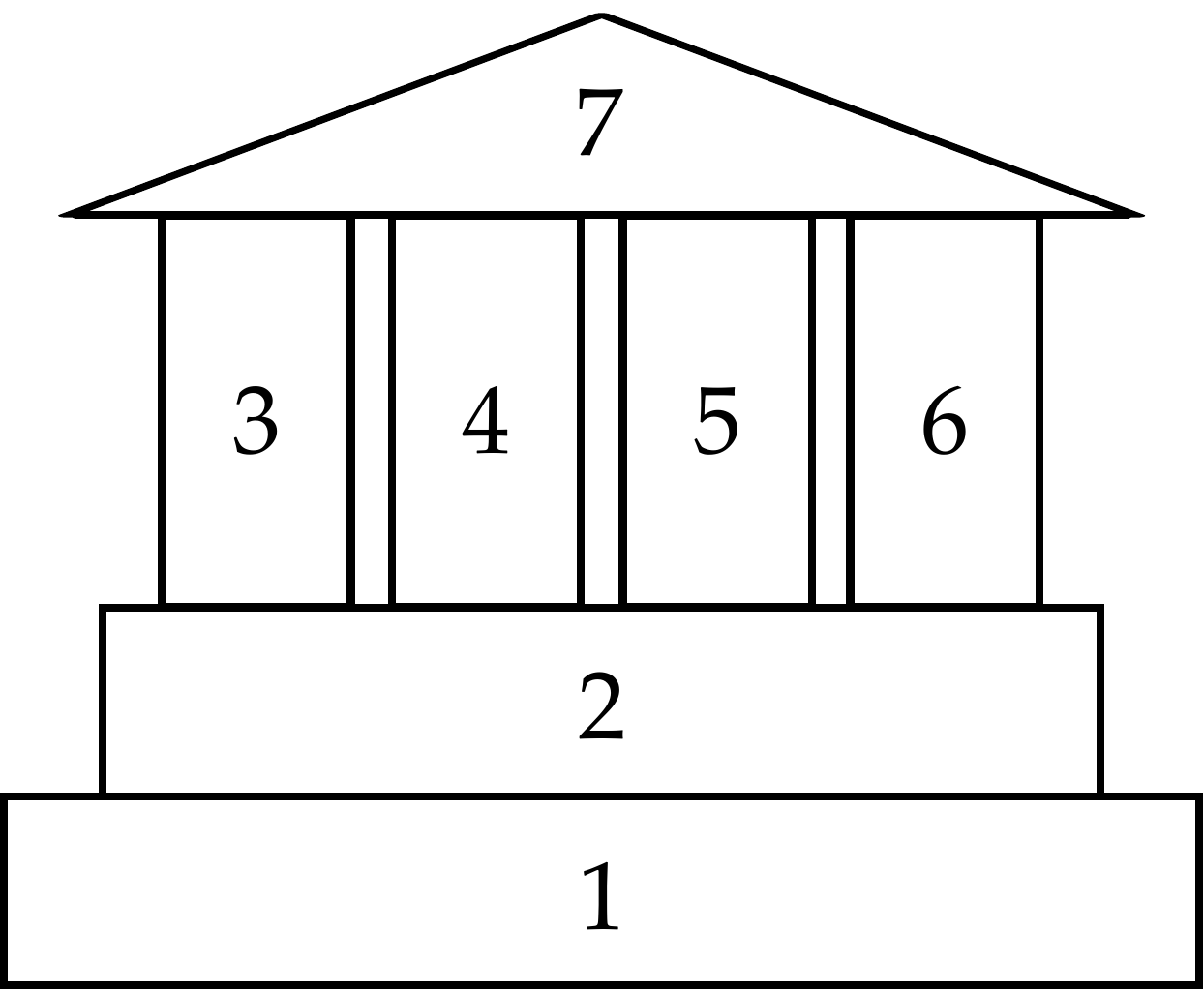}
  \end{center}

%% file: chapters/introduction.tex
\chapter{Information}

\begin{quotation}
  \noindent The fundamental problem of communication is that of reproducing at one
  point either exactly or approximately a message selected at another point.
  
  --- Claude E Shannon
  
  \quad \emph{A Mathematical Theory of Communication} \cite[page 1]{Shannon}
\end{quotation}

In this chapter, we examine the subject of information theory, and in particular
channel coding, which forms the mathematical basis for studying communication.

We start by giving a brief overview of the subject, and making note of
a `handbook' of useful definitions and facts.  We then go through the
more formal mathematics of point-to-point communication, concentrating
on accurate models for real-life wireless communication.  Finally, we
study fading, which allows us to model how signals distort as passed
through space.

\section{Information theory: a very short introduction} \label{vsi}

\defn{Information theory} is the mathematical framework used for studying the transmission of messages in the presence of noise.  It was founded by Claude Shannon in his seminal paper of 1948, ``A mathematical theory of communication'' \cite{Shannon}.  Shannon's information theory involves the sending of a message -- information that a transmitter (typically called Alice) might wish a receiver (Bob) to know -- through a channel, such as a telephone line, an internet connection or a computer cable. 

  \begin{center}
    \includegraphics[scale=0.89]{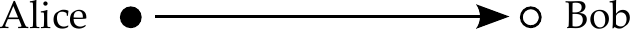}
  \end{center}


The method, language or standard used to transfer the message is called the \defn{code}. (We use `code' in the sense of `Morse code' -- there is no intention to keep the transmitted signal secret as well.)

One of the major goals of information theory involves quantifying how much information can be sent down a channel and how quickly.  For example, if Alice wants to arrange a meeting with Bob, she might send the message ``Hi Bob, meet me at five o'clock on Monday.''  But if the line is very crackly, Bob might interpret the message incorrectly: ``Hi Bob, meet me at nine o'clock on Sunday.''  How could Alice ensure that this doesn't happen?  Perhaps Alice could use a code where she repeats the sentence a number of times, hoping that it would be more likely that Bob could deduce the intended message.  But this takes a longer amount of time -- we say that her \defn{rate} of communication is very low -- and phone calls are expensive, making this undesirable.


Unsurprisingly, there is a trade-off to be made between the rate at which Alice can send the information and the probability that Bob receives it without error.  Before Shannon, it was widely assumed that the only way to make the error probability as small as desired was to reduce the rate of communication toward zero too \cite[Section 5.1]{TseViswanath}.  
However, Shannon discovered that, while this trade-off certainly exists, the error probability can made arbitrarily small while maintaining a communication rate bounded away from zero.  

	\begin{center}
		\includegraphics[width=0.58\textwidth]{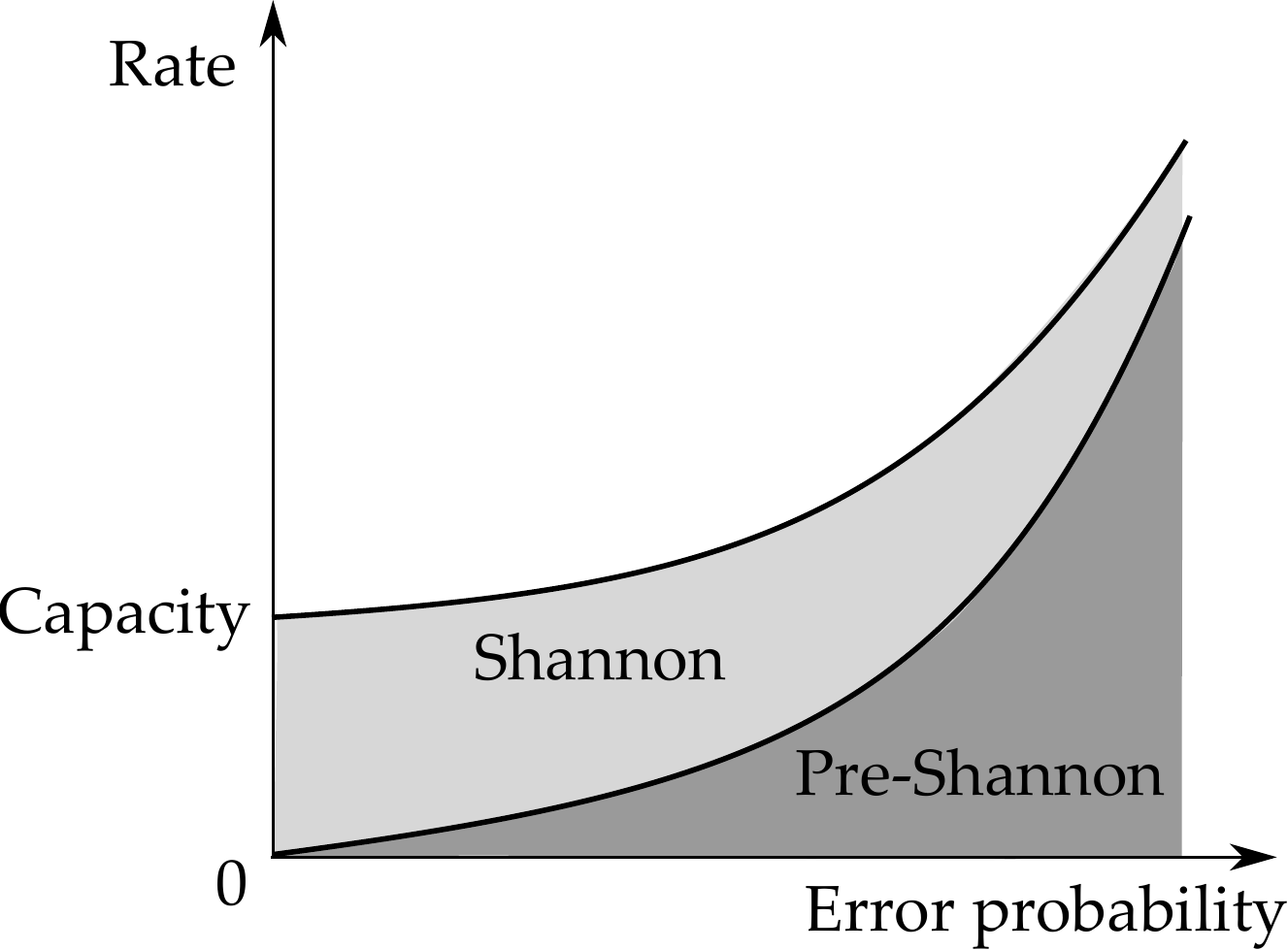}
  \end{center}

In other words, there is a cutoff rate $c$ such that if we attempt to send information at a rate $r<c$, we can do so with an arbitrarily low risk of error, whereas if we attempt to send information at a rate $r>c$, the probability of error is bounded away from $0$.  Shannon called this cutoff rate $c$ the $\defn{capacity}$ of the channel.

	\begin{center}
		\begin{picture}(300,40)(0,0)
      \put(0,20){\vector(1,0){300}}
      \put(20,23){Error-free communication}
      \put(48,10){at rates $r<c$}
      \put(150,17){\line(0,1){6}}
      \put(0,17){\line(0,1){6}}
      \put(-3,8){$0$}
      \put(147,8){$c$}
      \put(213,23){Errors}
      \put(198,10){at rates $r>c$}
      \put(289,10){$r$}
    \end{picture}
  \end{center}

Shannon managed to calculate the capacity of a number of communication channels in terms of the statistical properties of the noise in a channel.
%
%
%
The result is known as \defn{Shannon's channel coding theorem}  (and is stated as Theorem \ref{Shannons2} later).


The aim of this thesis is to find bounds and approximations for capacities of complicated multi-user networks.  In particular, we will be interested in networks that model large real-world wireless networks, such as WiFi computer networks or BlueTooth.

The capacity of a channel, such as a wireless link, tells us the maximum rate at which we can send information while being assured the messages are received accurately.  Note however that merely knowing the capacity does not give us a method of achieving communication at, or even near, capacity.  Nonetheless, the capacity is still a useful benchmark for the quality of a channel.  First, it gives us a `best case' against which we can compare any technologies: if a new code allows us to communicate at a rate near capacity, then this technology is about as good as it's going to get, and there is no need to spend more money on research.  Second, studying the mathematical form of the capacity may help us improve the channel itself; for instance, whether extra resources would be best spent increasing power, bandwidth, or the number of antennas.

\section{Handbook of useful facts}

The following basic concepts of information theory will be referred to often in this thesis; we collect them here for reference.  

More information is available in any basic information theory textbook -- Cover and Thomas's \emph{Elements of Information Theory} \cite{CoverThomas} is a favourite of mine.

\newpage

\noindent\fbox{
\begin{minipage}{0.96\textwidth}

\bigskip

\hbtitle{Mass and density functions} For a discrete random variable $X$, we denote its
\defn{probability mass function} by $p(x) := \Prob(X=x)$.  If $X$ is continuous,
$p(x)$ denotes its \defn{probability density function}.  \defn{Joint} and \defn{conditional mass/density
functions} are denoted $p(x,y)$ and $p(y \given x)$ respectively.  

\hbnew

\hbtitle{Entropy and related concepts} The \defn{entropy} of a discrete random variable
$X$ is
  \begin{equation}
    \HH(X) := \Ex \log \frac{1}{p(X)} = \sum_x p(x) \log \frac{1}{p(x)} , \tag{HB1}
  \end{equation}
where here, as everywhere in this thesis, $\log :\equiv \log_2$ denotes the binary
logarithm.  When $X$ is continuous, the sum is replaced by an integral.  (This
last comment holds for all the following definitions.)

\quad The \defn{joint entropy} of the pair $(X,Y)$ is similarly
  \begin{equation}
    \HH(X,Y) := \Ex \log \frac{1}{p(X,Y)} = \sum_x \sum_y p(x,y) \log \frac{1}{p(x,y)} . \tag{HB2}
  \end{equation}

\quad The \defn{conditional entropy} of $Y$ given $X$ is
  \begin{equation}
    \HH(Y \given X) := \Ex \log \frac{1}{p(Y \given X)} = \sum_x \sum_y p(x,y) \log \frac{1}{p(y \given x)} . \tag{HB3}
  \end{equation}
If $X$ and $Y$ are independent, then it is easy to show \cite[Theorem 2.6.5]{CoverThomas} that
  \begin{equation}
    \HH(Y \given X) = \HH(Y)   \qquad 
    \HH(Y + X \given X) = \HH(Y)   \tag{HB4} 
  \end{equation}
  
\quad The \defn{relative entropy distance} from one probability function $p(x)$ to another
$q(x)$ is
  \begin{equation}
    \D\big( p(x) \,\|\, q(x) \big) := \Ex_{p(x)} \log \frac{p(X)}{q(X)} = \sum_x p(x) \log \frac{p(x)}{q(x)} \geq 0 . \tag{HB5}
  \end{equation}
  
\quad The \defn{mutual information} between $X$ and $Y$ is
  \begin{align}
    \I(X:Y) &:= \D\big( p(x,y) \,\|\, p(x)p(y) \big)  \notag \\
             &= \Ex_{p(x,y)} \log \frac{p(X,Y)}{p(X)p(Y)} \notag \\
             & = \sum_x \sum_y p(x,y) \log \frac{p(x,y)}{p(x)p(y)} . \tag{HB6} 
  \end{align}
It easy to show \cite[Theorem 2.4.1]{CoverThomas} that
  \begin{equation}
    \I(X:Y) = \HH(Y) - \HH(Y \given X)  , \tag{HB7}
  \end{equation}

\smallskip

\end{minipage}
}

\noindent\fbox{
\begin{minipage}{0.96\textwidth}

\bigskip

\textbf{Hu correspondence} A useful way to memorise the relationship between these concepts
is the Hu correspondence.  In the following picture, the area of each rectangle corresponds
to the quantity.

	\begin{center}
		\begin{picture}(250,100)(0,0)
		  
		  \put(100,0){\framebox(65,15){$\I(X:Y)$}}
		  \put(0,20){\framebox(98,15){$\HH(X\given Y)$}}
		  \put(167,20){\framebox(83,15){$\HH(Y\given X)$}}
		  \put(0,40){\framebox(250,15){$\HH(X,Y)$}}
		  \put(100,60){\framebox(150,15){$\HH(Y)$}}
		  \put(0,80){\framebox(165,15){$\HH(X)$}}
		  
%
%
%
		  
    \end{picture}
\end{center}


\hbtitle{The complex Gaussian distribution}
A \defn{circularly-symmetric complex Gaussian random variable} $Z\sim\CN(0,\sigma^2)$ with variance $\sigma^2 = \Ex |Z|^2$
is defined by the probability density function
  \begin{equation}
    p(z) = \frac{1}{\pi\sigma^2} \ex^{-|z|^2/\sigma^2} \qquad z \in \CC . \tag{HB8}
  \end{equation}

\hbnew

\hbtitle{Maximum entropy}
Out of all discrete random variables $X$ on the set $\{\zero, \one, \dots, \q-\one \}$, the maximum entropy
is achieved when $X$ is uniform \cite[Theorem 2.6.4]{CoverThomas}, giving
  \begin{equation}
    \max_X \HH(X) = \HH \big( U (\{\zero, \one, \dots, \q-1 \}) \big) = \log q . \tag{HB9}
  \end{equation}

\quad Of all continuous random variables $Z$ on $\CC$ with power $\Ex |Z|^2$ at most $\sigma^2$, the maximum entropy
is achieved when $Z \sim \CN(0,\sigma^2)$ \cite[Example 12.2.1]{CoverThomas}, giving
  \begin{equation}
    \max_{Z:\Ex |Z|^2 \leq \sigma^2} \HH(Z) = \HH \big(\CN(0,\sigma^2)\big) = \log (\pi \ex \sigma^2) . \tag{HB10}
  \end{equation}    

\hbnew

\hbtitle{Typical set}
Let $X$ be a random variable on a countable set $\mathcal X$. Given a sequence $\vec X =(X[1], X[2], \dots, X[T]) \in \mathcal X^T$ of $T$ random draws from $X$, then $\vec X$ will very likely take one of only $2^{T\HH(X)} \leq |\mathcal X|^T$ different values, and each of these values is almost equally likely.  These values make up the so-called \defn{typical set}, and this property is called the \defn{asymptotic equipartition property}.

\quad We say $(\vec X, \vec Y)$ is \defn{jointly typical} of $(X, Y)$ if $\vec X$ is in the
typical set of $X$, $\vec Y$ is in the typical set of $Y$, and the pair $(\vec X, \vec Y)$
is in the typical set of the pair $(X,Y)$.

\bigskip

\end{minipage}
}

\newpage

\section{Channels, codes, and capacity}

We will now set out mathematically the ideas of channels, codes, and capacity from Section \ref{vsi}.

To specify a channel, we need to say what inputs are allowed into the channel, what outputs can be produced, and how
the noise randomly corrupts the input.

A common example of a channel is the \defn{binary symmetric channel}.  The BSC allows `bits' -- binary digits: $\zero$s and $\one$s -- into the channel.  It then either outputs the same bit or, with some fixed probability, flips to the other bit.

\addcontentsline{lof}{figure}{\numberline{\protect\textbf{Definition 1.1}} Channel}
\begin{definition}
  A communication \defn{channel} consists of
  \begin{enumerate}
    \item a set $\mathcal X$, the \defn{input alphabet};
    \item a set $\mathcal Y$, the \defn{output alphabet};
    \item a \defn{probability transition function}
      $p(y\given x)$ relating the two.
  \end{enumerate}
\end{definition}


%
		  
\begin{center}
		\begin{picture}(335,40)(0,0)

		  \put(90,0){\framebox(55,30)}
		  \put(106,18){Input}
		  \put(93,3){alphabet $\mathcal X$}
		 
		  \put(145,15){\vector(1,0){45}}
		  \put(151,20){$p(y\given x)$}
		  \put(115,33){$x$}
		  
		  \put(190,0){\framebox(55,30)}
		  \put(201,18){Output}
		  \put(193,3){alphabet $\mathcal Y$}
		  \put(215,34){$y$}
	  
    \end{picture}
\end{center}

\addcontentsline{lof}{figure}{\numberline{\protect\textbf{Definition 1.2}} Binary symmetric channel}
\begin{definition}
We can now formally define the \defn{binary symmetric channel} with error probability $p < \frac12$.  This channel is defined by alphabets $\mathcal X = \mathcal Y = \FF$
and transition function
  \begin{align*}
    p(\zero\given\zero) &= 1-p  &  p(\one\given\zero) &= p    \\
    p(\zero\given\one ) &= p    &  p(\one\given\one ) &= 1-p  . 
  \end{align*}

\end{definition}

	\begin{center}
		\begin{picture}(150,130)(0,0)
      \put(20,100){\vector(1,0){115}}
      \put(20,100){\vector(3,-2){115}}
      \put(20,20){\vector(1,0){115}}
      \put(20,20){\vector(3,2){115}}
      
      \put(10,97){$\zero$}
      \put(10,17){$\one$}
      \put(140,97){$\zero$}
      \put(140,17){$\one$}
      
      \put(-5,55){$x$}
      \put(150,55){$y$}
      
      \put(70,105){$1-p$}
      \put(70,10){$1-p$}
      \put(85,70){$p$}
      \put(85,45){$p$}
    \end{picture}
  \end{center}

A channel which is an accurate and widely-used model for wireless communication \cite[Chapter 5.1]{TseViswanath} is the \defn{Gaussian channel}.  It arises from the sampling of a bandlimited continuous-time channel with white noise \cite{ShannonNoise}.  (For convenience, we assume a unit bandwidth.)  Gaussian white noise is used for two reasons.  First, it seems a reasonable model: the superposition of lots of small pieces of noise ought to (due to the central limit theorem) look roughly Gaussian. Second, it is the simplest case mathematically, and often leads to analytically tractable solutions.

The Gaussian channel takes any number as an input and corrupts it by adding random Gaussian noise.  By convention, and for useful modelling reasons, the channel is usually defined in terms of complex numbers \cite[Subsection 2.2.4]{TseViswanath}.

\addcontentsline{lof}{figure}{\numberline{\protect\textbf{Definition 1.3}} Gaussian channel}
\begin{definition}
The \defn{Gaussian channel} with noise power $\sigma^2$ has alphabets $\mathcal X = \mathcal Y = \CC$, and probability transition density $p(y\given x)$ 
defined implicitly by the relationship $Y = x + Z$, where the $Z \sim \CC \text{N} (0,\sigma^2)$ are all independent.

(If the channel is used multiple times, we assume a new random $Z$ is drawn each time.)
\end{definition}

	\begin{center}
       \includegraphics[scale=0.78]{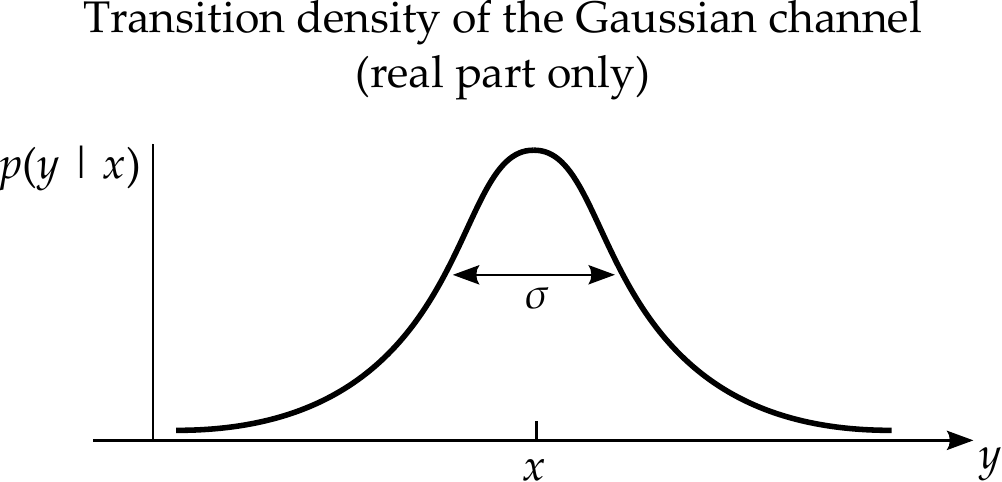}
	\end{center}

Another channel we will examine in this thesis is the finite field channel.  This channel can be a useful model of a Gaussian channel that has been quantised (or discretised).

\addcontentsline{lof}{figure}{\numberline{\protect\textbf{Definition 1.4}} Finite field channel}
\begin{definition}
Let $q$ be prime, and let $Z$ be a random variable defined on the finite field $\F_q = \{\zero, \one, \dots, \q - \one \}$.  The \defn{finite field channel} of size $q$ with noise $Z$ has alphabets $\mathcal X = \mathcal Y = \F_q $, and probability transition density $p(y\given x)$  defined implicitly by the relationship $Y = x + Z \pmod{q}$.
(Again, multiple uses assume independent $Z$s.)
\end{definition}

Note that the BSC is a special case of the finite field channel where $q=2$ and $Z = \one$ with probability $p$ or $Z = \zero$ otherwise.

Now that we have a channel, we can design codes for that channel.  A code takes a set of $M$ messages, and encodes each message into a string -- called a \defn{codeword} -- of length $T$.  After the channel has been used $T$ times to send the codeword, there must be a rule for decoding the received string, to estimate which message was sent.

\addcontentsline{lof}{figure}{\numberline{\protect\textbf{Definition 1.5}} Code}
\begin{definition}
  An \defn{$(M,T)$-code} for the channel $(\mathcal X, \mathcal Y, p(y \given x))$ consists of:
  \begin{enumerate}
    \item a \defn{message set} $\mathcal M$ of cardinality $M$, 
    \item an \defn{encoding function} $\vec x\colon \mathcal M\to\mathcal X^T$,
    \item a \defn{decoding function} $\hat m \colon\mathcal Y^T\to\mathcal M$.
  \end{enumerate}
  
  The set of codewords $\{ \vec x(m) : m \in \mathcal M \}$ is called the \defn{codebook};
  the parameter $T$ is called the \defn{block length}.
\end{definition}

The problem of channel coding works like this:
  \begin{itemize}
    \item The transmitter (Alice) requires to send a message $m \in \mathcal M$.
    \item Alice encodes this message into a codeword $\vec x(m) \in \mathcal X^T$.
    \item Alice sends the first letter  $x(m)[1]$ of the codeword  through the channel to the
      receiver (Bob).  Bob receives a corrupted version $y[1]\in\mathcal Y$ of the letter, where the
      corruption has occurred at random according to $p(y \given x)$.
    \item Alice sends the second letter $x(m)[2]$ of the codeword through the channel to Bob.
      Bob receives a corrupted version $y[2]\in\mathcal Y$ of the letter, where the
      corruption has occurred at random according to $p(y \given x)$.
    \item \mbox{} $\vdots$
    \item Alice sends the final letter $x(m)[T]$ of the codeword through the channel to
      Bob.  Bob receives a corrupted version $y[T]\in\mathcal Y$ of the letter, where the
      corruption has occurred at random according to $p(y \given x)$.
    \item Bob now decodes his received word $\vec y$ to make his estimate $\hat m(\vec y)$
      of Alice's original message $m$.  Hopefully, $\hat m = m$, and the message has
      been communicated successfully.
  \end{itemize}

So the system as a whole looks like this (the channel is in grey):
	\begin{center}
		\begin{picture}(335,40)(0,0)

		  \put(0,0){\framebox(55,30)}
		  \put(9,18){Message}
		  \put(14,3){set $\mathcal M$}
		  \put(24,33){$m$}
		  
   	  \put(55,15){\vector(1,0){35}}
		  \put(70,18){$\vec x$}
		  
		  \color{grey}
		  \put(90,0){\framebox(55,30)}
		  \put(106,18){\textcolor{grey}{Input}}
		  \put(93,3){\textcolor{grey}{alphabet $\mathcal X$}}
		  
		  \put(145,15){\vector(1,0){45}}
		  \put(151,20){\textcolor{grey}{$p(y\given x)$}}
		  \color{black}
		  \put(108,33){$\vec x(m)$}
		  \put(151,4){$T$ times}
		  
		  \color{grey}
		  \put(190,0){\framebox(55,30)}
		  \put(201,18){\textcolor{grey}{Output}}
		  \put(193,3){\textcolor{grey}{alphabet $\mathcal Y$}}
		  \put(214,34){\textcolor{grey}{$\vec y$}}
		  
		  \color{black}
		  \put(245,15){\vector(1,0){35}}
		  \put(259,18){$\hat m$}
		  
		  \put(280,0){\framebox(55,30)}
		  \put(289,18){Message}
		  \put(294,3){set $\mathcal M$}
		  \put(300,33){$\hat m(\vec y)$}
		  
    \end{picture}
\end{center}

All channels considered in this thesis will be \defn{memoryless}, in that the channel's current performance is independent of earlier behaviour.  In other words, the transmission of codewords follows a product distribution
     \[ p(\vec y \given \vec x) = \prod_{t=1}^T p(y[t] \given x[t]) . \]
     
So far, we have considered only \defn{static} channels, where the probability transition function remains fixed over time.  (Later, we will look at fast-fading channels, where the transition function is no longer fixed but changes from timeslot to timeslot.)

\bigskip

\noindent Note that, from a mathematical point of view, what the messages \emph{are} is unimportant -- what matters is how  many of them there are. So it often makes sense to choose the message set $\mathcal M$ to be something convenient. For example, when dealing with a finite field channel of size $q$, we often take $\mathcal M$ to be $\mathbb F_q^S$, which has cardinality $M = q^S$. In particular, when $q=2$, the message set $\mathcal M = \mathbb F_2^S = \zo^S$ is the set of all bit strings of length $S = \log_2 M$.

We define the \defn{rate} of a code to be the number of bits that we can send per channel use. The number of bits is $\log_2 M$, as above, and the number of channel uses is $T$, so the rate is $(\log_2 M)/T$.

\addcontentsline{lof}{figure}{\numberline{\protect\textbf{Definition 1.6}} Rate}
\begin{definition}
  The \defn{rate} of an $(M,T)$-code is defined to be $(\log_2 M)/T$ bits per transmission.
\end{definition}

(From now on, all logarithms are to base $2$, and we just write $\log$ for $\log_2$.)

Also associated with a code we have its \defn{error probability}, the chance that a message is decoded incorrectly.  (We take the average error probability across all messages, but if we were to use the maximum error probability, the main results of this chapter would be the same.)

\addcontentsline{lof}{figure}{\numberline{\protect\textbf{Definition 1.7}} Error probability}
\begin{definition}
  The \defn{average error probability} is
  \begin{align*}
    e &:= \frac1M \sum_{m\in\mathcal M} \Prob{\big(\hat m(\vec Y) \neq m\given \vec x(m) \text{ sent}\big)} \\
       &\phantom{:}= \frac1M \sum_{m\in\mathcal M} \sum_{\vec y \in \mathcal Y^T} p\big(\vec y \given \vec x(m)\big)
                                          \, \mathbb{1}[\hat m(\vec y) \neq m] .
  \end{align*}
\end{definition}

We can now give an example of a code for the BSC.

Suppose there are two messages we might wish to send:
``No'' and ``Yes'', so $\mathcal M = \{ \text{No}, \text{Yes} \}$.

A very simple code could assign $\zero$ to be the codeword for ``No'' and $\one$ to be
the codeword for ``Yes''.  This is a $(2,1)$-code.  It clearly has a rate of $(\log 2)/1 = 1$ bit per transmission and error probability $p$.

To reduce the error probability, we will need a more sophisticated code.  One method of coding would be the repetition code where each symbol is repeated $T$ times, so
  \[ \vec x(\text{No}) = \zero\zero\cdots\zero \in \mathcal X^T \qquad
     \vec x(\text{Yes} ) = \one \one \cdots\one  \in \mathcal X^T . \]

\addcontentsline{lof}{figure}{\numberline{\protect\textbf{Definition 1.8}} Repetition code}     
\begin{definition}
  A $(2,T)$-code is called a \defn{$T$-repetition} code if the codebook consists
  solely of the all-$\zero$ and all-$\one$ codewords of length $T$.
\end{definition}
     
The obvious method of decoding is the `majority rule': if $\vec y$ has more $\zero$s than
$\one$s, decode as ``No''; if $\vec y$ has more $\one$s than $\zero$s, decode as ``Yes''.
(If there are exactly $T/2$ of each, decoding may be performed arbitrarily.)

Note that the rate of this code is $(\log 2)/T = 1/T$, and the error probability is bounded by
$e \geq p^T$, the probability that all $T$ symbols flip.  
%

\bigskip

\noindent What does it mean to be able to communicate through a channel at some desired rate $r$?
Well, it means that there must exist a code for the channel with rate at least $r$ and
a low probability of error.  How low?  As low as we desire.  If we want to limit the
error probability to 5\%, then there must be a code with rate at least $r$ and error
probability no more than 5\%; but if we want the error probability to be as low as 1\%
or even 0.01\%, there has to be a code for that too.

\addcontentsline{lof}{figure}{\numberline{\protect\textbf{Definition 1.9}} Achievable rate}  
\begin{definition}
  Consider a channel $(\mathcal X, \mathcal Y, p(y \given x))$.
  
  A rate $r$ is
  \defn{achievable} if for any error tolerance $\epsilon > 0$, there exists a
  code with rate at least $r$ and error probability lower than $\epsilon$.
  
  Otherwise, $r$ is \defn{not achievable}, in that there exists an error threshold
  $\epsilon$ such that there exists no code with rate at least $r$ and error
  probability lower than $\epsilon$.
\end{definition}

The capacity is defined to be the maximum achievable rate.

\addcontentsline{lof}{figure}{\numberline{\protect\textbf{Definition 1.10}} Capacity}  
\begin{definition}
  Consider a channel $(\mathcal X, \mathcal Y, p(y \given x)) $.  Then we define
  the \defn{capacity} of the channel, $c$, to be the supremum of all achievable
  rates:
    \[ c := \sup \{r : r \text{ is achievable} \}. \]
    
  In other words, all rates $r$ less than $c$ are achievable,
  but no $r$ above $c$ is achievable.
\end{definition}

Shannon calculated the capacity as the maximum mutual information between the input and output of a channel \cite[Theorem 11]{Shannon}.

The mutual information $\I (X:Y)$ between $X$ and $Y$ can be seen as a measure
of `how independent' $X$ and $Y$ are.  If $\I(X:Y)$ is large, then $X$ and $Y$ are highly
dependent, so knowledge of the output $Y$ gives us lots of information about the input
$X$; if $\I(X:Y)$ is small, then $X$ and $Y$ are highly
independent, so knowledge of the output $Y$ gives us little information about the input
$X$.

\addcontentsline{lof}{figure}{\numberline{\protect\textbf{Theorem 1.11}} Shannon's channel coding theorem}  
\begin{theorem}[Shannon's channel coding theorem]\label{Shannons2}
  Consider a discrete channel $(\mathcal X, \mathcal Y, p(y \given x))$, that is, a
  channel where both $\mathcal X$ and $\mathcal Y$ are both countable sets.
  
  Then the capacity $c$ is given
  by the formula $c = \max_X \I(X:Y)$, where $Y$ is related to $X$ through $p(y \given x)$,
  and the maximum is over all input random variables $X$ defined on $\mathcal X$.
\end{theorem}

The characterisation of capacity given by Shannon's channel coding theorem (Theorem \ref{Shannons2}) allows us to calculate the capacity of the finite field channel and the BSC.

\addcontentsline{lof}{figure}{\numberline{\protect\textbf{Theorem 1.12}} Capacities of finite field channel and BSC} 
\begin{theorem}
  The capacity of the finite field channel of size $q$ with noise $Z$ is
  \[ c = \log q - \HH(Z) =: \D(Z) . \]
  
  The capacity of the BSC with error probability $p$ is
  \[ c = 1 - \left( p \log \frac1p + (1-p) \log \frac{1}{1-p} \right) . \]
\end{theorem}

	\begin{center}
       \includegraphics[scale=0.74]{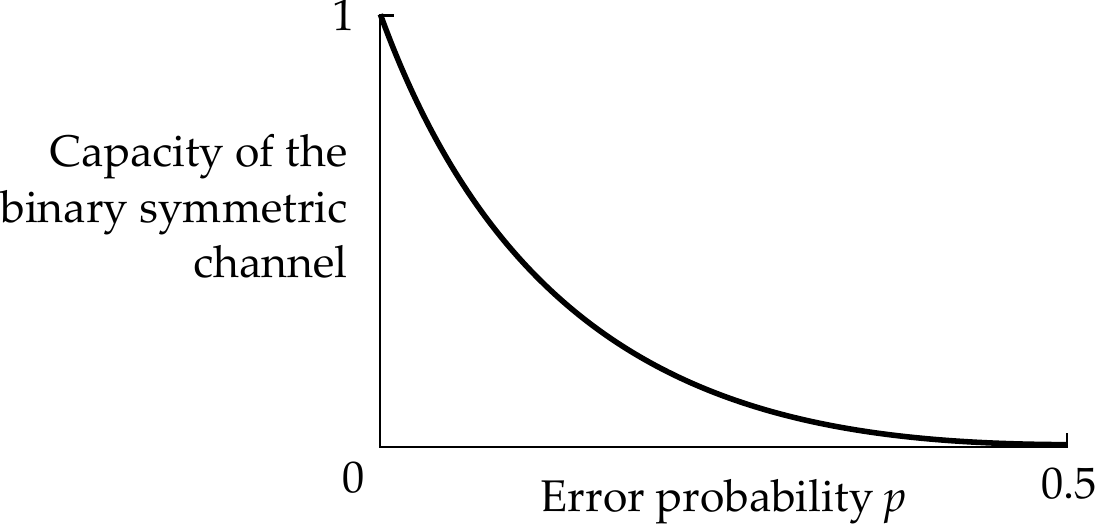}
	\end{center}

(We use the abbreviation $\D(Z) := \log q - \HH(Z)$ since this is equal to the relative
entropy distance $\D(p(z) \,\|\, p(u))$ between $Z$ and a uniform random variable $U$ over $\F_q$.)

\begin{proof}
  \emph{Finite field channel.} By Shannon's channel coding theorem (Theorem \ref{Shannons2})
  we need to calculate the mutual information.  This is
    \begin{align}
      \I(X:Y) &= \HH(Y) - \HH(Y\given X)   \tag{HB7} \\
              &= \HH(Y) - \HH(X+Z\given X) \tag{$Y = X + Z$, Definition 1.3}                \\
              &= \HH(Y) - \HH(Z)  .        \tag{HB4}
    \end{align}
  This is maximised by choosing $X$ to be uniform on $\F_q$, by (HB9), so that $Y$ is uniform also, giving
    \[ c = \max_X \I(X:Y) =\log q - \HH(Z) = \D(Z) , \]
  as required.
  
  \emph{BSC.} This result follows from recalling that the BSC is a special case of the finite field model.
\end{proof}

Later in this thesis, we will often see examples of channels and networks whose
capacities are a constant fraction of the finite field channel capacity.
If a channel or network has capacity $c = d \, \D(Z)$ for some constant $d$, we
say that the channel has $d$ \defn{degrees of freedom} (also known as the
\defn{multiplexing gain} or \defn{pre-$\log$ term}).

\addcontentsline{lof}{figure}{\numberline{\protect\textbf{Definition 1.13}} Degrees of freedom (finite field case)} 
\begin{definition} \label{def:ffdof}
  Given a discrete channel or network with capacity $c$, we define
  the \defn{degrees of freedom} to be
    $ \dof = c / \D(Z)$ . 
\end{definition}

Clearly, the finite field channel itself has a single degree of freedom, that is we have
$\dof = 1$.

\bigskip

\noindent We have not yet talked about the proof of Shannon's channel coding
theorem (Theorem \ref{Shannons2}).

To prove Shannon's channel coding theorem (and related theorems) we
must prove two things:
  \begin{description}
    \item[Achievability] First, we must show that any rate below
      capacity $r<c$ is achievable.  That is, we must find a sequence
      of codes all with rates at least $r<c$, but with arbitrarily
      low error probabilities.
    \item[Converse] Second, we must show that any rate above capacity
      $r>c$ is not achievable.  That is, we must show that for any
      sequence of codes
      all with rates at least $r>c$, the error probabilities must be
      bounded away from $0$.
  \end{description}

The converse part is proved using Fano's inequality \cite{Fano}, which
bounds the error probability in terms of the conditional entropy
across the channel $\HH(Y \given X)$ and the size of the message set $M$.

Shannon's key insight into the achievable part is the following:
instead of trying carefully to design special codes with high rates
and low error probabilities, we can instead just pick the code at
random.  That is, we choose the codeword letters $X(m)[t]$ IID
according to some distribution $X$.  If we set $M = \big\lceil 2^{Tr} \big\rceil$,
then the rate of the code will be at least $r$.  We hope that by choosing
$T$ sufficiently large, the error probability will be driven arbitrarily
low.
(Later, we can optimise over the
choice of $X$.)

This random encoder can be twinned with an effective decoder to show
that any rate $r<c$ can be achieved.  Two different decoders can be
used:
  \begin{description}
    \item[Joint typicality decoder] The receiver takes the channel
      output $\vec y$ and finds the unique codeword $\vec x(m)$ such
      that the pair $(\vec x(m), \vec y)$ is jointly typical of
      $(X,Y)$. (See the handbook, Section 1.2, for definitions.)
      This $m$ is the decoding estimate.  (If there is no
      such $\vec x(m)$ or it isn't unique, we declare an error.)
    \item[Maximum likelihood decoder] The receiver takes the channel
      output $\vec y$ and decodes to the message most likely to
      have yielded it.  That is, we pick $\hat m$ to maximise
        \[ p\big(\vec y \given \vec x(m) \big)
             = \prod_{t=1}^T p \big( y[t] \given x(m)[t] \big) \]
      (If the maximum isn't unique, we declare an error.)
  \end{description}

Shannon himself \cite[Section 13]{Shannon} and most subsequent authors
(for example \cite[Chapter 7]{CoverThomas}, \cite[Chapter 10]{MacKay}) use the joint typicality approach, as it gives a fairly simple
and short proof.

On the other hand, Gallager \cite{Gallager} used the maximum likelihood
approach to prove Shannon's channel coding theorem, by bounding the
error probability by $e \leq 2^{-TE(r)}$, where $E(r)$ is called the
\defn{error exponent}, and examining the error exponent for different
values of $r$.  We shall return to Gallager's maximum likelihood
approach in Chapter 6, when proving a similar theorem for group
testing.

We now outline the achievability proof using the joint typicality
decoder.  Basic facts about typical sets are given in the handbook (Section 1.2).

\begin{proof}[Sketch proof of achievability, Theorem \ref{Shannons2}]
  As above, we set the number of messages to be $M = \big\lceil 2^{Tr} \big\rceil$,
  choose codeword letters $X(m)[t]$ IID at
random according to some distribution $X$, and decode using a joint
typicality decoder.  There are two ways we could get an error.

  First, the actual codeword $\vec X$ and the received output $\vec Y$
could fail to be jointly typical.  But the theory of typical sets tell
us that this event is very unlikely.

 Second, another codeword $\vec X(\hat{m})$ could be jointly typical
with $\vec Y$, despite $\vec X(\hat{m})$ and $\vec Y$ actually being
independent of each other.  Since $\vec X(\hat{m})$ and $\vec Y$ are
very likely to (marginally) typical, joint typicality occurs with
approximate probability
  \begin{align*}
    \frac{ \# \text{ jointly typical }(\vec x, \vec y) }
          { \# \text{ typical }\vec x \times
            \# \text{ typical }\vec y }
       &\approx \frac{2^{T \HH(X,Y)}}{2^{T \HH(X)} 2^{T \HH(Y)}} \\
       &= 2^{-T(\HH(X) + \HH(Y) - \HH(X,Y))} \\
       &= 2^{-T \I(X:Y)} ,
  \end{align*}
by standard facts about typical sets (see the handbook, Section 1.2).
Hence the probability of error is approximately
  \begin{align*}
    e &\leq \sum_{\hat m \neq m} \Pr ( \text{$\vec X(\hat m)$ and $\vec Y$
                                       jointly typical}) \\
      &\approx \sum_{\hat m \neq m} 2^{-T \I(X:Y)} \\
      &= (M-1) 2^{-T \I(X:Y)} \\
      &= \big(\big\lceil 2^{Tr} \big\rceil - 1\big) 2^{-T \I(X:Y)} \\
      &\leq 2^{Tr} 2^{-T \I(X:Y)} \\
      &= 2^{-T (\I(X:Y) - r)} .
  \end{align*}
So provided $r < \I(X:Y)$, then by choosing $T$ large enough, the
error probability can be made arbitrarily small.

Choose $X$ to maximise $\I(X:Y)$ to get the result.
\end{proof}

In a sense, this theorem is quite a `lucky'
result: it turns out that the lower bound on
capacity given by Shannon's random coding argument and the upper bound given by
Fano's inequality coincide, to give us the equality $c = \max_X \I(X:Y)$.

For other networks, we may not be so lucky.  However, similar proof strategies can be useful.
If we can show that all rates below some $r_*$ are achievable, this gives us a lower bound
on capacity: $c\geq r_*$.  Conversely, if we can show that no rates above some $r^*$
are achievable, then we have an upper bound: $c \leq r^*$.  In the point-to-point case, we have
$r_* = r^* = c$.  But even if there is a gap between the upper and lower bounds,
the result can be useful in giving us an approximation to the capacity.  In particular,
sometimes there may be a limiting sense in which the upper and lower bounds are asymptotically equal --
for example, as signal power or number of users tends to infinity.

\bigskip

\noindent It is often useful to think not just of individual codes, but of \emph{families} of codes. One family of codes we have already seen is the repetition code (Definition 1.8).

MacKay \cite[Section 11.4]{MacKay} divides families of codes into three separate categories, depending on how effective they are for their channel.
\begin{description}
  \item[Bad codes] In bad families of codes, as we force the error probability to $0$,
    the rate of the codes approaches $0$ also.
  \item[Good codes] In good families of codes, as we force the error probability to $0$,
    the rate of the codes is bounded above $0$, but below capacity.
  \item[Very good codes] In very good families of codes, as we force the error probability to $0$,
    the rate of the codes can be maintained arbitrarily close to capacity.
\end{description}

Earlier we saw that rate of the repetition code is $1/T$, and its error probability is bounded by
$e > p^T$. Hence, to force the error probability $e$ to $0$, we must send $T \to \infty$, and the rate tends to $0$. Hence, the repetition code is a bad code. (MacKay notes, however, that bad codes are not necessarily practically useless \cite[p. 183]{MacKay}.)

It is sufficient for this thesis to note that Shannon's channel coding theorem (Theorem \ref{Shannons2}) tells us that very good (capacity achieving) codes do exist, and that we can do much better than simple codes like the repetition code. (The design of good and very good practical codes is outside the scope of this thesis.)

In later work in this thesis, rather than finding new codes from
scratch, we will instead \emph{adapt} these very good point-to-point channel codes
for use in large networks.

\bigskip

\noindent A useful class of codes for finite field channels is the class of \defn{linear codes}. If we take $\mathcal M = \mathbb F_q^S$ for the message set again, then a linear code is a code whose encoding function $\vec x \colon \mathbb F_q^S \to \mathbb F_q^T$ is a linear map.

It's often useful to represent this linear map by an $S \times T$ matrix $\mat G$, so $\vec x(\vec m) = \mat G \vec m$. We call $\mat G$ the \emph{generator matrix} of the code. The rate of such a code is $(\log M)/T = (S/T) \log q$.

\addcontentsline{lof}{figure}{\numberline{\protect\textbf{Definition 1.14}} Linear code} 
\begin{definition}
  Consider a finite field channel of size $q$. Then a \defn{linear code} is a $(q^S,T)$-code
  with message set $\mathcal M = \mathbb F_q^S$ and encoding function
  $\vec x(\vec m) = \mat G \vec m$ for some \defn{generator matrix}
  $\mat G \in\mathbb F_q^{S\times T}$. Any decoding function may be used.
  
  We call $S$ the \defn{rank} of the code.
\end{definition}

For example, the $T$-repetition code is a linear code with field size $q=2$, rank $S=1$ and $1 \times T$ generator matrix $\mat G = (\one\ \one \cdots\ \one)$.

The important fact about linear codes (at least for finite field channels) is that, when paired with an optimal decoder, very good (capacity achieving) linear codes exist \cite[Chapter 14]{MacKay}.  Thus, if we restrict our attention only to linear codes, we can still achieve all rates up to the capacity $c = \mathbb D(Z)$ of the finite field channel.

\addcontentsline{lof}{figure}{\numberline{\protect\textbf{Theorem 1.15}} Very good linear codes exist} 
\begin{theorem}
  Very good linear codes exist for all finite field channels with nonzero capacity.
\end{theorem}

The current state of the art for high-rate practical codes -- that is codes with low encoding and decoding complexity and moderate block lengths -- is a class of random linear codes called low-density parity-check codes \cite[Chapter 47]{MacKay}. (See the textbook of Richardson and Urbanke \cite{RichardsonUrbanke} for more details.)


\section{Power}

We have not yet looked at codes for the Gaussian channel.

The capacity of the Gaussian channel is infinite, as there exist codes with arbitrarily
high rates and simultaneously arbitrarily low error probabilities.

To see this, consider the following.  Let $\mathcal{M} = \{1,2,\dots,M\}$, encode
using $x(m) = mN$ for some very large $N$, and decode to the nearest positive integer to $y/N$
(which should be roughly $m$).  This is an $(M,1)$-code,
with rate $\log M$. By picking $N$ large enough, the error probability can be made arbitrarily
small; but by picking $M$ large enough, the rate can be made arbitrarily high.

This is neither mathematically interesting nor physically realistic.  Antennas for wireless networks are not
capable of transmitting at arbitrarily high powers.  Thus, we introduce a \defn{power constraint}: that for all
codewords $\vec x$, the power -- the mean square value -- is limited by a prescribed value $P$.

\addcontentsline{lof}{figure}{\numberline{\protect\textbf{Definition 1.16}} Power}
\begin{definition}
  The power of a codeword $\vec x$ is defined to be
  \[ \overline{| x|^2} := \frac1T \sum_{t=1}^T |x[t]|^2. \]

The power of a code is defined to be the maximum power of any codeword.
\end{definition}

So we want to limit our attention to codes whose power is at most
the power constraint $P$.

\addcontentsline{lof}{figure}{\numberline{\protect\textbf{Definition 1.17}} Achievable rate with power constraint}
\begin{definition}
  Consider a channel $(\mathcal X, \mathcal Y, p(y \given x))$.
  
  A rate $r$ is
  \defn{achievable with power $P$} if for any error tolerance $\epsilon > 0$, there exists a
  code of power at most $P$ with rate at least $r$ and error probability lower than $\epsilon$.
  
  Otherwise, $r$ is \defn{not achievable with power $P$}, in that there exists an error threshold
  $\epsilon$ such that there exists no code of power at most $P$ with rate at least $r$ and error
  probability lower than $\epsilon$.
\end{definition}

\addcontentsline{lof}{figure}{\numberline{\protect\textbf{Definition 1.18}} Capacity with power constraint}
\begin{definition}
  Consider a channel $(\mathcal X, \mathcal Y, p(y \given x)) $.  Then we define
  the \defn{capacity of the channel with power $P$} to be the supremum of all achievable
  rates:
    \[ c := \sup \{r : r \text{ is achievable with power } P \}. \]
    
  In other words, all rates $r$ less than $c$ are achievable with power $P$,
  but no $r$ above $c$ is achievable.
\end{definition}

Shannon calculated the capacity of the Gaussian channel with a power constraint in his original paper \cite[Theorem 17]{Shannon}

\addcontentsline{lof}{figure}{\numberline{\protect\textbf{Theorem 1.19}} Capacity of Gaussian channel with power constraint}
\begin{theorem}
  Consider the Gaussian channel with power constraint $P$ and noise power $\sigma^2$.
  
  Then the capacity is given
  by the formula $c = \log{(1+\snr)}$, where we have defined the \defn{signal-to-noise ratio}
  $\snr := P/\sigma^2$ to be the ratio of the signal power to the noise power
\end{theorem}

	\begin{center}
       \includegraphics[scale=0.75]{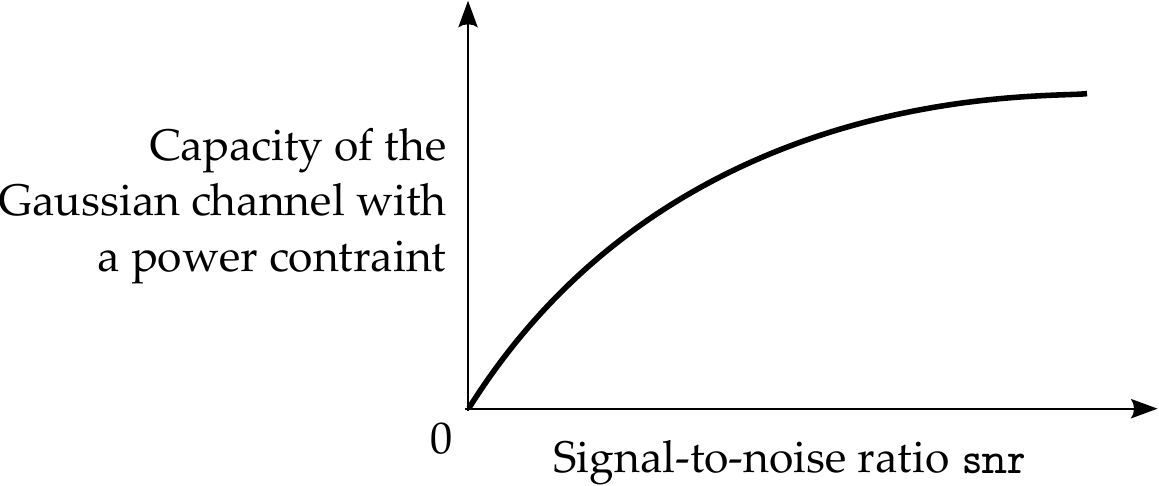}
	\end{center}

Two useful approximations for the capacity of the Gaussian channel are
  \begin{equation} \label{eq:gapprox}
    c = \log(1+\snr) \approx \begin{cases} \snr \log \ex & \text{for small $\snr$,} \\
                                           \log\snr      & \text{for large $\snr$.} \end{cases} 
  \end{equation}
So at low $\snr$, capacity grows linearly with $\snr$; whereas at high $\snr$, capacity
only grows logarithmically.

Later in this thesis, we will often see examples of channels and networks whose
capacities at high-$\snr$ are often a constant fraction of the Gaussian channel capacity $c \approx \log \snr$.
If a channel or network has capacity $c = d \log \snr + o(\log \snr)$ at $\snr\to\infty$ for some constant $d$, we
say that the channel has $d$ \defn{degrees of freedom}.  (This is the Gaussian analogy
to the finite field degrees of freedom in Defintion \ref{def:ffdof}.)

\addcontentsline{lof}{figure}{\numberline{\protect\textbf{Definition 1.20}} Degrees of freedom (Gaussian case)}
\begin{definition} \label{def:gaussiandof}
  Given a channel with capacity $c(\snr)$ and power constraint $P = \snr\, \sigma^2$,
  we define the \defn{degrees of freedom} to be
    \[ \dof = \lim_{\snr \to \infty} \frac{c(\snr)}{\log \snr}  \]
  where this limit exists.
\end{definition}

From \eqref{eq:gapprox} the Gaussian channel itself has a single degree of freedom,
that is, we have $\dof = 1$.
                                            
\begin{proof}[Sketch proof of Theorem 1.19]
  In a similar manner to Shannon's channel coding theorem (Theorem \ref{Shannons2}),
  it can be shown that the capacity of the Gaussian channel with a power constraint
  is $\max_{X : \Ex |X|^2 \leq P} \I(X:Y)$.  (This is certainly a believable result:
  it is the same formula as for discrete channels, with the additional constraint
  that the expected power satisfies the constraint.)
  
  It remains to calculate the mutual information.  This is
    \begin{align}
      \I(X:Y) &= \HH(Y) - \HH(Y\given X)   \tag{HB7} \\
              &= \HH(Y) - \HH(X+Z\given X) \tag{$Y = X + Z$, Definition 1.3}                \\
              &= \HH(Y) - \HH(Z)           \tag{HB5}\\
              &= \HH(Y) - \log(\pi \ee \sigma^2). \tag{$Z \sim \CN (0, \sigma^2)$, Definition 1.3} 
    \end{align}
  Note that 
    \[ \Ex |Y|^2 = \Ex |X+Z|^2 = \Ex|X|^2 + \Ex |Z|^2 \leq P + \sigma^2 . \]
  Hence the entropy of $Y$  is maximised by choosing $X$ to be complex Gaussian with variance $P$,
  by (HB10), so that $Y$ is Gaussian also, with power $P + \sigma^2$, giving
    \begin{align*}
      c := \max_X \I(X:Y) 
        &= \log\big((\pi \ee (P+\sigma^2)\big) - \log(\pi \ee \sigma^2) \\
        &= \log\left( \frac{\pi \ee (P+\sigma^2)}{\pi \ee \sigma^2} \right) \\
        &= \log\left( 1 + \frac{P}{\sigma^2} \right) ,
    \end{align*}
  as required.
\end{proof}

Note that the input distribution that achieves capacity is $X \sim \CN(0,P)$.  So
the signal is statistically the same -- that is, distributed in the same
parametric family -- as noise, but with a different power.  This fact will be useful later.

As we mentioned earlier, Theorem 1.19 tells us that very good (capacity achieving)
codes exist for the Gaussian channel.  Later, we will adapt these very good point-to-point codes
for use in large Gaussian networks.

\section{Fading}

In their standard forms, the finite field and Gaussian channels are represented
by the formula $Y[t] = x[t] + Z[t]$.  We can interpret this as the signal being
transmitted perfectly through the channel, except for the addition of some noise.
However, for a more realistic model of wireless networks, we need to account for
the way the signal itself transforms as it is sent through the channel.  For example,
in the Gaussian channel, we might expect the signal power to decay over long
distances, and standard physical models  suggest that the phase of the signal
$\arg x[t]$ will alter as it is transmitted through space \cite[Section 2.1]{TseViswanath}.

We can model these concepts by introducing a \defn{fading} (or
\defn{channel state}) \defn{coefficient} $H[t]$. Our channels now become
$Y[t] = H[t]x[t] + Z[t]$.

We are interested in three cases:
  \begin{description}
    \item[Fixed fading] where $H[t] = h$ is a fixed deterministic constant
      (Subsection 1.4.1);
    \item[Fast fading] where the $H[t]$ are IID random
      (Subsection 1.4.2);
    \item[Slow fading] where $H[t]=H$ is random, but fixed for all time
      (Subsection 1.4.3).
  \end{description}

\bigskip

\noindent Before we continue, we have a useful simplification to make.  Since the fading Gaussian channel
  \[ Y[t] = H[t]x[t] + Z[t] \qquad Z[t] \sim \CN (0,\sigma^2) \qquad \overline{|x|^2} \leq P  \]
will be used a lot in this thesis, it makes sense to change our units, so that the noise power
and power constraint are both unity.
To that end, set
  \[ \tilde Y[t] := \frac{1}{\sigma} Y[t]  \qquad
     \tilde H[t] := \sqrt{\frac{P}{\sigma^2}} H[t]  \qquad 
     \tilde x[t] := \frac{1}{\sqrt P} x[t]  \qquad
     \tilde Z[t] := \frac{1}{\sigma} Z[t] . \]
Under this change of units we have, after dividing through by $\sigma$,
  \begin{equation}\label{units}
    \tilde Y[t] = \tilde H[t]\tilde x[t] + \tilde Z[t] \qquad \tilde Z[t] \sim \CN (0,1)
    \qquad \overline{|\tilde x|^2} \leq 1 .
  \end{equation}
From now on, we will solely use this model, so will shall drop the tildes.
Note that under this change of units,
the signal-to-noise ratio
  \[ \frac{P |H|^2}{\sigma^2} = \frac{1 |\tilde H|^2}{1} \]
remains the same.
(Note also that under this change, the Gaussian channel with no fading inherits a
fixed fading coefficient $H[t] = h = \sqrt{P/\sigma^2}$.)

\subsection{Fixed fading}

Fixed fading models fading that is constant and predictable, such
    as the decay in signal power between a non-moving transmitter and a non-moving receiver within a fixed
    environment.
    
    We model the fading coefficient as a deterministic constant fixed for all
    time, $H[t] = h$ for all $t$, giving $Y[t] = hx[t] + Z[t]$.
    
    How does the capacity alter now?
    
    For the finite field channel we have $h \in \F_q$.  Note that for nonzero
    $h$, the function $x \mapsto hx$ is a bijection, so the channel is
    equivalent to that without fading ($h=\mathtt{1}$), and still has capacity
    $\D(Z)$, from Theorem 1.12.  On the other hand, if $h = \zero$, then $Y$ is always $\zero$,
    and all signals are indistinguishable, so the capacity is $0$.  Hence,
    for the finite field channel, the capacity is
      \[ c = c(h) = \begin{cases} 0     & \text{if $h=\zero$} \\
                                  \D(Z) & \text{otherwise.}  \end{cases} \]
		(For simplicity, we often just assume that $h$ is nonzero, so the capacity
		is unchanged as $c = \D(Z)$.)

    For the Gaussian channel with power constraint $P=1$ we have $h \in \C$.  The
    power constraint is now $|hx|^2 = |h|^2 |x|^2 \leq |h|^2$.
    So this channel is equivalent to one with no fading, but with the power
    constraint changed from $1$ to $|h|^2$.  The capacity is thus
      \begin{equation*} \label{ffgauss}
        c = c(h) = \log (1 + |h|^2 ) = \log(1+\snr) ,
      \end{equation*}
    with the new convention that $\snr$ denotes the signal-to-noise ratio
    at the receiver: $\snr = |h|^2 P /\sigma^2 = |h|^2$.
    
    To summarize:
    \addcontentsline{lof}{figure}{\numberline{\protect\textbf{Theorem 1.21}} Capacity with fixed fading}
      \begin{theorem}
        For $h\neq\zero$, the capacity of the finite field channel with fixed
        fading is $c = \D(Z)$.
        
        The capacity of the Gaussian channel with fixed fading is
        $c = \log(1 + \snr)$, where $\snr = |h|^2$.
      \end{theorem}
    
    For modelling wireless networks, we will often use a Gaussian channel with
    fixed fading coefficient $h$ that decays like a power law
    over distance.  That is, we have $h = k \rho^{-\alpha/2}$, where $\rho > 0$ is the distance
    between a transmitter and receiver and $k>0$ is a constant.  
    
    The parameter $\alpha>0$ -- called
    the \defn{attenuation} -- represents how resistive the environment is to
    the transmission of radio waves.  (Some authors call $\alpha/2$ the attenuation;
    we do not.)  Low $\alpha$ represents an environment with few
    obstacles to signals; high $\alpha$ implies that a lot of the signal power is
    absorbed before reaching the receiver.  In free space, standard physical
    considerations imply that $\alpha = 2$ or $3$; for built-up areas,
    values of $\alpha$ of roughly $4$ or $5$ seem more appropriate \cite[Section 2.1]{TseViswanath}.
    The capacity of such a channel is $c = \log (1+k^2 \rho^{-\alpha})$, by Theorem 1.21.
    
    (Power-law attenuation fails to be
    realistic for small distances $\rho \ll 1$.  Here the received power would be greater
    than the transmitted power, which violates the conservation of energy.
    Some authors therefore prefer alternative models such as $h = \min\{1,k\rho^{-\alpha/2}\}$
    or $h = k(\rho+\rho_0)^{-\alpha/2}$ for some fixed constant $\rho_0$.)
    
    We could also include a fixed phase change in this model by setting
    $h = k \rho^{-\alpha/2} \ex^{\ii\theta}$.
    In free space, the phase would scale linearly with distance, so
    $h = k \rho^{-\alpha/2} \ex^{2\pi\ii\rho/\lambda}$, where $\lambda$ is the carrier
    wavelength \cite{leveque3}.
    Note that we still have $|h|^2 = k^2 \rho^{-\alpha}$, so the capacity is the
    same.
    
\subsection{Fast fading}

Fast fading models a situation where the state of the channel is changing
    rapidly, such as a commuter using a mobile phone on a train.  We model this as
    $H[t]$ being random according to some distribution $H$ but renewing at each
    channel use; that is, the $H[t]$ are independent and identically distributed
    like $H$.
    
When we deal with fast fading, the performance of a channel will depend on
whether the transmitter and receiver know the current value of $H[t]$, or
just the general distribution $H$, and whether the transmitter can use this
knowledge to vary their power.

Throughout this thesis, we assume that both the transmitter and the receiver
know $H[t]$.  This is known as having \defn{perfect channel state information at the
transmitter} (CSIT) and \defn{at the receiver} (CSIR).  We presume that this
knowledge is \defn{causal}, that is, the receiver and transmitter learn $H[t]$ immediately
prior to the transmission of $x[t]$ and reception of $y[t]$ respectively.  In other
words, they have no prediction of the channel future to use (except, of course, knowing
the future channel states will be IID according to $H$).

When the transmitter has CSIT in a Gaussian channel,
we must specify whether or not she can use this information
to operate at varying power.  So there are two different types of power constraint.  Let
$x[t](h)$ be the $t$th codeword letter chosen in channel state $h$,
and let $\mathcal H$ be the support of $H$.  (Recall from our simplification
\eqref{units} that we now have $P=1$.)
  \begin{description}
    \item[Universal] A universal power constraint demands that the power
      constraint is held universally over each channel state realisation.
      That is, we demand
        \[ \overline{|x[t](h)|^2} = \frac1T\sum_{t=1}^T |x[t](h)|^2 \leq 1 \qquad \text{for all $h \in \mathcal H$.} \]
    \item[Average] An average power demands that the power constraint
      is held when averaged over all channel state realisations.  That is,
      we demand
        \[ \mathbb E_H \overline{|x[t](H)|^2}
             = \frac1T \sum_{h \in \mathcal H} \Prob(H=h) \sum_{t=1}^T |x[t](H)|^2
             \leq 1 . \]
      (The second term assumes $H$ is discrete; the summation can be replaced
      by an integral if $H$ is continuous.)
  \end{description}

An average power constraint allows a transmitter to use extra power when the channel is at
its strongest, and save power when the channel is weak.  (For more details, see the
textbook of Tse and Viswanath \cite[Subsection 5.3.3]{TseViswanath}.)
  
In this thesis, we always assume a universal power constraint.  First, it is mathematically
simpler to deal with.  Second, it is physically unrealistic for transmitters to operate
above their average power for long periods of time.

(Similarly, in a frequency-selective channel, one must specify whether the the power constraint
is enforced in each individual subchannel, or an average across all frequencies.)

So, assuming perfect CSIT and CSIR (with a universal power constraint in the Gaussian case)
we have the following.
\addcontentsline{lof}{figure}{\numberline{\protect\textbf{Theorem 1.22}} Capacity with fast fading}
      \begin{theorem}\label{fastfading}
        Consider a fast fading channel, and
        let $c(h)$ be the capacity of the channel under fixed fading parameter $h$.
        
        Then the fast fading capacity is equal to the average fixed fading capacity, in that $c = \EE c(H)$.
      \end{theorem}
    
    In the most general case, this theorem is due to Goldsmith and Varaiya
    \cite{GoldsmithVaraiya}.  We sketch the achievability proof for the case
    when $H$ is discrete.  (Goldsmith and Varaiya attribute the result for
    this simpler case to Wolfowitz \cite[Theorem 4.6.1]{Wolfowitz}.)

    \begin{proof}[Sketch proof]
    Assume $H$ is discrete, in that it can only take values in some
    countable set $\mathcal H$.  Then if we `collect together' all occasions
    when $H[t]$ has some particular value $h$, we can treat that collection of
    channel uses as being through a fixed fading channel with deterministic fading
    coefficient $h$.  So at these times, we can achieve rates up to $c(h)$.
    
    Writing $\pi(h,T)$ for the proportion of time periods when $H[t]= h$,
      \[ \pi(h,T) := \frac1T \big| \big\{ t \in \{1,2,\dots,T\} : H[t] = h \big\} \big| \qquad h\in\mathcal H , \]
    we can achieve the rate
      \[ r = \lim_{T\to\infty} \sum_{h\in\mathcal H} \pi(h,T) c(h) . \]
    But by the strong law of large numbers, we have the ergodicity property
    that $\lim_{T\to\infty} \pi(h,T) = \prob{H=h}$ (almost surely).
    Hence (again almost surely),
      \[ c \geq \sum_{h\in\mathcal H} \prob{H=h} c(h) = \EE c(H) .\]
    
    The converse can be proved using Fano's inequality, as with Shannon's
    channel coding theorem (Theorem \ref{Shannons2}).
    
    The result for continuous $H$ can be derived from this using a quantisation
    argument \cite[Appendix]{GoldsmithVaraiya}.
    \end{proof}
    
    Since this result relies on the sequence of fading parameters being ergodic,
    $c$ is sometimes called the \defn{ergodic capacity}.  So the above result can
    be interpreted as `the ergodic capacity is the average capacity.'  (Later, we
    will see how using interference alignment in networks can allow us to achieve an ergodic
    capacity that is higher than the average capacity.)
    
    Applying Theorem \ref{fastfading} to the finite field channel (Theorem 1.12), we get
      \[ c = \sum_{h\in\mathcal H} \prob{H=h} c(h) = \big( 1 - \prob{H=\zero} \big) \D(Z) . \]
    (Again, we often assume $H$ is never $\zero$, so $c = \D(Z)$ still.)
    
    For the Gaussian channel (Theorem 1.19), we have
      \[ c = \EE c(H) = \EE \log (1 + |H|^2 ) = \EE \log(1 + \SNR). \]
    (Since $\SNR$ is random here, we capitalise it.)
    
    One type of fast fading for the Gaussian channel could be a rapidly changing phase,
    $H[t] = k \ex^{\ii \Theta[t]}$, where $\Theta[t] \sim \text{U}[0,2\pi)$ IID
    over $t$.  This is a good model of wireless communication when there are many paths
    a signal could take from transmitter to receiver (in a built-up area, for example)
    \cite[Subsection 2.4.2]{TseViswanath}.
    Note that here the signal-to-noise ratio is in constant, so the capacity is unchanged.
    
    Another model of wireless communication is \defn{Rayleigh fading} \cite[Subsection 2.4.2]{TseViswanath}, where
    $H[t] \sim \mathbb{C}\text{N}(0,\tau^2)$ for some $\tau>0$.  In this case,
    $|H[t]|^2$ is exponentially distributed with mean $\tau^2$  \cite[(2.53)]{TseViswanath}.
      
\subsection{Slow fading}

Slow fading models the situation where the state of a channel is varying, but is doing so very slowly,
or where the channel state can only be modelled as random, but remains fixed.  Here we take $H[t] = H$ as
initially random, but remaining fixed for all times $t = 1, 2, \dots, T$.

Since the channel state is random, so is the capacity $C=c(H)$: if the fading is particularly deep,
$H \approx 0$, then the capacity is likely to be very low; if the fading is lighter, then the
capacity will be higher.  Specifically, under the event that $H = h$ we have $C = c(h)$.

(As with $H$,
when the capacity is a random variable, we capitalise it as $C$.)

One way to summarise the random variable $C$ would be through its cumulative distribution function
$\pout(r) := \Prob{(C \leq r)} = \Prob{(c(H) \leq r)}$, known as the \defn{outage probability}.
We can interpret this as the following: if we are trying to communicate at some fixed rate $r$, then
then $\pout(r)$ is the probability that we are unable to do so -- we say the channel is in
$\defn{outage}$.

\addcontentsline{lof}{figure}{\numberline{\protect\textbf{Definition 1.23}} Outage probability}
\begin{definition}
  For a slow fading channel with (random) capacity $C$, the \defn{outage probability}
  $\pout\colon \bR_+ \to [0,1]$ of the channel
  is defined by $\pout(r) := \Prob{(C \leq r)}$.
  
  The event $\{C \leq r\}$ is called \defn{outage}.
\end{definition}

For the Gaussian channel, we have (following Theorem 1.19 and recalling that $\log$
denotes $\log_2$)
  \[ \pout(r) = \Prob{(C \leq r)} = \Prob{\big(\log (1 + |H|^2 ) \leq r\big)} = \Prob{( \SNR \leq 2^r - 1)} , \]
where $\SNR = |H|^2$ is the signal-to-noise ratio.

For the finite field channel, we have (following Theorem 1.12)
  \[ \pout(r) = \begin{cases} 
                              \Prob{(H=\zero)} & \text{if $0 \leq r \leq \D(Z)$}  \\
                              0              & \text{if $r > \D(Z)$.}   \end{cases} \]

In wireless networks, a good model is to position nodes at random and use
distance-based attenuation fading.  Since distances between nodes are random,
the fading is random too. But once the nodes are positioned, the distances remain
fixed.  Hence, this gives a form of slow fading.

\section*{Notes}
\addcontentsline{toc}{section}{Notes}

The section consists of a review of the existing literature; the mathematical
content is not claimed to be new.

The basic concepts of information theory as outlined in this chapter are all due to Shannon's
original paper \cite{Shannon}.  An exception is the concept of relative entropy distance, due to
Kullback and Leibler \cite{KullbackLeibler}; and the Hu correspondence, due to Hu \cite{Hu}.

The presentation here closely follows the textbook of
Cover and Thomas \cite[Chapters 2, 7--9, 15]{CoverThomas}.  The textbooks of MacKay \cite[Part II]{MacKay}
and Tse and Viswanath \cite[Chapters 2, 5, 6]{TseViswanath} were also useful.

Although Shannon \cite[Theorem 11]{Shannon}
first came up with the channel coding theorem (Theorem \ref{Shannons2}), he provided only a sketch proof;
the sketch proof provided here is along the lines of the
rigorous proof by Cover \cite{CoverShannon}.  The maximum likelihood approach is
due to Gallager \cite{Gallager}.

Fading was first studied by Shannon \cite{ShannonSide}.
Our treatment of fading follows closely that of Tse and Viswanath \cite[Sections 2.1, 5.4]{TseViswanath}.
The review paper of Biglieri, Proakis, and Shamai (Shitz) \cite{BiglieriProakisShamai},
and a paper by Caire and Shamai (Shitz) \cite{CaireShamai} were useful.

\addtocontents{lof}{\protect\addvspace{20 pt}}

%% file: chapters/interference.tex
\chapter{Interference}

In this chapter, we will look at ways of dealing with interference in
communications networks.

To start with, we will define information theoretic networks, in a similar
manner to our definition of channels in Chapter 1.

We will then look at methods of combating interference -- the unwanted signals
from other transmitters that a receiver is not interested in.

For the purpose of definiteness, we will consider these in the context
of the interference networks and (mostly) the fading Gaussian case.  
However,
the techniques are useful in wider contexts.

We look at some simple schemes -- interference as noise,
decode and subtract, and resource division -- and then look at a
family of new schemes known as interference alignment. We pay particular
attention to a scheme called ergodic interference alignment, which
we will use later in Chapters 4 and 5.

\section{Wired and wireless networks}

In this chapter, we will outline the theory of networks.  We will concentrate
on accurate models of real-world wireless networks.

Wireless communications are becoming increasingly ubiquitous. From older
technologies like radios, to cutting-edge innovations such as WiFi, Bluetooth
and ZigBee, the convenience of the untethered nature of wireless is popular on
both large and small scales for businesses and consumers alike.

Compared to a wired (or wireline) network, wireless networks provide much greater
challenge to engineers and technicians.  The main problems are:
  \begin{description}
    \item[Broadcast] Each receiver can send only one signal, regardless of how
      many messages they are trying to send to how many people.
    \item[Interference] Receivers receive not just the signal corresponding
      to messages intended for them, but also all of the other transmitted signals
      as well.  These signals are called \defn{interference}.
    \item[Superposition] Receivers cannot tell which signal corresponds to which
      message, but rather receive the superposition (that is, the sum) of all
      such signals.
  \end{description}

\begin{center}
\begin{small}
\begin{tabular}{m{0.17\textwidth}m{0.35\textwidth}m{0.35\textwidth}}
  \toprule
     & \textbf{Wired networks} & \textbf{Wireless networks} \\
  \hline
    Transmission & Different signal transmitted down each wire & Same signal broadcast to all receivers \\
    Channel & Interference-free, independent noise along each wire & Interference from other transmitters and background noise \\
    Reception & Different signal received from each wire & Superposition of all signals \mbox{received} \\
    Central \mbox{difficulty} & Scheduling and routing \mbox{messages} around the network & Dealing with interference \\
  \bottomrule
\end{tabular}
\end{small}
\end{center}

In this thesis, we will mainly be looking at networks where each transmitter wishes to send a single message to a single
receiver, and each receiver requires a single message from a single transmitter.  Thus, the broadcast and superposition problems are less important that that of interference.

We capture this problem mathematically by modelling the network by a probability transition function
  \[ p(y_1, \dots, y_n \given x_1, \dots x_n) \]
relating all transmitted signals to all received signals.

  \begin{center}
    \includegraphics[scale=0.89]{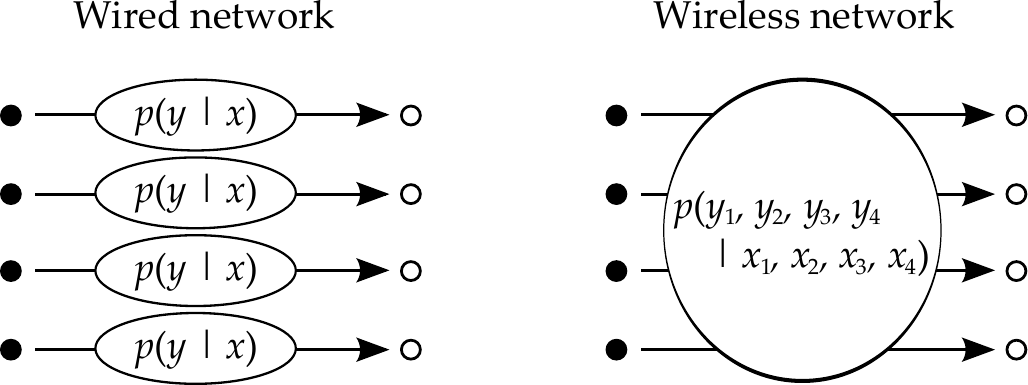}
  \end{center}

Later in this chapter, we consider a number of methods for dealing with such interference.

\section{Networks}

Point-to-point links, as we discussed in the previous chapter, are fairly well understood.  Networks, however, are much trickier.

By a \defn{network}, we mean a number of transmitters and receivers, all trying to send and receive messages through the same medium.

\addcontentsline{lof}{figure}{\numberline{\textbf{Definition 2.1}} Network}
\begin{definition}
A communications \defn{network}  consists of
  \begin{enumerate}
    \item a set $\mathcal T$ of transmitters, each with an input alphabet $\mathcal X_i$;
    \item a set $\mathcal R$ of receivers, each with an output alphabet $\mathcal Y_j$;
    \item a \defn{probability transition function}
      $p\big( (y_j : j\in \mathcal R) \given (x_i : i\in \mathcal T) \big)$ relating them.
  \end{enumerate}
\end{definition}

(In general, an agent is allowed to be a \defn{duplex} operator, that is to be both a transmitter and a receiver,
which can act as a relay in a network.
However, duplex operation will not be used in this thesis, so our definition precludes this.)

Again, we will be interested in the Gaussian and finite field channels with fading.  Because there
are now many transmitters and receivers, we will let $H_{ji}[t]$ denote the fading coefficient
at receiver $j$ from transmitter $i$.

The Gaussian and finite-field networks work in much the same way as the point-to-point channels, with
the change that now receivers experience the superposition (that is, sum) of all signals sent.

\addcontentsline{lof}{figure}{\numberline{\textbf{Definition 2.2}} Gaussian and finite field networks}
\begin{definition}
  \defn{Gaussian networks} have $\mathcal X_i = \mathcal Y_j = \C$ for all $i$ and $j$.  The probability
  transition measure is implied by the relationship
    \[ Y_j[t] = \sum_{i \in \mathcal T} H_{ji}[t] x_i[t] + Z_j[t] \qquad j \in \mathcal R , \]
  where $Z_j[t] \sim \CN(0,1)$ independently across $j$ and $t$.
  
  The \defn{finite field network} of size $q$ with noise $Z$ has $\mathcal X_i = \mathcal Y_j = \F_q$
  for all $i$ and $j$.  The probability transition measure is implied by the relationship
    \[ Y_j[t] = \sum_{i \in \mathcal T} H_{ji}[t] x_i[t] + Z_j[t] \pmod q \qquad j \in \mathcal R , \]
  where $Z_j[t]$ are independently and identically distributed like $Z$.
\end{definition}

When it's convenient, we will write these networks in matrix form, that is
  \[ \vec Y[t] = \mathsf{H}[t] \vec x[t] + \vec Z[t] , \]
where $\vec Y[t] = \big(Y_j[t] : j\in\mathcal R \big)$ is the received vector, $\vec x[t] = \big(x_i[t] : i\in\mathcal T \big)$
is the transmitted vector, $\vec Z[t] = \big(Z_j[t] : j\in\mathcal R \big)$ is the noise vector,
and $\mat H[t] = \big(H_{ji}[t] : i \in \mathcal T, j\in\mathcal R \big)$ is the channel-state matrix.

To design a code for a network, we need to specify which transmitters are trying to send a message to which receivers;
then each transmitter needs an encoding function, and each receiver a decoding function (or more than one, if they are
receiving many messages).

\addcontentsline{lof}{figure}{\numberline{\textbf{Definition 2.3}} Code for a network, rate, sum-rate, error probability}
\begin{definition}
  A \defn{code} for the network
    \[ \bigg( \mathcal T, \mathcal R, (\mathcal X_i : i \in \mathcal T),
        (\mathcal Y_j : j \in \mathcal R), p\big( (y_j : j\in \mathcal R) \given (x_i : i\in \mathcal T) \big)\bigg) \]
  consists of
  \begin{enumerate}
    \item a set $\mathcal L \subseteq \mathcal T \times \mathcal R$ of $L$ direct links (we call
      the other links in $(\mathcal T \times \mathcal R) \setminus \mathcal L$ the
      \defn{crosslinks});
    \item a \defn{message set} $\mathcal M_{ij}$ of cardinality $M_{ij}$ for
      each link $i\link j \in \mathcal L$;
    \item an \defn{encoding function} $\vec x_i\colon \prod_{j:i\to j\in\mathcal L} \mathcal M_{ij} \to\mathcal X_i^T$
      for each transmitter $i \in \mathcal T$;
    \item a \defn{decoding function} $\hat m_{ij} \colon\mathcal Y_j^T\to\mathcal M_{ij}$
      for each link $i\link j \in \mathcal L$.
  \end{enumerate}
  
  On link $i \link j \in \mathcal L$, the rate is $r_{ij} := \log M_{ij} / T$, the rate vector
  is $\vec r = (r_{ij} : i\link j \in \mathcal L )$, and the sum-rate is
  $r_\Sigma := \sum_{i \to j \in \mathcal L} r_{ij}$.  The error probability on link $i \link j$ is
    \[ e_{ij} := \frac1{M_{ij}} \sum_{m\in\mathcal M_{ij}} \prob{\hat m_{ij}(\vec y_j) \neq m\given \vec x_i(m) \text{ sent}} . \]
    
  When dealing with the Gaussian case, the power of transmitter $i$ is
  the maximum value of $\overline{|x_i|^2} := 1/T \sum_{t=1}^T |x_i[t]|^2$ over all $i$'s codewords $\vec x_i$.
\end{definition}

Some common examples of networks are the following:

\addcontentsline{lof}{figure}{\numberline{\textbf{Definition 2.4}} Examples of networks}
\begin{definition}
The \defn{point-to-point link} is just a special case of a network with
    \[ \mathcal T = \{ \text{Alice} \} , \qquad
       \mathcal R = \{ \text{Bob} \} , \qquad
       \mathcal L = \mathcal T \times \mathcal R = \{ \text{Alice} \link \text{Bob} \} . \]

  \begin{center}
    \includegraphics[scale=0.89]{figures/networkp2p}
  \end{center}
  
The \defn{multiple-access network} has multiple transmitters sending to one receiver, so
    \[ \mathcal T = \{ 1, \dots, n \} , \qquad
       \mathcal R = \{ \text{Bob} \} , \qquad
       \mathcal L = \mathcal T \times \mathcal R = \{ 1 \link \text{Bob} , \dots,  n \link \text{Bob} \} . \]
    
      \begin{center}
    \includegraphics[scale=0.89]{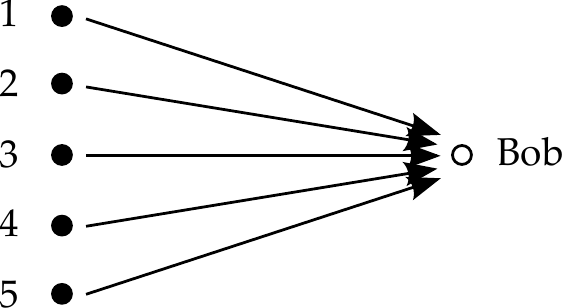}
  \end{center}
    
The \defn{broadcast network} has one transmitter sending to multiple receivers, so
    \[ \mathcal T = \{ \text{Alice} \} , \quad
       \mathcal R = \{ 1, \dots, n \} , \quad
       \mathcal L = \mathcal T \times \mathcal R = \{ \text{Alice} \link 1 , \dots,  \text{Alice} \link n \} . \]
       
         \begin{center}
    \includegraphics[scale=0.89]{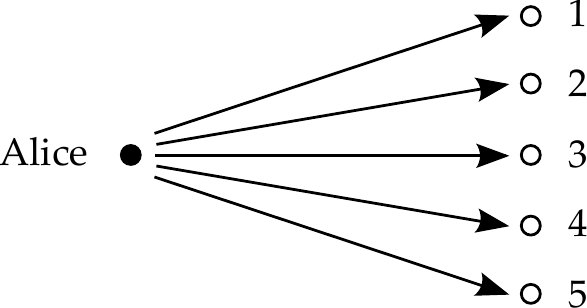}
  \end{center}
       
The \defn{interference network} consists of multiple point-to-point links communicating
  over the same medium, so
    \[ \mathcal T = \{ 1, \dots, n \} , \qquad
       \mathcal R = \{ 1, \dots, n \} , \qquad
       \mathcal L = \{ 1 \link 1 , 2 \link 2 \dots,  n \link n \} . \]
  (Note that this differs from $n$ independent point-to-point links, since each receiver is also receiving
  the signals from the other transmitters, even though they have no use for that signal.)
       
         \begin{center}
    \includegraphics[scale=0.89]{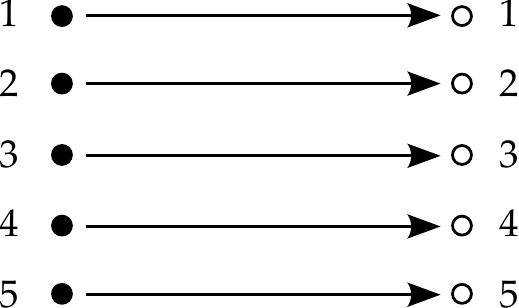}
  \end{center}
       
The \defn{X network} consists of an equal number of transmitters and receivers communicating across
    all possible links, so
    \[ \mathcal T = \{ 1, \dots, n \} , \ 
       \mathcal R = \{ 1, \dots, n \} , \ 
       \mathcal L = \mathcal T \times \mathcal R = \{ 1 \link 1 , 1 \link 2, 1\link 3 \dots,  n \link n \} . \]
       
         \begin{center}
    \includegraphics[scale=0.89]{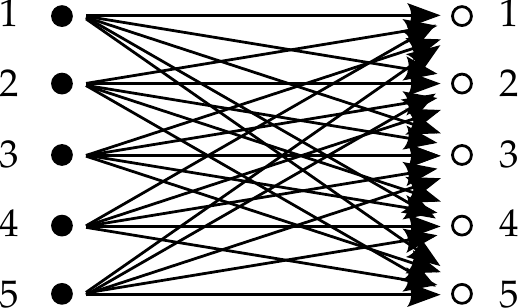}
  \end{center}
       
\end{definition}

As with point-to-point links, we are interested in the maximum rate at which we can send information through
a network.  However, since we now have several competing links in the same network, no one benchmark
will describe this.  For example, achieving a high rate on one particular link may use up a lot of
channel resources, leading to slower communication on another link.  (Consider trying to hold a conversation
in a room where lots of other people are shouting.)

Instead the set of achievable rate vectors will be a region of $L$-space;
we call this the \defn{capacity region}.  (Recall from Definition 2.3 that
$L = |\mathcal L|$ is the number of links in the network.)

\addcontentsline{lof}{figure}{\numberline{\textbf{Definition 2.5}} Achievable rate vector}
\begin{definition}
  Consider a network $\big( \mathcal T, \mathcal R, (\mathcal X_i ),
        (\mathcal Y_j ), p( (y_j ) \given (x_i ) )\big) $. 
  
  A rate vector $\vec r = (r_{ij} : i\link j \in \mathcal L)$ is
  \defn{achievable} for links $\mathcal L$ if for any error tolerance
  $\epsilon > 0$, there exists a code for $\mathcal L$ with rate on
  each link $i\link j$ at least $r_{ij}$ and all error probabilities
  lower than $\epsilon$.
  
  Otherwise, $\vec r$ is \defn{not achievable}, in that there exists an error threshold
  $\epsilon$ such that there exists no code for $\mathcal L$ with rates at least $r_{ij}$
  and all error probabilities lower than $\epsilon$.
\end{definition}

\addcontentsline{lof}{figure}{\numberline{\textbf{Definition 2.6}} Capacity region}
\begin{definition}
  Consider a network $\big( \mathcal T, \mathcal R, (\mathcal X_i ),
        (\mathcal Y_j ), p( (y_j ) \given (x_i ) )\big) $.    Then we define
  the \defn{capacity region} of the channel, $\mathcal C$, to be the closure of
  the set of all achievable vectors:
    \[ \mathcal C := \overline{ \{\vec r \in \bR_+^L : \vec r \text{ is achievable} \} }. \]
    
  In other words, all rate vectors $\vec r$ in the interior of $\mathcal C$ are achievable,
  but no $\vec r$ outside $c$ is achievable.
\end{definition}

Note that the capacity region will always be convex: Suppose the rate vectors $\vec r_1$ and $\vec r_2$ are
both achievable.  Then the rate vector $\lambda \vec r_1 + (1-\lambda) \vec r_2$, for $\lambda \in [0,1]$ is achievable by operating
at $\vec r_1$ for $\lambda T$ of the time points and at $\vec r_2$ for the remaining $(1-\lambda)T$ timeslots.  This
strategy is known as \defn{time sharing}; we discuss this further in Subsection 2.4.1.

\addcontentsline{lof}{figure}{\numberline{\textbf{Definition 2.7}} Sum-capacity}
\begin{definition}
  We define the \defn{sum-capacity} $c_\Sigma$ to be the maximum achievable sum-rate, so
    \[ c_\Sigma := \max_{\vec r \in \mathcal C} r_\Sigma = \sup \{ r_\Sigma : \text{$\vec r$ is achievable} \} . \]
\end{definition}

The current knowledge of capacity regions for these networks in the Gaussian and general cases is
summarised in the table below.

\begin{center}
    \begin{small}
	
		\begin{tabular}{m{2.4cm}m{3.2cm}m{4.8cm}}
		  \toprule
		    \textbf{Network} & \textbf{General case} & \textbf{Gaussian case} \\
		  \hline
		    Point-to-point & known (Theorem 1.11) & known (Theorem 1.19) \\
		    Multiple-access & known \cite{Ahlswede,Liao} & known (Theorem 2.8) \\
		    Broadcast & unknown; known for some special cases \cite{CoverComments} & known \cite{CoverBroadcast} \\
		    Interference & unknown;  known for some special cases \cite{KramerInterference} & unknown; known for some special cases \cite{KramerInterference}; sum-capacity known for most two-user cases \cite{KramerSum} \\
		    X & unknown & unknown \\
		  \bottomrule	
		\end{tabular}
		
		\end{small}
\end{center}

Later, we will use the capacity region of the multiple-access network.  It was discovered independently by
Ahlswede \cite{Ahlswede} and Liao \cite{Liao} in the 1970s.  In the Gaussian case, it simplifies to
the following:

\addcontentsline{lof}{figure}{\numberline{\textbf{Theorem 2.8}} Capacity region of the multiple-access channel}
\begin{theorem}
  The capacity region of the multiple-access network of $n$ transmitters with fixed fading is
  the set of $(r_1,r_2,\dots,r_n) \in \bR_+^n$ satisfying
    \[ \sum_{i \in \mathcal S} r_i \leq \log\left(1+ \sum_{i \in \mathcal S} \snr_i \right) \]
  for all $\mathcal S \subseteq \{1,2,\dots,n\}$, where $\snr_i := |h_i|^2$ is the
  signal-to-noise ratio from transmitter $i$.
  
  The sum-capacity is
    \[ c_\Sigma = \log\left(1+ \sum_{i =1}^n \snr_i \right) . \]
\end{theorem}
  
In this thesis, we will mostly be interested in the $n$-user interference network. We will
use the word `user' to denote a matching transmitter--receiver pair.  Hence, an $n$-user
network consists of $n$ transmitters and $n$ receivers.  We will mostly be
interested in the large $n$ limit.

Recent work by Jafar \cite{Jafar} in the fixed $\snr$, $n\to\infty$ regime has
shown much promise.  We review Jafar's work in detail later in this chapter and in Chapter 4, and extend it
to physical models of wireless networks.


%
%
%

Alternatively, in the fixed $n$, $\snr\to\infty$ regime, Cadambe and Jafar \cite{CadambeJafar} 
used interference alignment to deduce the limiting behaviour within
$o(\log(\snr))$.  These techniques were extended by the same authors \cite{CadambeJafar2}
to more general models  in the presence of feedback and other effects.

For small $n$, the classical bounds due to Han and Kobayashi \cite{HanKobayashi}
as refined by Chong, Montani, Garg and El Gamel \cite{chong} for the two-user Gaussian interference network
have recently been extended. For example Etkin, Tse, and Wang
\cite{EtkinTseWang} have produced a characterization of capacity accurate to within one
bit. These results were extended by Bresler, Parekh, and Tse \cite{BreslerParekhTse},
using insights based on a deterministic channel which approximates the Gaussian
channel with sufficient accuracy, to prove results for many-to-one and one-to-many
Gaussian interference channels.

A different approach towards finding the capacity of large communications
networks is given by the deterministic approach of Avestimehr, Diggavi, and Tse \cite{Avestimehr}. They
show how capacities can be calculated up to a gap determined by the
number of users $n$, across all values of $\snr$. 

More generally, in problems concerning networks with a large number of 
nodes, the work of Gupta and Kumar \cite{GuptaKumar} uses
techniques based on Voronoi tesselations to establish scaling laws.
(See also the survey paper of Xue and Kumar \cite{XueKumar} for a review of the information theoretical
techniques that can be applied to this problem.)

{\"O}zg{\"u}r, L{\'e}v{\^e}que, and Tse  \cite{Ozgur} and {\"O}zg{\"u}r and L{\'e}v{\^e}que \cite{Ozgur2}
use a similar model of
dense random network placements, though using the same points as both transmitters
 and receivers.
They describe a hierarchical scheme, where nodes are successively
assembled into groups of increasing size, each group collectively acting as a multiple antenna
transmitter or receiver, and restrict to transmissions at a common
rate. They show \cite[Theorems 3.1, 3.2]{Ozgur} that
for any $\epsilon > 0$ there
exists a constant $k = k(\epsilon)$ depending on $\epsilon$ and a fixed constant $K$ such that 
\begin{equation} \label{eq:ozgur}
k n^{1-\epsilon} \leq \csum \leq K n \log n.
\end{equation}
These bounds are close to stating that $\csum$
grows like $n$,  but without 
the explicit constant
that Jafar \cite{Jafar} and the work in Chapter 4 of this thesis achieve. 
(Later, we produce a version of the upper bound  
(\ref{eq:ozgur}) without the logarithmic factor 
and being explicit about the constant $K$. Note that this result is proved under a model 
that differs from that of {\"O}zg{\"u}r, L{\'e}v{\^e}que, and Tse  \cite{Ozgur}, and the fact that we have a total of $2n$ nodes rather than $n$ -- although this is unimportant for asymptotic results. Further, in their work, local collaboration and relaying are both
allowed, meaning that the true rate in their scenario
could indeed be $n \log n$ .)

\section{Interference as noise}

Recall that we mentioned in Section 1.4 that the optimal input distribution
to the Gaussian channel is $\CN(0,P)$, while the noise is
$\CN(0,\sigma^2)$.  Thus, interference has the same distribution as noise
(after an appropriate scaling).
So to receiver $j$, the received signal
  \[ Y_j = \sum_{i=1}^n h_{ji} x_i[t] + Z_j[t] \qquad Z_j[t] \sim \CN(0,1) \]
is statistically indistinguishable from
  \[ \tilde Y_j = h_{jj}x_j[t] + \tilde Z_j[t] \qquad Z_j[t] \sim \CN \left(0,1+\sum_{i\neq j}|h_{ji}|^2 \right) . \]

In other words, 
treating interference as noise allows user $j$ to communicate at
rate
  \[ r_j = \log (1 + \sinr_j) = \log \left(1 + \frac{|h_{jj}|^2}{1+\sum_{i\neq j} |h_{ji}|^2} \right) .\]
Here, $|h_{jj}|^2$ is the received power of the signal, and
$\sum_{i\neq j} |h_{ji}|^2$ the received power of the interfering
signals from other transmitters.  We call 
  \[\sinr_j := \frac{|h_{jj}|^2}{1+\sum_{i\neq j} |h_{ji}|^2} \]
the \defn{signal-to-interference-plus-noise ratio} at receiver $j$.

\addcontentsline{lof}{figure}{\numberline{\textbf{Theorem 2.9}} Achievable rates treating interference as noise}
\begin{theorem}
  Consider an $n$-user Gaussian interference network.  The rates
  $r_j = \log(1+\sinr_j)$ are simultaneously achievable.
  
  The convex hull of the set of $\vec r$ with $r_j \leq \log(1+\sinr_j)$
  for all $j$ is an inner bound for the capacity region.
\end{theorem}

When the interference is weak, that is when $\sum_{i\neq j}|h_{ji}|^2 \ll 1$,
this strategy will be quite effective, as we will have
  \[ \sinr_j = \frac{|h_{jj}|^2}{1+\sum_{i\neq j} |h_{ji}|^2} \approx |h_{jj}|^2 = \snr_j . \]
Hence, in this situation, user $j$ can communicate at almost the
same rate it could were there no interference at all.

However, if the interference is strong, this strategy will lead to a
dramatic decrease in the rate.  In this case we will need different strategies.

\section{Decode and subtract}

The tactic of treating interference as noise suggests a method of
dealing with strong interference.

Suppose we have just one interfering
link $2\link 1$ that is very strong.  Then we can treat the interference as signal,
and treat the signal itself as noise.  This allows
us to decode the interfering signal $x_2$ (with $x_1$ as
noise). Once we have decoded $x_2$, we know the interference
$h_{12}x_2$ (since we have perfect channel state information), allowing us to subtract it.  This forms the
interference free signal
  \[ \tilde Y_1 := Y_1 - h_{12}x_2 = (h_{11} x_1 + h_{12}x_2 + Z_1) - h_{12}x_2 = h_{11}x_1 + Z_1 . \]
  
In this case, receiver $1$ requires 
  \[ r_1 \leq \log \left(1 + \frac{|h_{12}|^2}{|h_{11}|^2+1} \right) \]
to decode the interference, and then $r_1 \leq \log(1+\snr_1)$ to decode
the signal, once the interference  has been subtracted.

If the `interference-to-noise-plus-signal ratio' $|h_{12}|^2/(|h_{11}|^2+1)$ is large
-- that is, if the interference $|h_{12}|^2$ is strong compared to the signal $|h_{11}|^2$--
then this will be almost as effective as if no noise were present.


\section{The problem of mid-level interference}

So far, we have two principles:
\begin{description}
  \item[Weak interference] should be treated as noise.
  \item[Strong interference] should first be decoded, and then subtracted.
\end{description}

This leads to a natural question: what about mid-level interference?  That is, what
is the best way of dealing with interference when the power of the interference is
roughly equal to the power of the signal?

Indeed, there are many plausible real-life situations where the mid-level interference
would seem to be the most likely.

For example, in cellular networks (such as mobile phone networks), we have a
phenomenon known as the \defn{edge-of-cell effect}.  This describes the phenomenon
that near the edge of a cell, the strength of a signal is of a very similar
level to if it just fell out of the cell.  Combating the mid-level interference
of these edge-of-cell transmitters is essential to maintaining a high-quality system.

	\begin{center}
		\includegraphics[width=0.92\textwidth]{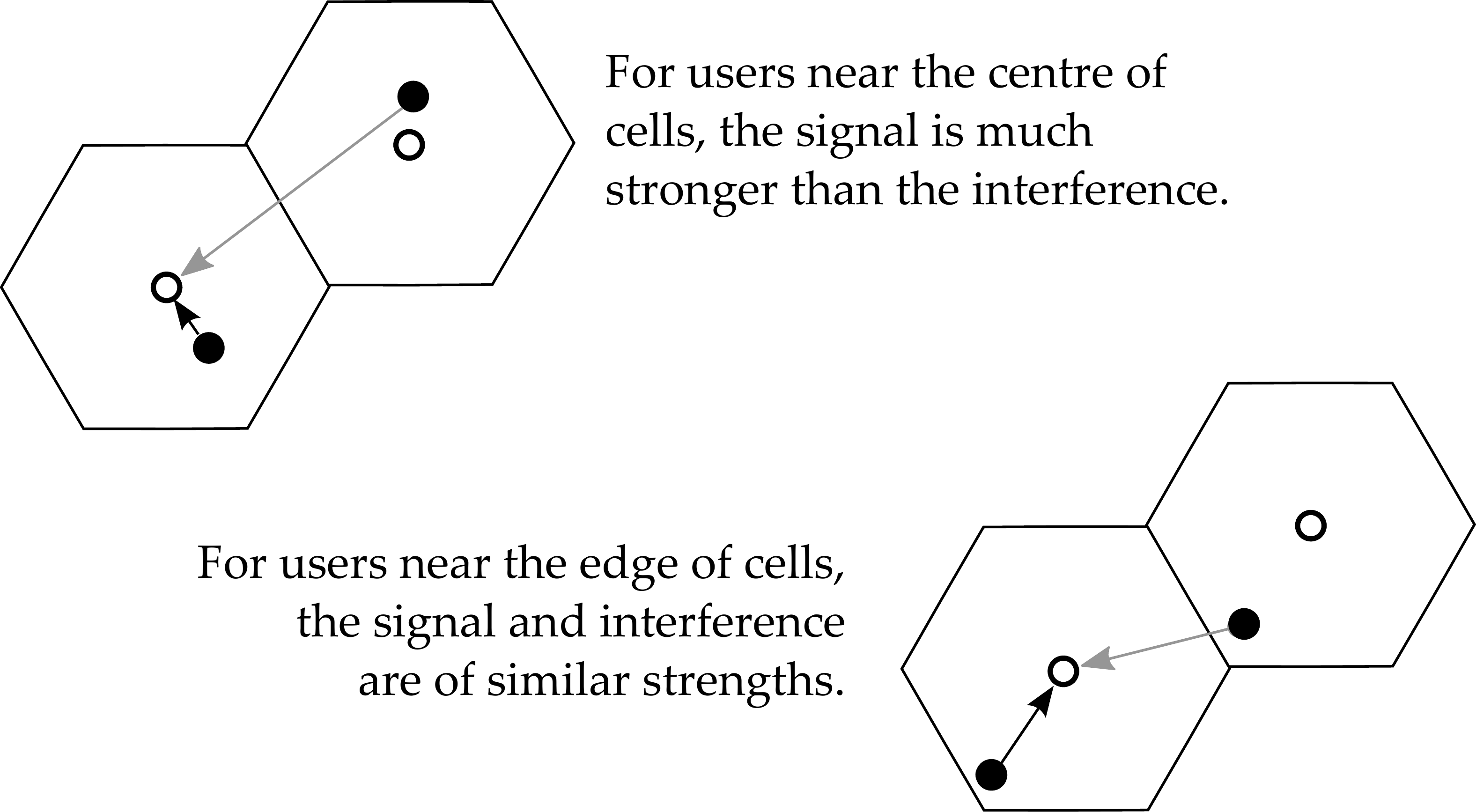}
\end{center}

In fact, dealing with this mid-level interference turns out to be particularly
important in multi-user networks; while weak and strong interference are easily dealt
with, even rare occurrences of mid-level interference can severely restrict the
performance of such a network.  The following example, due to Jafar \cite{Jafar},
illustrates this.

Consider a two-user Gaussian interference network, as governed by the
input--output equations
  \begin{align*}
    Y_1[t] &= h_{11} x_1[t] + h_{12} x_2[t] + Z_1[t] \\
    Y_2[t] &= h_{21} x_1[t] +h_{22} x_2[t] + Z_2[t].
  \end{align*}

We will use a model with a fast-fading phase.  So
 \begin{align*}
   h_{11}[t] &= \sqrt{\snr_1} \exp(\ii \Theta_{11}[t]) &
   h_{12}[t] &= \sqrt{\inr_{12}} \exp(\ii \Theta_{12}[t]) \\
   h_{21}[t] &= \sqrt{\inr_{21}} \exp(\ii \Theta_{21}[t]) &
   h_{22}[t] &= \sqrt{\snr_2} \exp(\ii \Theta_{22}[t]) ,
  \end{align*}
where the $\Theta_{ji}[t]$ are IID uniform on $[0, 2\pi)$.

Here $\inr_{ji} = |h_{ji}|^2$ is the \defn{interference-to-noise}
ratio at receiver $j$ from transmitter $i$.

For simplicity, we shall fix the direct links to be of equal strength:
$\snr_1 = \snr_2 =: \snr$.

Now suppose just one of the two interfering crosslinks is precisely this
mid-level interference: so $\inr_{21} = \snr$ too.

	\begin{center}
		\includegraphics[width=0.48\textwidth]{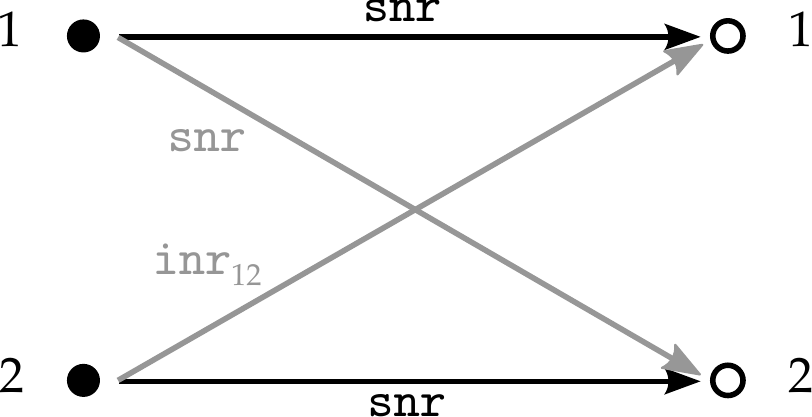}
  \end{center}

Jafar \cite[Lemma 1]{Jafar} then showed the following surprising result:

\addcontentsline{lof}{figure}{\numberline{\textbf{Theorem 2.10}} Sum-capacity of a network with a bottleneck link}
\begin{theorem}
The sum-capacity of the above network is
$\csum = \log(1+2\snr)$, regardless of the value of
$\inr_{12}$.
\end{theorem}

\emph{``Regardless of the value of $\inr_{12}$''!}
This is worth emphasising: just one crosslink of this mid-level
interference has completely determined the sum-capacity of the whole
network.

In Chapter 4, we shall see similar phenomena for networks with
many more users, where the study of bottleneck links will be vitally
important.

\begin{proof}
  \emph{Direct part}.  Achievability follows from using ergodic interference alignment (Theorem 2.14),
    which we consider later in Section 2.6, or by timesharing (Section 2.6.1).
    
  \emph{Converse part.} Suppose we have a code allowing us to achieve the sum-rate $r_\Sigma = r_1 + r_2$.
  
  Suppose a genie provides receiver $2$ with transmitter $1$'s message (which could
  only increase the capacity of the network). This allows receiver $2$ to cancel the interference
  due to $x_1$.
  
  By assumption, receiver $1$ can decode his own message, and thus cancel the
  intended signal $x_1$ from his received signal $y_1$.
  
  This leaves the two receivers with statistically equivalent signals.
  Therefore, if receiver $2$ can decode message $m_2$ -- which it can
  by assumption -- then so can receiver $1$.
  
  Since receiver 1 is able to decode both messages, the sum rate cannot be more
  than the sum rate capacity of the multiple access channel seen at receiver $1$,
  which by Theorem 2.8 is $r_\Sigma \leq \log(1+2\snr)$.
\end{proof}

(We will use a similar proof strategy later to prove Lemma 4.11.)

This shows that dealing well with this mid-level interference will be
particularly important. We examine ways of doing so in the next section.

\section{Resource division}

When faced with mid-level interference, a simple way of dealing with it
is to share out the channel resources between the transmitters, to stop
the interference from getting in the way.  (This technique is often known
as orthogonalisation.)

This idea is best illustrated by examples, of which we give three below.

A big advantage of resource-division strategies is that they are fairly
easy to set up, and require neither detailed ongoing channel knowledge nor
high computational complexity.  There is also little need for cooperation
between users after setup.

In this section, we assume for simplicity that all $\snr$s are equal.

\subsection{\dots by time}

In schemes which share out the time resource, each transmitter is given
sole use of the channel for some period of time, in return for which they
may not transmit the rest of the time.

Consider the case of the finite-field model.  At any particular time,
only one user has control of the channel, and
they can communicate up to their single-user capacity $\mathbb D(Z)$.  All the other
transmitters are silent, so the sum-rate is also $r_\Sigma = \mathbb D(Z)$.

Each user can communicate at a rate $r = \mathbb D(Z)/n$,
for $\dof = 1/n$ degrees of freedom each (recall Definition 1.13), which in particular
tends to $0$ as the number of users $n$ gets large.  Naturally this is undesirable.

For the Gaussian model, transmitters can take advantage of the fact that they
will only be transmitting part of the day, yet are operating under an \emph{average}
power constraint.  Hence, when they are transmitting, transmitters can use
power $nP$ instead.  Hence, each user can communicate at a rate $r = \frac 1n \log(1+n\snr)$.

Note that for low $\snr$ we have
  \[ \frac{\csum}{n} = \frac1n \log(1+n\snr) \approx \frac1n n\,\snr \log \ee = \snr \log \ee, \]
which is the same as if there was no interference at all.  Hence, for low $\snr$, these schemes
are optimal.

For high $\snr$ (the more common case), we have 
  \[ \frac{\csum}{n} = \frac1n \log(1+n\snr) \approx \frac1n (\log \snr + \log n) \approx \frac1n \log \snr , \]
a reduction over the single-user case of a factor of $1/n$ again.  That is,
each user has $\dof = 1/n$ degrees of freedom each, for a total of $\dof_\Sigma = 1$ degrees of freedom all together
(recall Definition 1.20).

In summary:

\begin{theorem}\addcontentsline{lof}{figure}{\numberline{\textbf{Theorem 2.11}} Achievable rates using resource division}
  Consider an $n$-user finite field interference network.  Then the
  rates $r_i = \D(Z)/n$ are simultaneously achievable, for a sum-rate
  of $r_\Sigma = \D(Z)$ and $\dof_\Sigma = 1$.
  
  Consider an $n$-user Gaussian interference network.  Then the
  rates $r_i = \frac 1n \log(1+n\snr_i)$ are simultaneously achievable.
  If all $\snr$s are equal, this has a sum-rate 
  of $r_\Sigma = \log(1+n\snr)$, which has $\dof_\Sigma = 1$.
\end{theorem}

H\o st-Madsen and Nosratinia showed that $\dof=1/2$ each, for a total of
$\dof_\Sigma = 1$ is optimal for $n=2$ users,
and, more generally, showed that \cite[Section IV]{HostMadsenNosratinia}
  \begin{equation} \label{HMN}
    1 \leq \dof_\Sigma \leq \frac n2 \qquad \text{for all $n\geq 2$.}
  \end{equation}

They further conjectured that in
fact it is the lower bound that is tight, and that $\dof_\Sigma = 1$ is optimal
\cite[Section IV]{HostMadsenNosratinia}.  In other words, they conjectured
that resource division strategies were also optimal at high $\snr$.  We
will later see that this is not the case.

Note that this scheme, as with all resource division schemes,
requires a small amount of precoordination between users, to decide
which user is alloted which timeslot.  We do not consider in this thesis
the problem of how to conduct this cooperation.

\subsection{\dots  by frequency}

In schemes that share out the the frequency resource, each user is alloted
a section of the frequency spectrum along which they may transmit, while
remaining quiet over all other bandwidths.

This leads again to an average per-user capacity of $c/n$ again.

Second-generation (GSM) mobile phone networks use a mixture of resource
division by time and by frequency to share a channel of bandwidth
$25$\,MHz between up to $1000$ transmitters.  First the spectrum is shared,
giving $125$ channels of bandwidth $200$\,kHz each.  Then within each
of these sub-channels, the time is divided between up to $8$ transmitters in
time slots of $577\,\mu$s at a time \cite[Example 3.1]{TseViswanath} \cite[Example 14.2]{Goldsmith}.

\subsection{\dots  by codeword space}

Code division multiple access (CDMA) is another way of allowing multiple users to
use a shared channel.  It works as follows.  (For the purpose of simplifying this example,
we shall think of a noiseless binary channel.) 

Assume each transmitter $i$ wishes to send a symbol $x_i \in \zo$.  Using CDMA
it does this over $T$ channel uses.  Each transmitter $i$ is given a vector
$\vec v_i \in \zo^T$.  If that transmitter wish to send $x_i = \mathtt 1$, she
instead transmits the vector $\vec v_i$ over $T$ channel uses; if she wishes
to send $x_i = \mathtt 0$, she instead transmits $\vec 0$, the zero vector of length
$T$.  In other words, transmitter $i$ sends $x_i \vec v_i$.

The receivers then receive the superposition $\vec y = \sum_{i=1}^n x_i \vec v_i$.
If the vectors were linearly independent (for which $T\geq n$ is necessary but
not sufficient), then each $x_i$ can be recovered.

Schemes such as this can be thought of as sharing the dimensions of the codeword space
$\zo^T$ among the $n$ users, and so is another form of resource division.

\section{Interference alignment}

Interference alignment is the name for a new class of schemes for
dealing with interference based on the following idea: if transmitters
plan their signalling correctly, interference can `align' at each receiver,
with the desired signal split off separately.  This
allows receivers to share their resources just two ways -- half for the
signal, and half for all the `aligned' interference.  Thus all users
can communicate at (roughly) the rate they could if there were just one interfering link.

In particular, this means that each user can obtain $1/2$ a degree of freedom, leading
to $\dof = n/2$ degrees of freedom overall.  So in fact, the upper bound \eqref{HMN} of
H\o st-Madsen and Nosratinia \cite[Section IV]{HostMadsenNosratinia} turns out to be correct,
disproving their conjecture that the lower bound $\dof = 1$ was tight.
(The conjecture was formally settled in the
negative by Cadambe, Jafar, and Wang \cite{CadambeJafarWang}.)

Like resource division, interference alignment can be performed in a number
of different ways.  We show these by three toy examples, before concentrating
on the specific case that will be useful to us later.

These schemes require channel state information at the transmitter
(CSIT), in that the signal set in any time slot depends on the
channel state coefficients at that time.

\subsection{\dots  by codeword space}

Note that CDMA (see Section 2.6.3) was, in some sense, wasteful for the interference network, since 
it allowed every receiver to decode every message, as if it were an X network.  Whereas in fact, we only
required each receiver to decode its own message.

Interference alignment by codeword choice, due to Cadambe and Jafar \cite{CadambeJafar},  develops this idea
further.
Each receiver receives a superposition of all the transmitters' (faded)
signals, but by `aligning' the interference, the receiver can
work out its own signal, at less of a sacrifice than CDMA.

Consider the following $3$-user fading interference network:

  \begin{align*}
    Y_1 &= x_1 + \ii x_2 + \ii x_3 + Z_1 \\
    Y_2 &= \ii x_1 + x_2 + \ii x_3 + Z_2 \\
    Y_3 &= \ii x_1 + \ii x_2 + \ii x_3 + Z_3
  \end{align*}

Note for this toy example we have chosen the unusual channel state matrix
  \[ \mat H = \begin{bmatrix} 1   & \ii & \ii \\
                              \ii & 1   & \ii \\
                               \ii & \ii & 1  \end{bmatrix}, \]
which has diagonal entries (corresponding to direct links) all equal to $1$,
and off-diagonal entries (interfering links) all equal to $\ii$.

Suppose now that transmitters send their signal as just real numbers (taking
advantage of the power constraint to send at twice the power).
The receivers will receive their desired signal in the real subspace,
but the interference will be aligned in the (purely) imaginary subspace.

Thus, each user can communicate interference-free at rate $r_i = \frac12 \log(1+2\snr)$.
Since we saw earlier that this was optimal at high $\snr$ for $2$ users, it must certainly be optimal
for $3$ users too.  Hence, this interference network has a sum-capacity of
 \[ c_\Sigma = \frac32 \log (1 + 2\snr) > \log(1 + 3\snr)   , \]
better than the $r_\Sigma = \log(1 + 3\snr)$ achievable by resource division.

Cadambe and Jafar \cite{CadambeJafar} managed to develop this idea, to show that it was possible
for any values of fading parameters (provided there are no `unexpected' linear
dependencies -- this would be avoided almost surely if the fading coefficient were from continuous
distributions, for example).

They showed how transmitters can construct their signals so that at each receiver
the signal is contained in one subspace of the signal space, and all the interference
in another disjoint subspace, with both subspaces using roughly
half of the available dimensions.

Specifically, they showed the following \cite[Theorem 1]{CadambeJafar}.  (Recall
the definition of degrees of freedom from Definition 1.20.)

\addcontentsline{lof}{figure}{\numberline{\textbf{Theorem 2.12}} Sum-capacity using interference alignment}
\begin{theorem}
  Consider a Gaussian interference network with fixed fading coefficients
  $h_{ji}$.  Then the total number of degrees of freedom is
  $\dof_\Sigma = n/2$.
  
  That is, for equal $\snr$s, the sum-capacity is
    \[ c_\Sigma = \frac{n}{2} \log \snr + o(\log \snr) \qquad \text{as $\snr \to \infty$.}\]
\end{theorem}

Note that the per-user capacity is
  \[ \frac{c_\Sigma}{n} = \frac{\frac{n}{2}\log(\snr) + o(\log \snr)}{n} = \frac12\log \snr + o(\log \snr) , \]
compared with a single-user rate of
  \[ c = \log(1+\snr) = \log \snr + o(\log \snr) , \]
so the rate has been roughly halved.  This compares well to
the reduction to $1/n$ of resource division strategies, at least
when $n > 2$.

El Ayach, Peters, and Heath \cite{AyachPetersHeath} have conducted experiments that show that
this interference alignment technique can perform well in real life
for $n=3$ users, with performance close to that predicted by theory.

\subsection{\dots  by time}

If a network has time delays, we can take advantage of these delays
to align interference in the time domain.  Interference alignment by
time was first considered by Grokop, Tse and Yates \cite{GrokopTseYates}.
However, due to the computational complexity of such schemes and
the lack of physical applicability, it has received little
attention since.

Specifically let $\tau_{ji}$ be the time delay between transmitter
$i$ and receiver $j$.  Thus we have the model
  \[ Y_j[t] = \sum_{i=1}^n h_{ji} x_i[t-\tau_{ji}] + Z_j[t] ,\]
(with the convention that $x_i[t]$ is $0$ for $t\leq 0$).

Consider a toy example with the following time delays:
  \begin{align*}
    \tau_{11} &= 3  &  \tau_{12} &= 4  &  \tau_{13} &= 6 \\
    \tau_{21} &= 4  &  \tau_{22} &= 7  &  \tau_{23} &= 2 \\
    \tau_{31} &= 2  &  \tau_{32} &= 8  &  \tau_{33} &= 1
  \end{align*}
Note that this has been set up so that delays on direct links
are odd numbers, while delays on crosslinks are even numbers.

This allows us to use the following strategy. Transmitters only send
symbols at the odd numbered times, $t=1,3,5,\dots$.  Then at even-numbered
times, receivers will only get their desired signal
(since $\text{odd}+\text{odd}=\text{even}$), and at odd-numbered times, two
lots of `aligned' interference ($\text{odd}+\text{even}=\text{odd}$).

Hence, users can communicate at half their single-user rate.

Grokop, Tse, and Yates showed a generalisation of this, but it is quite complicated:
see their paper \cite[Theorem 3.1]{GrokopTseYates} for details.

\subsection{\dots by channel state}

Consider a fast fading $3$-user interference network where the channel
state matrix can take either of the following two values
  \[ \mat H  = \begin{bmatrix} 1 & 0 & 0 \\
                               0 & 0 & 1 \\
                               1 & 1 & 1 \end{bmatrix} \qquad
     \mat H' = \begin{bmatrix} 0 & 0 & 0 \\
                               0 & 1 & 1 \\
                               1 & 1 & 0 \end{bmatrix} . \]
Note that that this toy example has been set up such that
$\mat H + \mat H' = \mat I \pmod{2}$.

Nazer, Gastpar, Jafar, and Viswanath \cite{Nazer} discovered a method
of interference alignment that can code across the two channel states
to recover a single message.  Transmitters send the
same signal in both states, and receivers combine two estimates to
recover the desired message.

Nazer and coauthors named this scheme \defn{ergodic interference
alignment}.  We investigate this further in the next section.

\subsection{\dots over the rational numbers}

Interference alignment by codeword space and by channel state both require a channel
which changes over time (or, equivalent, across the frequency spectrum), while interference
alignment by time requires the existence of time delays.  For some time, this left open
the question of whether a form of interference alignment could be performed over
a static channel without delays.

The question was answered in the positive Motahari, Gharan, Maddah-Ali,
and Khandani, with a scheme they call \emph{real interference alignment} \cite{RealIA}.

The strategy works by effectively `vectorising' the channel.  Specifically, we can
treat the real numbers $\mathbb R$ as a vector space over the rational numbers $\mathbb Q$.
Then we can treat a real signal $x \in \mathbb R$ as a vector $x = \sum_k \lambda_k v_k$, where
$\lambda_k \in \mathbb Q$ and the $v_k \in \mathbb R$ are some basis real numbers.

Using theorems about Diophantine rational approximations to real numbers, Motahari and
coauthors deduce that real interference alignment achieves $n/2$ total degrees of freedom
for the Gaussian interference channel.

Interested readers are referred to the original paper for further details \cite{RealIA}.

\section{Ergodic interference alignment}

It's easiest to analyse ergodic interference alignment by first
looking at the finite field channel.  For convenience, 
we will assume that the
fast-fading coefficients are IID uniform on $\F_q \setminus \{\zero\}$.

Recall from Theorem 1.12 that the single-user capacity of the finite field channel
with non-zero fading is $\mathbb D(Z) := \log q - \mathbb H(Z)$. 

The main lemma that gets ergodic interference alignment to work is
the following \cite[Theorem 1 and Corollary 2]{NazerGastpar}. It is
based on this observation: although receiver $j$ would normally wish
to reconstruct just its own message ${\vec m}_j$, it is, in fact, easier to
reconstruct the `pseudomessage' $\tilde{\vec m}_j := \sum_{i=1}^n H_{ji}{\vec m}_i$.

\addcontentsline{lof}{figure}{\numberline{\textbf{Lemma 2.13}} Reconstructing a pseudomessage}
\begin{lemma}
  Let $\mathcal M = \mathbb F_q^S$.  Consider
  the finite field interference network.  Then
  each receiver $j$ can decode the linear combination of messages
  $\tilde {\vec m}_j = \sum_{i=1}^n H_{ji}\vec m_i$ at rate $\D(Z)$.
\end{lemma}

\begin{proof}
  The key here is for all transmitters to use the same linear code.  Let
  the generator matrix of this code be $\mat G$.  
  Write $\vec g[t]$ for the $t$th row of $\mat G$, so $x_i[t] = \vec g[t] \vec m_i$
  (since $\vec m_i$ is a column vector).
  Then each receiver $j$ sees signal
    \begin{align*}
      Y_j[t] &= \sum_{i=1}^n H_{ji} x_i[t] + Z_j[t] \\
             &= \sum_{i=1}^n H_{ji} (\vec g[t] \vec m_i) + Z_j[t] \\
             &= \vec g[t]  \sum_{i=1}^n H_{ji} \vec m_i + Z_j[t] \\
             &= \vec g[t]  \tilde{\vec m}_j + Z_j[t] .
    \end{align*}
    
  But this is precisely as if the single message $\tilde{\vec m}_i$ was
  sent with the linear code.  Since very good linear codes exist (Theorem 1.15),
  this can be done at rates up to the single-user capacity
  $c = \D(Z)$.
\end{proof}

The technique proceeds as follows: Match a state $\mat H$ with the
\defn{complementary state} $\mat H' = \mat I - \mat H$.
Transmitters send the same signal (encoding the same message) in both states.
Then after $T$ occurrences of the first state receiver $j$ decodes $\tilde m_j = \sum_{i=1}^n H_{ji}m_i$,
at rate $\D(Z)$ and after $T$ occurrences of the complementary state decodes
  \[ \tilde {\vec m}_j' = \sum_{i=1}^n H_{ji}' {\vec m}_j = \sum_{i=1}^n (\delta_{ji} - h_{ji}) {\vec m}_i \]
also at rate $\D(Z)$.
Receiver $j$ then computes the message estimate
  \[ \hat {\vec m}_j = \tilde {\vec m}_j + \tilde {\vec m}_j'
           = \sum_{i=1}^n (h_{ji} + \delta_{ji} - h_{ji}){\vec m}_i
           = \sum_{i=1}^n \delta_{ji} {\vec m}_i = {\vec m}_j ,          \]
as desired.  Since decoding the message required twice the blocklength, the rate
is half what it would be, $(\log 2^{T\D(Z)})/2T = \D(Z)/2$.

(Observe that the receiver needs to perform two separate
estimates, and cannot simply add the channel outputs together.
To do this would lead to a channel of the form $Y = x_j + Z*Z$ , where the convolution $Z*Z$
means the sum of two IID copies of $Z$. The overall rate in this
case is $\D(Z * Z) \leq  \D(Z)/2$, unless $Z$ is deterministic, with strict inequality unless
$\D(Z) = 0$ \cite{JohnsonSuhov}. In general, the $K$-fold convolution has
relative entropy from the uniform $\D(Z*\cdots *Z)$ that usually decreases exponentially in $K$ \cite{JohnsonSuhov}.)

Nazer and coauthors use a typical set argument to show that sufficiently many
channel states can be matched up in this way, showing that with high probability
each matrix and its complement show up almost the same number of times. They use this
to prove the following theorem \cite[Lemma 3 and Theorem 2]{Nazer}:

\addcontentsline{lof}{figure}{\numberline{\textbf{Theorem 2.14}} Achievable rate using ergodic\\ interference alignment (finite field case)}
\begin{theorem}
  For the model as outlined above, the rates
    $r_i = \D(Z)/2$ are simultaneously achievable.
\end{theorem}

Using a quantisation argument, they show the following \cite[Theorem 3]{Nazer}:

\addcontentsline{lof}{figure}{\numberline{\textbf{Theorem 2.15}} Achievable rate using ergodic\\ interference alignment (Gaussian case)}
\begin{theorem}
  For the fast fading Gaussian interference channel with
  symmetric fading (that is, $\mat H$ and $-\mat H$ have the
  same distribution), the rates $r_i = \frac12 \mathbb{E}  \log(1+2\SNR)$
  are simultaneously achievable.
\end{theorem}

We will use this result more in Chapter 4.

\section*{Notes}
\addcontentsline{toc}{section}{Notes}

The section consists of a review of the existing literature; the mathematical
contents is not claimed to be new.

On networks, the textbook of Cover and Thomas \cite[Chapter 15]{CoverThomas} and the review papers of
El~Gamal and Cover \cite{ElGamalCover} and Kramer \cite{KramerFT} were useful.

The first detailed study
of the point-to-point link was by Shannon \cite{Shannon},
of the multiple access network was by Ahlswede \cite{Ahlswede} and Liao \cite{Liao},
of the broadcast network was by Cover \cite{CoverBroadcast},
and of the interference network was by Carleial \cite{Carleial}.

The strategies of interference as noise, decode-and-subtract, and resource division
are old and well-known, making tracking down details
of their discovery difficult.  The textbooks by Tse and Viswanath \cite{TseViswanath}
and by Goldsmith \cite{Goldsmith} give good background on this material.

The concept of interference alignment -- first discovered in the alignment
by codeword choice paradigm -- is due to Cadambe and Jafar \cite{CadambeJafar},
and also to Maddah-Ali, Motahari, and Khandani \cite{MaddahAli}, who independently
discovered a similar method (published in the same issue of the same journal).

The toy example of interference alignment by codeword choice is due to
Jafar \cite{Tutorial}.  The example of interference alignment by time is after
Grokop, Tse, and Yates \cite{GrokopTseYates}.  Ergodic
interference alignment was discovered by Nazer and coauthors \cite{Nazer}.

A tutorial by Jafar \cite{Tutorial} was useful for the material on interference alignment.

\addtocontents{lof}{\protect\addvspace{20 pt}}

%% file: chapters/poisson.tex
\chapter[Regular and Poisson random networks][Regular and Poisson random networks]{Regular and Poisson\\random networks}

In this chapter, we look at two networks based on models of how nodes are positioned in space.  In the \defn{regular network}, nodes are positioned at regular spacings, as in a grid; in the \defn{Poisson random network}, nodes are positioned at random according to a Poisson point process.

We examine how large networks can operate using simple `interference as noise' techniques.  In particular, we show the important relationship between the attenuation $\alpha$ of the signals (which describes how quickly signals die off over distance) and the dimension $d$ of the network.

\section{Model}

We use the model of a Gaussian network with slow fading based on power-law attenuation.
That is we have a countable set of points (\emph{nodes}) placed in $d$-dimensional Euclidean
space $\bR^d$ --  each node $i \in \mathbb Z_+$ is positioned at the point $\vec T_i \in \mathcal \bR^d$.

On the $t$th channel use, the signal received by node $j$  is
  \[ Y_j[t] = \sum_{i \neq j} h \rho(i,j)^{-\alpha/2} x_i[t] + Z_j[t] . \]
Here, $\rho(i,j) = \| \vec T_j - \vec T_i \|$ is the Euclidean distance between nodes $i$ and $j$.
We call $h$ the fixed fading parameter.  Large $h$ corresponds to signals being much more powerful than noise;
small $h$ corresponds to signals being much less powerful than noise.  To concentrate on
the interference-limited regime, we will sometimes consider the limit $h \to \infty$, which
is equivalent to the noiseless network with $h = 1$, 
  \[ Y_j[t] = \sum_{i \neq j} \rho(i,j)^{-\alpha/2} x_i[t]. \]



The Euclidean norm in $\bR^d$ will be denoted $\|\ \|$, where $d$ is the
dimension of the network.  It will be useful later to define
    \[ v(d) := \frac{\pi^{d/2}}{\Gamma(1+d/2)} , \]
which is the volume of the Euclidean unit ball in $\bR^d$.
  
The $d=2$ case is the most commonly studied, as this obviously has real-world applications.  The case $d=3$ has applications in, for example, tall office buildings; and the $d=1$ is attracting more attention for `car-to-car' protocols (see for example recent work by the US Department of Transportation \cite{dot}), where a long road can be modelled as a one-dimensional line.



\section{Regular networks}

In a $d$-dimensional \emph{regular network}, (see for example Xie and Kumar \cite{XieKumar}), nodes
are placed at points of $\bZ^d$ in $d$-dimensional space $\bR^d$.  In particular, any two
nodes are a distance at least $1$ from each other.

Below we show regular networks in one and two dimensions respectively.

\begin{center}
  \begin{picture}(170,40)(0,60)
    \put(10,80){\vector(1,0){150}}
  
    \put(20,80){\circle*{5}}
    \put(50,80){\circle*{5}}
    \put(80,80){\circle*{5}}
    \put(110,80){\circle*{5}}
    \put(140,80){\circle*{5}}
  \end{picture}

\medskip

  \begin{picture}(170,160)(180,0)
   \put(190,80){\vector(1,0){150}}
    \put(260,10){\vector(0,1){150}}
    
    \put(200,80){\circle*{5}}
    \put(230,80){\circle*{5}}
    \put(260,80){\circle*{5}}
    \put(290,80){\circle*{5}}
    \put(320,80){\circle*{5}}
    
    \put(200,50){\circle*{5}}
    \put(230,50){\circle*{5}}
    \put(260,50){\circle*{5}}
    \put(290,50){\circle*{5}}
    \put(320,50){\circle*{5}}
    
    \put(200,20){\circle*{5}}
    \put(230,20){\circle*{5}}
    \put(260,20){\circle*{5}}
    \put(290,20){\circle*{5}}
    \put(320,20){\circle*{5}}
    
    \put(200,110){\circle*{5}}
    \put(230,110){\circle*{5}}
    \put(260,110){\circle*{5}}
    \put(290,110){\circle*{5}}
    \put(320,110){\circle*{5}}
    
    \put(200,140){\circle*{5}}
    \put(230,140){\circle*{5}}
    \put(260,140){\circle*{5}}
    \put(290,140){\circle*{5}}
    \put(320,140){\circle*{5}}
  \end{picture}
\end{center}

Each node $\vec t \in \bZ^d$ is a transmitter, transmitting to one of its $2d$ nearest neighbours,
chosen arbitrarily.  We write $\vec t \to \vec r$ to indicate that node $\vec t$ transmits to node $\vec r$.

All nodes will use standard Gaussian codebooks of power $P=1$, generated independently
of each other.  This means that -- power aside -- signals are statistically indistinguishable
from a) each other, and b) the background noise.  We use the principle that interference
should be treated as noise (as in Section 2.3).

Following Gupta and Kumar \cite{GuptaKumar} the interference at any node $\vec r \in \bZ^d$ 
  \begin{equation}
    I = \mathop{\sum_{\vec t \neq \vec r}}_{\vec t \not\to \vec r}  h \|\vec r- \vec t\|^{-\alpha} ,
  \end{equation}
and using interference-as-noise the communication rate
  \[ r := \log(1+\sinr) = \log{ \left( 1 + \frac{h}{I + 1} \right) } \]
is achievable for each link $\vec r \to \vec t$.

That is, if the interference $I$ is finite, then every node can transmit at this fixed rate.  This is known as
\emph{linear growth} for the following reason: if we have a sequence of sets of nodes $(S_n : n \in \bN)$,
where $S_n \subset \bZ$ is of cardinality $n$, then there exists an achievable
rate $n$-tuple $(r_{i} : i \in S_n)$ such that
  \[ \sum_{\vec j \in S_n} r_{j} = nr = O(n) . \]

The following theorem generalises the work of Xie and Kumar \cite{XieKumar}, who proved it for the case $d=2$.

\addcontentsline{lof}{figure}{\numberline{\textbf{Theorem 3.1}} Linear growth in regular networks}
\begin{theorem} \label{regthm}
  The $d$-dimensional linear network supports linear growth, provided the
  ratio of the attenuation to the dimension of the network is sufficiently
  large, specifically if $\alpha>d$.
\end{theorem}

\begin{proof}
We proceed by induction on the dimension $d$, showing that the interference
$I = I(\alpha, d)$ is finite for $\alpha > d$.  Without loss of generality,
receiver $\vec 0$ is receiving a message from transmitter $\vec t^* := (1,0, \dots, 0)$.

First, the base case, $d=1$. The interference is
  \[ I(\alpha, 1) = h \sum_{t \neq 0,1} |t|^{-\alpha} = h \left(2 \sum_{t=1}^\infty t^{-\alpha} - 1\right), \]
which is finite for $\alpha > 1$, as desired.

The inductive hypothesis is that $I(\alpha,d-1)$ is finite.

Now the inductive step.  Again, the interference is
  \[ I(\alpha, d) = h \sum_{\vec t \neq \vec 0,\vec t^*} \|\vec t\|^{-\alpha}
               \leq h \sum_{\vec t \neq \vec 0} \|\vec t\|^{-\alpha}  . \]
We now split $\bZ^d$ into the $d$ different $(d-1)$-dimensional coordinate spaces (where
at least one coordinate is $0$), and the $2^d$ open orthants (where all
coordinates are nonzero). This gives
  \begin{equation} \label{ind1}
    I(\alpha, d) \leq d I(\alpha, d-1) + h 2^d \sum_{\vec t \in \mathbb N^d} \| \vec t \|^{-\alpha} .
  \end{equation}
The first term in \eqref{ind1} is finite by the inductive hypothesis;
we concentrate on the second term.  By treating the point
$\vec 1 = (1,1,\dots,1)$ separately, we have
  \begin{equation} \label{ind2}
    h 2^d \sum_{\vec t \in \mathbb N^d} \| \vec t \|^{-\alpha}
      = h 2^d d^{-\alpha/2} + h 2^d \sum_{\vec t \in \mathbb N^d \setminus \{ \vec 1 \} } \| \vec t \|^{-\alpha} .
  \end{equation}
The second term in \eqref{ind2} can be approximated by an integral, since for
$\vec t \in \mathbb N^d$,
  \[ \| \vec t \|^{-\alpha} \leq \int_{t_1-1}^{t_1} \cdots \int_{t_d-1}^{t_d} 
                                   \| \vec t \|^{-\alpha} \, \ud t_1 \cdots \ud t_d . \]
Hence,
  \begin{align*}
    h 2^d \sum_{\vec t \in \mathbb N^d \setminus \{ \vec 1 \} } \| \vec t \|^{-\alpha} 
      &\leq h 2^d \int_{1}^{\infty} \cdots \int_{1}^{\infty} \| \vec t \|^{-\alpha} \, \ud t_1 \cdots \ud t_d \\
      &= h \int_{\bR^d \setminus [-1,1]^d} \| \vec t \|^{-\alpha} \, \ud \vec t \\
      &\leq h \int_{\bR^d \setminus B(\vec 0, 1)} \| \vec t \|^{-\alpha} \, \ud \vec t ,
  \end{align*}
where $B(\vec 0,1)$ is the $d$-dimensional unit ball centered at the origin.  We now use a change
of coordinates to $\rho = \| \vec t \|$, $\vec s = \vec t/\rho$, so that
$\ud \vec t = \rho^{d-1} \, \ud \rho \, \ud \vec s$.  This gives
    \begin{align}
    h 2^d \sum_{\vec t \in \mathbb N^d \setminus \{ \vec 1 \} } \| \vec t \|^{-\alpha} 
      &\leq h\int_{\| \vec s \| = 1} \int_{\rho = 1}^\infty \rho^{-\alpha} \rho^{d-1} \, \ud \rho \, \ud \vec s \notag \\
      &= h d v(d) \int_{\rho = 1}^\infty \rho^{-(\alpha-d)-1} \, \ud \rho \notag \\
      &= h d v(d) \left[ - \frac{1}{\alpha - d} \rho^{-(\alpha - d)} \right]_1^\infty \notag \\
      &= h \frac{d v(d)}{\alpha - d} , \label{ind3}
  \end{align}
which is finite.  Putting together \eqref{ind1}, \eqref{ind2}, and \eqref{ind3},  we get
  \[ I(\alpha, d) \leq I(\alpha, d-1) + h2^d d^{-\alpha/2} + h \frac{d v(d)}{\alpha - d} < \infty . \]
The inductive step is complete and the theorem is proven.
\end{proof}

Note that this result is the best possible. (By `best possible', we mean that Theorem \ref{regthm} is not
true for $\alpha \leq d$.  We do not claim our bounds on $I(\alpha,d)$ are as tight as possible.) If
$\alpha \leq d$, then the interference is
  \[  I(\alpha, d) \geq h \sum_{t=1}^\infty t^d t^{-\alpha} \geq h \sum_{t=1}^\infty t^d t^{-d}
          h \sum_{t=1}^\infty 1 = \infty , \]
and the signal-to-interference-plus-noise ratio is $0$.

It is worth noting that a simple bound for $I(\alpha,d)$ is
  \[ I(\alpha,d) \leq h \frac{\alpha}{\alpha-d} 2^{d-1} (d+1)! . \]
This can be proved inductively, using $v(d) \leq 2^d$ (that is, the volume of the unit sphere is less than
that of the surrounding cube).

\section{Poisson random networks}

In this section, we define the Poisson random network, where nodes are distributed like a Poisson
process.  We give a local result -- bounds on the outage probability of a single transmission -- and
a global result -- showing that linear growth occurs in the network with high probability.

\subsection{Node positioning model}

In a $d$-dimensional Poisson random network (studied extensively by Haenggi \cite{Haenggi1, Haenggi2}, and Dhillon, Ganti, and Andrews \cite{jeffreys},  among others), the set of nodes $\{\vec T_i : i \in \bZ_+\}$ are placed
in $\bR^d$ as a Poisson point process of density $1$ (without loss of generality).  For simplicity,
we will translate the points such that $\vec T_0$ is at the origin $\vec 0$, and relabel the nodes
by order of distance from the origin.  So we have $0 = \|\vec T_0\| \leq \|\vec T_1\| \leq \cdots$ (and
strict inequalities almost surely).

The figures below shows Poisson random networks in one and two dimensions.

		\begin{center}
\begin{picture}(170,40)(0,60)
    \put(10,80){\vector(1,0){150}}
  
    \put(132,80){\circle*{5}}
    \put(61,80){\circle*{5}}
    \put(80,80){\circle*{5}}
    \put(55,80){\circle*{5}}
    \put(111,80){\circle*{5}}
    \put(28,80){\circle*{5}}
  \end{picture}
		
	  \begin{picture}(170,160)(180,0)
	  
	  \put(190,80){\vector(1,0){150}}
    \put(260,10){\vector(0,1){150}}
    
    \put(232,136){\circle*{5}}
    \put(261,51){\circle*{5}}
    \put(221,61){\circle*{5}}
    \put(295,118){\circle*{5}}
    \put(281,93){\circle*{5}}
   
    \put(215,21){\circle*{5}}
    \put(260,80){\circle*{5}}
    \put(252,72){\circle*{5}}
    \put(302,37){\circle*{5}}
    \put(248,22){\circle*{5}}
    
    \put(274,55){\circle*{5}}
    \put(260,102){\circle*{5}}
    \put(209,106){\circle*{5}}
    \put(226,74){\circle*{5}}
    \put(296,70){\circle*{5}}
    
    \put(316,119){\circle*{5}}
    \put(287,55){\circle*{5}}
    \put(203,44){\circle*{5}}
    \put(204,100){\circle*{5}}
    \put(297,39){\circle*{5}}
    
    \put(269,28){\circle*{5}}
    \put(229,93){\circle*{5}}
    \put(260,63){\circle*{5}}
    \put(268,70){\circle*{5}}
    \put(255,26){\circle*{5}}
	  
	  \end{picture}
	  \end{center}

We want to model the scenario of a multihop network; that is, where messages are sent to distant nodes by
being successively passed over short distances by a number of intermediate nodes.  For this model, we shall
assume that each node broadcasts its intended signal, whilst picking up the signal from its nearest
neighbour.  We will concentrate on the communication over just these short hops -- the large-scale
strategy has been studied by many others, particularly the multihop strategy of Gupta and Kumar \cite{GuptaKumar} and the work on hierarchical cooperation
by \"Ozg\"ur, L\'ev\^eque, and Tse \cite{Ozgur}.

If node $i$'s nearest neighbour is node $j$, we again write $j \to i$.  In particular, the
nearest neighbour to node $0$, at the origin, is node $1$ at
$\vec T_1$, so $1 \to 0$.

All nodes will use standard Gaussian codebooks of power $P=1$, generated independently of each other.  Nodes treat
signals other than they are picking up as Gaussian noise, as discussed in Section 2.3.

\subsection{Outage probability}

Given the link $j \to i$ and the positions of nodes $i$ and $j$, the position of other nodes in the
network is random.  In particular, for any given rate $r > 0$, we cannot guarantee that $j$ can communicate
to $i$ at rate $r$.  This is because the other nodes might (with non-zero probability) crowd round $i$,
drowning out the intended signal from $j$.

Hence, we need to study the \emph{outage probability} of the network (as discussed previously
in Subsection 1.5.3).

\addcontentsline{lof}{figure}{\numberline{\textbf{Definition 3.2}} Outage, $\SINR$, interference}
\begin{definition}
We define the outage probability as
  \[ p_{\text{out}}(r) = \mathbb{P} \left( r > \log{(1+\SINR_{ji})}  \right) . \]
Here (and throughout) $\mathbb{P}$ denotes probability over the Poisson point process.

The signal-to-interference-plus-noise ratio at node $j$, $\SINR_{ji}$, has marginal distribution function
  \[ F_{\SINR}(s) = \mathbb{P} \left( \frac{h \|\vec T_j - \vec T_i\|^{-\alpha}}{I_{ji} + 1} \leq s \right) , \]
independent of the link $i\to j$ where the interference is
  \[ I_{ji} = \sum_{k \neq i,j} h \|\vec T_j-\vec T_k\|^{-\alpha} .\]
\end{definition}


From now on, we deal with the link $1 \to 0$, without loss of generality.  We will suppress subscripts that
are no longer necessary.

An important special case is the high power regime $h \to \infty$, which is simpler to deal with mathematically and
is important for studying networks that are interference limited rather than noise limited.  Note that since
  \begin{align*} \lim_{h \to \infty} \SINR &= \lim_{h \to \infty} \frac{h \|\vec T_j - \vec T_i\|^{-\alpha}}{\sum_{k \neq i,j} h \|\vec T_j-\vec T_k\|^{-\alpha} + 1} \\
                                           &= \frac{ \|\vec T_j - \vec T_i\|^{-\alpha}}{\sum_{k \neq i,j}  \|\vec T_j-\vec T_k\|^{-\alpha} } , \end{align*}
this is equivalent to taking $h=1$ and ignoring the noise term.

The next two theorems give upper (Theorem \ref{outagethm}) and lower (Theorem 3.4) bounds on the outage probability
$\pout(r)$.  The figure below shows these bounds for the common case $d=2$, $\alpha = 3$ in the
high power $h \to \infty$ regime.

	\begin{center}
       \includegraphics[scale=0.74]{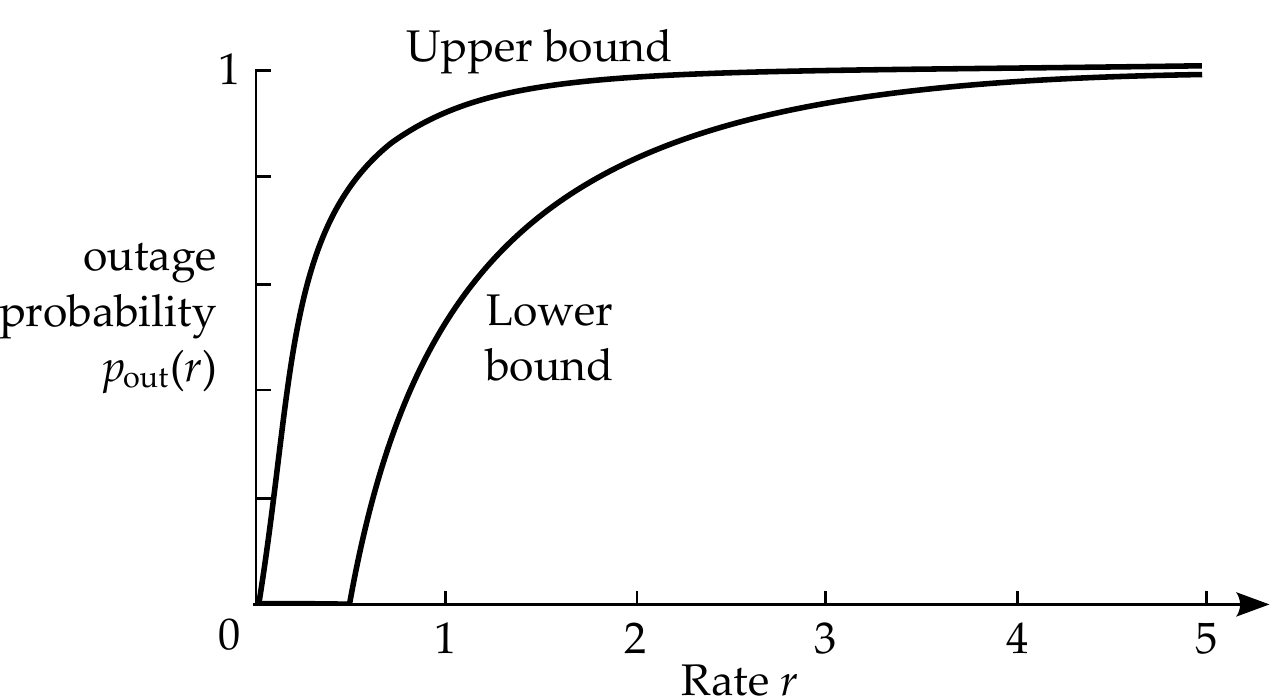}
	\end{center}

\addcontentsline{lof}{figure}{\numberline{\textbf{Theorem 3.3}} Outage probability in Poisson\\ random networks: upper bound}
\begin{theorem}\label{outagethm}
Consider a $d$-dimensional Poisson random network with fixed fading parameter $h$, and
suppose $\alpha > d$.  Then the outage probability is upper-bounded by
  \begin{multline*}
    p_{\text{out}}(r) \leq \frac{d(2^{r}-1)}{\alpha-d} \left(1 + 2\frac{2^r-1}{h} v(d)^{-\alpha/d}
      \,\Gamma\left(2+\frac{\alpha}{d} \right) \right) \\
              + \exp \left(-v(d) \left(2\frac{2^r-1}{h}\right)^{-d/\alpha}\right) .
  \end{multline*}
In the high power $h \to \infty$ regime, we have
  \[ p_{\text{out}}(r) \leq \frac{d(2^{r}-1)}{\alpha-d} \left( 1 - \exp{\left(-\frac{\alpha - d}{d(2^{r}-1)}\right)} \right) . \]
\end{theorem}

Note that since
\[ \frac{d}{\alpha-d} = \frac{1}{\alpha/d - 1} , \]
the attentuation $\alpha$ and dimension $d$ only enter this theorem through their ratio.

\begin{proof}
 First, note that the outage probability can be rewritten as
    \[
      p_{\text{out}}(R) = \mathbb{P} \left( r > \log{(1+\SINR)}  \right) 
                        = F_{\SINR}(s) , 
    \]
  where we have defined $s := 2^{r}-1$.  So now we need only bound $F_{\SINR}(s)$, the
  probability that $\SINR$ is large.

  We will bound $F_{\SINR}(s)$ by conditioning on the position of the nearest neighbour.  So
    \begin{align} 
      F_{\SINR}(s) &= \Prob{(\SINR \leq s)} \notag \\
             &= \int_0^\infty f_{\|\vec T_1\|}(t)
                       \Prob{\big( \SNR \leq s \,\big|\, \|\vec T_1 \| = t \big)}\,\ud t \notag \\
             &= \int_0^\infty f_{\|\vec T_1\|}(t)
                      \Prob {\left( I \geq \frac{1}{s} ht^{-\alpha} - 1 \,\Big|\, \|\vec T_1 \| = t \right)} \,\ud t . \label{FSNR}
    \end{align}

  It is known \cite[Theorem 1]{Haenggi3} (and can be easily shown) that $\|\vec T_1\|^d$
  has an exponential distribution with parameter $v(d)$,
    \[ \|\vec T_1\|^d \sim \text{Exp}\big(v(d)\big) , \]
  giving
    \begin{equation}\label{Exp}
      f_{\|\vec T_1\|}(t) = d v(d) t^{d-1} \ex^{-v(d) t^d} . 
    \end{equation}
  
  We will also need to bound the probability that the interference is large, which we do using the
  conditional Markov inequality.  Specifically, when $\frac 1s h t^{-\alpha} - 1 > 0$, we have
    \begin{equation} \label{back}
      \Prob {\left( I \geq \frac{1}{s} ht^{-\alpha} - 1 \,\Big|\, \|\vec T_1 \| = t \right)} 
        \leq \min \left\{ \frac{1}{\frac{1}{s} ht^{-\alpha} - 1}
          \Ex \big( I \,\big|\, \|\vec T_1\| = t \big) , 1 \right\} , 
    \end{equation}
  and when $\frac 1s h t^{-\alpha} - 1 \leq 0$, we take the trivial bound $1$ instead.
  
  We concentrate on the simpler high power scenario first.  Note that in the limit $h \to \infty$,
  we always have $\frac 1s h t^{-\alpha} - 1 > 0$.
  
    For the moment, we concentrate on the first argument. We can write
   \[  \Ex \big( I \,\big|\, \|\vec T_1\| = t \big)
             = \Ex  \sum_{\vec T \in \mathcal P(t)} h\|\vec T\|^{-\alpha} , \]
  where $\mathcal P(t)$ is the set of points of the Poisson process outside the ball
  of radius $t$ about the origin.  Using Campbell's theorem (see for example the monograph
  of Haenggi and Ganti \cite[Theorem A.2]{Haenggi1}), we then have
   \begin{align} 
     \Ex  \sum_{\vec T \in \mathcal P(t)} h\|\vec T\|^{-\alpha}
        &= \int_{\|\vec u\|\geq t} h\|\vec u\|^{-\alpha} \, \ud \vec u \notag \\
        &= hd v(d) \left[ - \frac{1}{\alpha - d} \rho^{-(\alpha-d)} \right]_{\rho = t}^\infty \notag \\
        &= h\frac{d v(d)}{\alpha - d} t^{-(\alpha - d)} . \label{int}
  \end{align}

  Substituting (\ref{int}) and (\ref{Exp}) this back into equation (\ref{FSNR}) gives
    \begin{align*}
      F_{\SINR}(s) 
        &\leq \lim_{h\to\infty} \int_0^\infty d v(d) t^{d-1} \ex^{-v(d) t^d}
                      \min \left\{\frac{1}{ \frac{1}{s} ht^{-\alpha} - 1}
                        h \frac{d v(d)}{\alpha - d} t^{d-\alpha} , 1 \right\} \, \ud t \\
        &= \int_0^\infty d v(d) t^{d-1} \ex^{-v(d) t^d}
                      \min \left\{\frac{1}{ \frac{1}{s} t^{-\alpha}}
                        \frac{d v(d)}{\alpha - d} t^{d-\alpha} , 1 \right\} \, \ud t \\                
        &= \int_0^{t^*} d v(d) t^{d-1} \ex^{-v(d) t^d}
                      st^{\alpha} \frac{d v(d)}{\alpha - d} t^{d-\alpha} \, \ud t \\
        &\qquad\qquad\qquad\qquad\qquad\qquad\qquad\quad {}+ \int_{t^*}^\infty dv(d) t^{d-1} \ex^{-v(d) t^d}
                      st^{\alpha} \, \ud t ,
    \end{align*}  
  where
    \[ t^* = \left( s \frac{d v(d)}{\alpha - d} \right)^{-1/d} \]
  is the point where we cross over from one argument of the `$\min$' to the other.
  
  Calculating these integrals using a substitution $y = v(d) x^d$ gives
    \[ F_{\SINR}(s) \leq \frac{ds}{\alpha-d} \left( 1 - \ex^{-(\alpha - d)/ds} \right) . \]
  This proves the theorem for the high power regime.
  
  We now move back to the general case.  Returning to \eqref{back}, we now have
    \begin{align*}
      &\int_0^\infty f_{\|\vec T_1\|}(t)
                      \Prob {\left( I \geq \frac{1}{s} ht^{-\alpha} - 1 \,\Big|\, \|\vec T_1 \| = t \right)} \,\ud t \\
        &\qquad\quad\leq \int_0^{t^{**}} dv(d) t^{d-1} \ex^{-v(d)t^d}
                       \min \left\{ \frac{1}{\frac1s ht^{-\alpha}-1} h \frac{d v(d)}{\alpha - d} t^{d-\alpha} , 1  
                          \right\} \, \ud t \\
        &\qquad\qquad\qquad\qquad\qquad\qquad\qquad\qquad\qquad\qquad
                 {}+ \int_{t^{**}}^\infty dv(d) t^{d-1} \ex^{-v(d)t^d}  \, \ud t ,
    \end{align*}
  where $t^{**} = (s/h)^{-1/\alpha}$ is the point at which $\frac 1s ht^{-\alpha}-1=0$.
  
  We could evaluate this integral numerically, or express it as a complicated sum of Gamma functions.  Instead,
  we will show a simple bound.
  
  When $t \leq 2^{-1/\alpha} t^{**}$, we have $\frac sht^\alpha \leq \frac12$,
  and hence
    \[ \frac{1}{ \frac1s ht^{-\alpha} - 1 } h
         = \frac{st^\alpha}{1 - \frac sh t^\alpha} \leq st^\alpha \left(1 + 2\frac sh t^\alpha \right) . \]
  This allows us to bound the integral by
     \begin{align*}
       &\int_0^{2^{-1/\alpha}t^{**}} dv(d) t^{d-1} \ex^{-v(d)t^d}
           s \left(1 + 2\frac sh t^\alpha \right) \frac{dv(d)}{\alpha-d} t^d  \, \ud t \\
         &\qquad{}+ \int_{2^{-1/\alpha}t^{**}}^\infty dv(d) t^{d-1} \ex^{-v(d)t^d} \, \ud t \\
         & \qquad\quad\quad\ \leq \frac{ds}{\alpha-d} \left(1 + 2\frac sh v(d)^{-\alpha/d} \Gamma \left(2+\frac{\alpha}{d} \right) \right) 
              + \exp \left( -v(d) \left(2\frac sh \right)^{-d/\alpha} \right) ,
    \end{align*}
  as required, where the first term follows upon replacing the upper limit of the integral by $\infty$.  
\end{proof}

\addcontentsline{lof}{figure}{\numberline{\textbf{Theorem 3.4}} Outage probability in Poisson\\ random networks: lower bound} 
\begin{theorem}
Under the same conditions as Theorem \ref{outagethm}, the
outage probability is bounded below by
\[ p_{\text{out}}(r)  \geq 1 - \frac{1}{ (2^{r} - 1)^{d/\alpha}} .
\]
\end{theorem}
\begin{proof}
The key is to observe that the interference $I$ is at least as
large as the contribution coming from the second-nearest neighbour
$\vec T_2$. That is,
\begin{align*}
  \Prob {\left( I \geq \frac 1sh t^{-\alpha} + 1 \,\Big|\, \|\vec T_1 \| = t \right)} 
    &\geq \Prob {\left( I \geq \frac 1sh t^{-\alpha}  \,\Big|\, \|\vec T_1 \| = t \right)} \\
    &\geq \Prob {\left( h \| \vec T_2 \|^{-\alpha} \geq \frac 1sh t^{-\alpha} \,\Big|\, \|\vec T_1 \| = t \right)}.
\end{align*}
Rearranging this, we get
\begin{align*}
  \Prob {\left( h \| \vec T_2 \|^{-\alpha} \geq \frac 1sh t^{-\alpha} \,\Big|\, \|\vec T_1 \| = t \right)}
    &=\Prob {\left( \| \vec T_2 \|^{-\alpha} \geq \frac1s t^{-\alpha} \,\Big|\, \|\vec T_1 \| = t \right)} \\
    &= 1 - \Prob {\left( \| \vec T_2 \| > t s^{1/\alpha} \,\Big|\, \|\vec T_1 \| = t \right)} \\
    &= 1 - \exp\Big(-v(d) \big( (t s^{1/\alpha})^d - t^d\big)\Big),
\end{align*}
where the final result follows on considering the probability that the
annulus $\{ \vec u \in \bR^d : t < \|\vec u\| \leq t s^{1/\alpha} \}$
is empty.

Again, combining this with Equations (\ref{FSNR}) and (\ref{Exp}),
we obtain a lower bound on the outage probability of the form
\begin{multline*}
  \int_0^\infty dv(d) t^{d-1} \ex^{-v(d)t^d} \left( 1 - \ex^{-v(d) t^d (s^{d/\alpha} - 1)} \right) \ud t  \\
  = \int_0^\infty dv(d) x^{d-1} \ex^{-v(d)t^d} dt - \int_0^\infty dv(d) t^{d-1} \ex^{-v(d)t^d s^{d/\alpha}} \ud t,
\end{multline*}
and the result follows on making the change of variables $y = v(d) t^d$.
\end{proof}   

\subsection{Linear growth}

Recall that we say that \defn{linear growth} occurs when the sum-rate of $n$ users scales linearly with $n$.
That is, if we have a sequence of sets of nodes $\big( S_n : n \in \bN \big)$,
where $S_n \subset \bZ_+$ is of cardinality $n$, then there exists an achievable
rate $n$-tuple $(r_{i\to j} : i \in S_n)$ such that $ \sum_{\vec i \in S_n} r_{i\to j} =  O(n)$.

In particular if a `proportion' $p$ of links were to support a given rate $r$, then we would have
  \[ \sum_{\vec i \in S_n} r_{i\to j} \approx pnr = O(n) , \]
which would be sufficient to show linear growth.

In particular, if communication
on distinct links were independent, this would be true with
$p = 1 - \pout(r)$. Although the links are not independent, links
`far enough away' are.  Thus by splitting the network into `close' and
`distant' nodes we can prove the theorem.

\addcontentsline{lof}{figure}{\numberline{\textbf{Theorem 3.5}} Linear growth in Poisson random networks}
\begin{theorem}\label{growth}
Consider a $d$-dimensional Poisson random network with attenuation $\alpha$
in the interference-limited $h\to\infty$ regime.  Suppose $\alpha > d$.  Then we have linear growth with
probability tending to $1$ as $n \to \infty$ at rate $O(n^{-(1-d/\alpha)})$.
\end{theorem}

The following Chernoff-type bound on Poisson random variables will be useful later.

\addcontentsline{lof}{figure}{\numberline{\textbf{Lemma 3.6}} Chernoff bound for Poisson random variables}
\begin{lemma}\label{lem:tech}
Let $X \sim \text{Po}(\lambda)$. Then $\Prob{(X \geq \ex \lambda)} \leq \ex^{-\lambda}$.
\end{lemma}

\begin{proof}
  We have, using Markov's inequality,
    \[ \Prob{(X \geq \ex \lambda)} = \Prob{(\ex^X \geq \ex^{\ex \lambda})}
      \leq \frac{\Ex \ex^X}{\ex^{\ex \lambda}} = \frac{ \ex^{(\ex -1)\lambda}}{\ex^{\ex \lambda}} = \ex^{-\lambda}. \qedhere \]
\end{proof}

We are now in the position to prove Theorem \ref{growth}.

\begin{proof}[Proof of Theorem \ref{growth}]
  First recall that $h\to\infty$ is equivalent to taking $h=1$ an
  deleteing the noise term.

  Let $n$ be some fixed integer.  (Later, we will consider the sum rate of $n$ nodes, and let $n \to \infty$.)
  Fix the nearest-neighbour link $i\!\to\!j$.
  
  Given a node $j$ situated at the point $\vec T_j$, we divide the other nodes into those \emph{close} to $j$
    \[ C(j) := \{ k \neq i, j : \| \vec T_j - \vec T_k \| \leq n^a \}  \]
  and those \emph{distant} from $j$
    \[ D(j) := \{ k \neq i : \| \vec T_j - \vec T_k \| > n^a \} , \]
  where $a$ is some parameter to be chosen later.
  
  Outage occurs when the signal-to-interference-plus-noise ratio is insufficient to support some given rate
  $r = \log{(1+s)}$.  We will consider whether an outage event is caused primarily by
  close or distant nodes.  Specifically, we define the events
    \begin{align*}
      \Out_C(j) &:= \left\{  \frac{\| \vec T_j - \vec T_i \|^{-\alpha}}
                      {\sum_{k\in C(j)}\| \vec T_j - \vec T_k  \|^{-\alpha}} \leq \frac{s}{2} \right\} , \\
      \Out_D(j) &:= \left\{  \frac{\| \vec T_j - \vec T_i \|^{-\alpha}}
                      {\sum_{k\in D(j)}\| \vec T_j - \vec T_k  \|^{-\alpha}} \leq \frac{s}{2} \right\} .
    \end{align*}
  Note that if a link $i\!\to\!j$ is in outage, at least one of $\Out_C(j)$ and $\Out_D(j)$ must be occurring.  Note
  also that the (marginal) distributions of $\Out_C(j)$ and $\Out_D(j)$ are independent of the node $j$.  When the
  node index is irrelevant, we suppress it.

  Suppose $n$ nodes all try to communicate at the rate $r = \log{(1+s)}$.  Then
  for $\epsilon \in (0,\frac12)$,
    \begin{align}
      &\Prob{(\text{total rate}\leq (1-2\epsilon)nr)} \nonumber \\
        &\quad {}\leq  \Prob{(\text{number of outages}\geq 2\epsilon n)} \nonumber \\
        &\quad {}= \Prob{(\text{number of distant outages}\geq \epsilon n}
                      \cup \text{number of close outages}\geq \epsilon n) \nonumber \\
        &\quad {}\leq \Prob{(\text{number of distant outages}\geq \epsilon n)} 
                      + \Prob{( \text{number of close outages}\geq \epsilon n)} \nonumber \\
        &\quad {}= \Prob{\left( \sum_{i=0}^{n-1} \mathbb{1}[\Out_D(j)] \geq \epsilon n \right)} 
                      + \Prob{\left( \sum_{i=0}^{n-1} \mathbb{1}[\Out_C(j)] \geq \epsilon n \right)} . \label{here}
    \end{align}  
  We bound the two terms above separately.
  
  For the first term in \eqref{here}, by Markov's inequality, we have
    \[
      \Prob{\left( \sum_{j=0}^{n-1} \mathbb{1}[\Out_D(j)] \geq \epsilon n \right)}
        \leq \frac{ \Ex \sum_{j=0}^{n-1} \mathbb{1}[\Out_D(j)]}{\epsilon n}
        = \frac{1}{\epsilon} \Prob{\big(\Out_D\big)}.
    \]
  Using the same ideas as in the proof of Theorem \ref{outagethm}, we can show that this term tends to zero for $a >0$ at rate
  $O(n^{-a(\alpha-d)})$.  (The key is to change the lower limit in the integral in \eqref{int} to $n^a$.)

  We bound the second term in \eqref{here} using the idea that 
 for most $j$ and $k$, $C(j)$ and $C(k)$ are disjoint, so the
corresponding contributions will be independent.

Let $N_j$ denote the number of nodes whose `close' regions overlap with $j$'s `close' region; that is
  \[  N_j := \# \{ k \neq j : C(j) \cup C(k) \neq \varnothing \} . \]
  Note that $N_j$ is the number of nodes in a ball of radius
$2n^a$, so is Poisson with mean $2^d v(d) n^{ad}$.  
  
We write 
$\N = \{ \max_{0 \leq j \leq n-1} N_j \geq 2^d  v(d) n^{ad} \ee \}$ for the event that
one of the $N_j$ is particularly large. We will argue that the
event $\N$ can be ruled out with high probability, allowing 
good control of the growth of the variance.

We exploit 
the fact that, conditioned on the event $\N^c$, for any
$j$ there
are at most $1+2^d v(d) n^{ad} \ee$ indices $k$ such that
$\Cov{\big(\mathbb{1}[\Out_C(j)], \mathbb{1}[\Out_C(k)] 
\given  \N^c \big)}$ is non-zero,
and each such covariance is no greater than $1$.
Hence we can control the growth of the variance of the sum as
\begin{align}
  \Var \left( \left. \sum_{j=0}^{n-1} \mathbb{1}[\Out_C(j)] \,\right|\, \N^c \right) 
    &\leq \sum_{j=0}^{n-1} (1 + 2^d v(d) n^{a d} \ee) \var \left( \mathbb{1}[\Out_C(j)] \,|\, \N^c \right) \notag \\
    &\leq (1 + 2^d v(d) n^{a d} \ee) n. \label{eq:varbd}
\end{align}

By the union bound and Lemma 
\ref{lem:tech}, we have
   \begin{equation}
     \Prob( \N) \leq \sum_{j=0}^{n-1} \Prob (N_j \geq 2^d v(d) n^{ad}) \leq n \ex^{-2^d v(d) n^{ad}} \label{eq:chernoff} .
   \end{equation}
  
Using the law of total probability, and substituting (\ref{eq:varbd}) and
(\ref{eq:chernoff}) into  Chebyshev's inequality gives
  \begin{align*}
    \Prob{\left( \sum_{j=0}^{n-1} \mathbb{1}[\Out_C(j)] \geq \epsilon n \right)}  
      &= \Prob{(\N)} \Prob{\left( \left. \sum_{j=0}^{n-1} \mathbb{1}[\Out_C(j)] \geq \epsilon n
             \,\right|\, \N \right)} \\
           & \qquad\qquad  {}+ \Prob{(\N^c)} \Prob{\left( \left. \sum_{j=0}^{n-1} \mathbb{1}[\Out_C(j)] \geq \epsilon n 
             \, \right|\, \N^c \right)} \\
      &\leq n \ex^{-2^d v(d) n^{ad}} + \Prob{\left( \left. \sum_{j=0}^{n-1} \mathbb{1}[\Out_C(j)] \geq \epsilon n 
             \,\right|\, \N^c \right)}
  \end{align*}
  
  \begin{align*}
  \phantom{\Prob{\left( \sum_{j=0}^{n-1} \mathbb{1}[\Out_C(j)] \geq ik \right)} }
      &\leq n \ex^{-2^d v(d) n^{ad}} +
             \frac{\Var \left( \sum_{j=0}^{n-1} \mathbb{1}[\Out_C(j)] \given \N^c \right)}
             {n^2 \big( \epsilon - \Ex (\mathbb{1}[\Out_C] \,\big|\, \N^c) \big)^2} \\
      &\leq n \ex^{-2^d v(d) n^{ad}} +
             \frac{ (1 + 2^d v(d) n^{a d} \ex) n}
             {n^2 \big( \epsilon - \Prob{ (\Out_C \given \N^c})  \big)^2}.
  \end{align*}
  As we send $n \to \infty$, this tends to $0$ at rate $O(n^{ad-1})$.
  
  Setting $a = 1/\alpha$ means the two terms both tend to $0$ at rate
    \[ ad - 1 = -a(\alpha - d) = - \left( 1 - \frac{d}{\alpha} \right) , \]
  as desired.
 
\end{proof}

\section{Further work}

The strength of a Poisson model is the extensibility and mathematical flexibility of the Poisson
point process.  

For just one example, in a spatially inhomogeneous Poisson process, there is a density function
$\lambda \colon \bR^d \to \bR_+$, such that the number of nodes in some region $A$ is
$\text{Po} \left(\int_A \lambda (\vec x)\,\ud\vec x\right)$.  Under what conditions on $\lambda$
is linear growth still possible?

Many other open questions remain.

\section*{Notes}
\addcontentsline{toc}{section}{Notes}

The new work in this chapter is joint work with Oliver Johnson and Robert Piechocki.
It has not previously been published, although the research was conducted prior to some since-published work
\cite{Bacelli, Haenggi1}

The monographs of Baccelli and B\l aszczyszyn \cite{Bacelli} and Haenggi and Ganti \cite{Haenggi1} were
useful background on the SINR approach to stochastic networks, as was the work of
Haenggi \cite{Haenggi2}, Xie and Kumar
\cite{XieKumar}, and Gupta and Kumar \cite{GuptaKumar}.

The textbook of Kingman \cite{Kingman} was useful background material for Poisson processes.

\addtocontents{lof}{\protect\addvspace{20 pt}}

%% file: chapters/sumcapacity.tex
\chapter[Sum-capacity of random dense Gaussian interference networks][Sum-capacity of random Gaussian interference networks]{Sum-capacity of\\random dense Gaussian\\interference networks}

In this chapter, we will find approximations to the sum-capacity of interference networks
with many users. Our interference networks will be motivated by physical models of wireless
networks.

When we say many users, we will be allowing the number of users $n$ to tend to infinity,
and looking at asymptotic behaviour.  We will be doing this without expanding
the area in which the nodes reside, so the nodes will get packed closer and closer
together -- for this reason, such networks are called \defn{dense} networks.

We will model transmitters and receivers as being placed randomly in space. This
means that signal- and interference-to-noise ratios will be random too, but fixed for the
duration of communication -- and hence a form of slow fading.

The crucial insight to prove these results is that the performance in these networks
is tightly constrained by the performance on a few so-called bottleneck links.
Ideas from interference alignment (specifically ergodic interference alignment) are
crucial in performing well on these links, and hence in the whole network.


The main result (Theorem 4.3) will be to show, for the model we will consider,
that the sum-capacity $C_\Sigma$ is roughly $\frac n2 \mathbb E \log(1 + 2\SNR)$. Specifically,
we show that $C_\Sigma / n$ converges in probability to 
$\frac12 \mathbb E \log(1 + 2\SNR)$ as $n \to \infty$.

In this chapter, random positions will give a form of slow fading, and we will
also add fast fading, to make a realistic model of real-world networks.  All
results in this chapter hold provided that this fast fading is symmetrical (in the sense
that $H_{ji}$ and $-H_{ji}$ are identically distributed, for all $i$ and $j$).

For simplicity and concreteness, we will choose to give all our results in
the context of random phase fading, so
  \[ H_{ji}[t] = |H_{ji}| \exp(\ii \Theta_{ji}[t]) = \begin{cases}
       \sqrt{\SNR_i} \exp(\ii \Theta_{ji}[t]) & \text{for $i = j$,} \\
       \sqrt{\INR_{ji}} \exp(\ii \Theta_{ji}[t]) & \text{for $i \neq j$.} \end{cases} \]
In Section 4.6, we briefly discuss how to apply other fading models such as Rayleigh fading.

\section{Introduction}

Recently, progress has been made on many-user approximations to the sum-capacity
$\Csum$ of random Gaussian interference networks.

In particular, in a 2009 paper, Jafar \cite[Theorem 5]{Jafar} proved a result on the
asymptotic sum-capacity of a particular random Gaussian interference network:

\addcontentsline{lof}{figure}{\numberline{\textbf{Theorem 4.1}} Asymptotic sum-capacity of the Jafar network}
\begin{theorem}
  Suppose direct $\SNR$s are fixed and identical, so $\SNR_i = \snr$ for all $i$,
  and suppose that all $\INR$s are IID random and supported on some neighbourhood
  of $\snr$.  Then the average per-user capacity $\Csum/n$ tends in probability
  to $\frac12 \log(1+2\snr)$ as $n\to\infty$.
\end{theorem}

We examine Jafar's result in detail later.

(Here and elsewhere, we use $\Csum$ to denote the sum-capacity of the network,
and interpret $\Csum/n$ as the average per-user capacity.)

A subsequent result by Johnson, Aldridge and Piechocki \cite[Theorem 4.1]{JAP} concerned a more
physically realistic model, the standard dense network:

\addcontentsline{lof}{figure}{\numberline{\textbf{Theorem 4.2}} Asymptotic sum-capacity of\\ the standard dense network}
\begin{theorem}
  Suppose receivers and transmitters are placed IID uniformly at random on the unit square $[0,1]^2$,
  and suppose that signal power attenuates like a polynomial in $1/\mathrm{distance}$.
  Then the average per-user capacity $\Csum/n$ tends in probability
  to $\frac12 \EE \log(1+2\SNR)$ as $n\to\infty$.
\end{theorem}


In this chapter, we prove a similar, but more general, result to Theorem 4.2, with a neater proof, using ideas
from Jafar's proof of Theorem 4.1.  We assume transmitters and receivers are situated independently at random
in space (not necessarily uniformly), and that the power of a signal depends in a natural way on the distance it travels.

Specifically our result is the following (full definitions of non-italicised technical terms are in Section 4.2):

\addcontentsline{lof}{figure}{\numberline{\textbf{Theorem 4.3}} Asymptotic sum-capacity of IID networks}
\begin{theorem} \label{thm:main}
  Consider a Gaussian interference network formed by $n$ pairs of 
  nodes placed in a \emph{spatially separated IID network} with \emph{power-law attenuation}.
  Then the average per-user capacity $\Csum/n$ converges in probability to 
  $\frac12 \EE \log(1 + 2 \SNR)$, in that for all $\epsilon > 0$
    \[  \PP \left( \left| \frac{\Csum}{n} - \frac12 \EE \log(1 + 2 \SNR) \right| > \epsilon \right)
          \to 0 \quad \text{as $n \to \infty$.} \]
\end{theorem}

The direct part of the proof uses interference alignment; specifically, we take advantage of ergodic interference alignment (see Section 2.6).

The converse part of the proof uses the idea of `bottleneck links' developed by Jafar
\cite{Jafar}. An information theoretic argument gives a capacity bound on such bottleneck
links, and probabilistic counting arguments show there are sufficiently many such links
to tightly bound the sum-capacity of the whole network.

Before going any further, we should mention two similar results in the same area using different
interference alignment techniques.

A paper by \"Ozg\"ur and Tse \cite{OzgurTse} proves linear scaling in interference networks
by showing that for any value of $\snr$ and $n$, the sum-rate is bounded below by
$k_1 n \log (1+k_2 \snr)$, where $k_1, k_2 >0$ are universal constants.

Secondly, a paper by Niesen \cite{Niesen} bounds the capacity region (rather than just the
sum-capacity) of arbitrary dense networks, albeit with a factor of $O(\log n)$ separating
the inner and outer bounds.

\section{Model}

We outline the model we will use. We will model separately how messages
are transmitted, and how nodes are positioned.

These ideas were introduced in an earlier paper \cite{JAP}, but were not
fully exploited, due to that paper's concentration on the standard dense network.

\subsection{Communication model}

We will use the $n$-user Gaussian interference network as our main model.  That is, we have
\[ Y_j[t] = \sum_{i=1}^n H_{ji}[t] x_i[t] + Z_j[t] . \]

The fading coefficients $H_{ji}[t]$ will be made up of a fast fading part, representing
the moment-to-moment changes in the channel, and a slow fading part, the power attenuation
due to node placing.

The results require the fast fading to be symmetric, in that $H_{ji}$ and $-H_{ji}$ are identically
distributed -- this is ensured by the random phase we use (and is also satisfied by Rayleigh fading).
For ease of notation and to simplify exposition, will assume the fast fading takes the form of a (uniformly)
random phase change (although we briefly discuss more general models in Section 4.6).
Therefore, we can write the fading coefficients in modulus--argument form as
  \[ H_{ii} = \exp(\ii \Theta_{ii}[t]) \sqrt{\SNR_i}, \qquad H_{ji} = \exp(\ii \Theta_{ji}[t]) \sqrt{\INR_{ji}} \quad j\neq i ,\]
where the $\Theta_{ji}[t]$ are IID uniform on $[0,2\pi)$, and
$\SNR_i = |H_{ii}|^2$ and $\INR_{ji} = |H_{ji}|^2$ are the
squared moduli (which are constant over time).

Our results are in the context of so-called `line of sight'
communication models, without multipath interference. That is,
we consider a model where signal strengths attenuate
deterministically with distance according to
some function $a$. 


\addcontentsline{lof}{figure}{\numberline{\textbf{Definition 4.4}} $\SNR$, $\INR$}
\begin{definition}
\label{def:transprot}
Fix transmitter node positions $\{ \vec T_1, \ldots, \vec T_n \} 
\in \RR^d$ and receiver node positions
$\{\vec R_1, \ldots, \vec R_n \} \in \RR^d$, and consider Euclidean distance
$\| \quad \|$ and an \defn{attenuation function} $a\colon \bR_+ \to \bR_+$.

We define $\SNR_i = a(\|\vec R_i-\vec T_i\|)$, and for all pairs with $i \neq j$, define 
$\INR_{ji} = a(\|\vec R_j-\vec T_i\|)$.
\end{definition}

We consider the $n$-user Gaussian interference network. So
transmitter $i$ sends a message encoded as a codeword $\vec{x}_i = (x_i[1], \ldots, x_i[T])$ to receiver $i$, under a power constraint
$\frac{1}{T} \sum_{t=1}^T |x_i[t]|^2
 \leq 1$ for each $i$. The $t$th symbol received at receiver $j$ is given as
\[ \label{eq:transmod} Y_j[t] =  \exp(\ii \Theta_{jj}[t]) 
\sqrt{\SNR_{j}} x_j[t] 
+ \sum_{i\neq j} \exp(\ii \Theta_{ji}[t]) 
\sqrt{\INR_{ji}} x_i[t]  + Z_j[t],
\]
where the noise terms $Z_j[t]$ are independent standard complex Gaussian random variables,
and the phases $\Theta_{ji}[t]$ are
independent $U[0,2\pi)$ random variables independent of all other terms.
The $\INR_{ji}$ and $\SNR_i$ remain fixed over time, since the node positions themselves are fixed,
but the phases are fast-fading, in that they are renewed for each $t$.

\addcontentsline{lof}{figure}{\numberline{\textbf{Definition 4.5}} Power-law attenuation}
\begin{definition} \label{def:decay}
We say an attenuation function $a$ has \defn{power-law attenuation} if there exist constants
$\alpha$ and $\cdec$ such that for all $\rho$, we have
  $ a(\rho)  \leq \cdec \rho^{-\alpha} $.
\end{definition}

In particular, standard power-law decay of the form $a(\rho) = h\rho^{-\alpha/2}$
clearly satisfies this with $\cdec = h$ and $\alpha$ set to $\alpha/2$.  Other models we discussed in Section 1.4 such
as $a(\rho) = h\max \{ 1, \rho^{-\alpha/2} \}$ and $a(\rho) = h(\rho+\rho_0)^{-\alpha/2}$
also satisfy this.

For brevity, we write $S_{ji}$ for the random variables $\frac12 \log(1 + 2 \INR_{ji})$ (when $i \neq j$),
and $S_{ii}$ for $\frac12 \log(1 + 2 \SNR_{i})$
which  are functions of the distance
between the transmitters and receivers.
 In particular, since the nodes are positioned
independently, under this model
the random variables $S_{ji}$ are identically distributed, and $S_{ji}$ and
$S_{lk}$ are IID when $\{i,j\}$ and $\{k,l\}$ are disjoint. 

We will also write
$E = \EE S_{ii} = \frac12 \EE \log(1 + 2 \SNR)$,
noting that this is independent of $i$.  (It is also true that
$E = \EE S_{ji}$ for all $i$ and $j$.) Lemma 4.12 later ensures
that $E$ is indeed finite.


\subsection{Node position model}

We believe that our techniques should work in a variety of models for the node positions. We outline
one very natural
scenario here.

\addcontentsline{lof}{figure}{\numberline{\textbf{Definition 4.6}} IID network}
\begin{definition} \label{def:nodeplace}
Consider two probability distributions $\PP_T$ 
and $\PP_R$ defined on $d$-dimensional space $\RR^d$. Given an integer $n$, we sample  
$n$ transmitter node positions
$\vec T_1, \ldots, \vec T_n$ independently from the distribution $\PP_T$. Similarly, we sample
$n$ receiver node positions $\vec R_1, \ldots, \vec R_n$ independently from distribution $\PP_R$. We
refer to such a model of node placement as an \defn{IID network}.
\end{definition}

(Equivalently, we could state that transmitter and receiver positions are distributed like two independent
non-homogeneous Poisson processes, conditioned such that there are $n$  points 
of each type.)

 We pair the transmitter and receiver nodes up so that transmitter $i$ at $\vec T_i$ wishes to communicate 
with receiver $i$ at $\vec R_i$ for each $i$.

\begin{center}
  \includegraphics[width=0.79\textwidth]{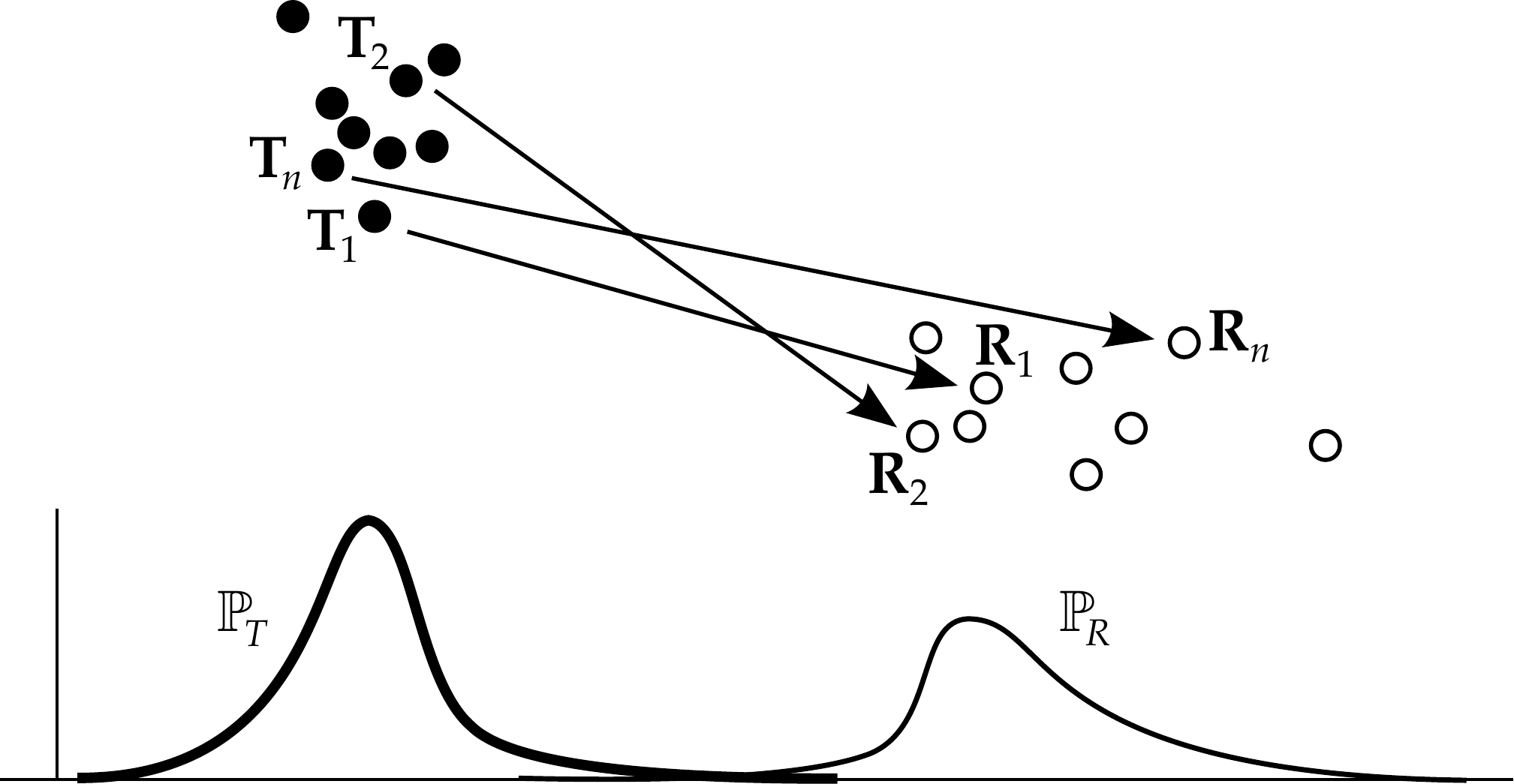}
\end{center}

Transmitters and receivers that are very close
to each other will lead to very strong interference or signals.  These extreme occurences could
prevent the network from operating as we would like.  Also, as we remarked in Section 1.4,
our attenuation models lose physical accuracy at very small distances.  For this
reason, we will demand that our node positioning model ensures that this is
rare, a property we call \defn{spatial separation}.

\addcontentsline{lof}{figure}{\numberline{\textbf{Definition 4.7}} Spatial separation}
\begin{definition} \label{def:spatsep}
Let $\vec T \sim \PP_T$ and $\vec R \sim \PP_R$ be placed independently in $\RR^d$.
We say the IID network is \defn{spatially separated} if there exists constants $\beta>0$
and $\csep$ such that for all $\rho$
  \[ \PP( \| \vec R - \vec T \| \leq \rho) \leq \csep \rho^{\beta} . \]
\end{definition}

In particular, we can show that the standard dense network
is spatially separated.

\addcontentsline{lof}{figure}{\numberline{\textbf{Definition 4.8}} Standard dense network}
\begin{definition}
  The \defn{$d$-dimensional standard dense network} is an IID
  network defined by $\PP_T$ and $\PP_R$ being
  independent uniform measures on $[0,1]^d$.
\end{definition}

  \begin{center}
%
  \includegraphics[scale=0.45]{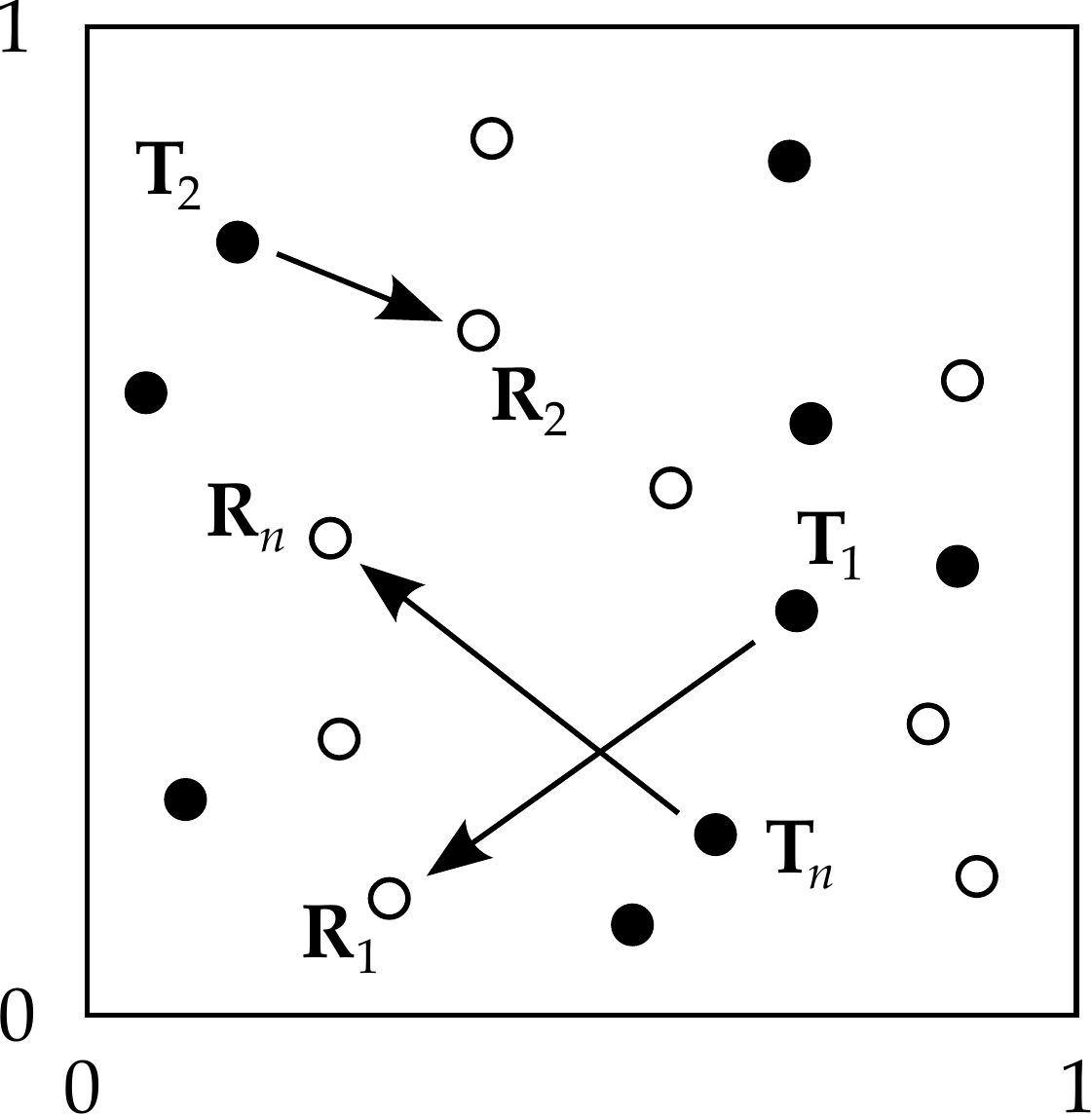}
  \end{center}

\addcontentsline{lof}{figure}{\numberline{\textbf{Lemma 4.9}} The standard dense network is spatially separated}
\begin{lemma}
  The standard dense network is spatially
  separated.
\end{lemma}

\begin{proof}
  We need to show that Definition 4.7 is fulfilled.  By conditioning
  on $\vec T = \vec t$, we get  
  
\begin{minipage}{0.35\linewidth}
  \begin{align*}
      &\Prob (\| \vec R - \vec T \| \leq \rho) \\
        &\qquad{}= \int_{[0,1]^d} \Prob_R \big(  B(\vec t, \rho) \cap [0,1]^d \big) \, \ud \vec t \\
        &\qquad{}\leq \int_{[0,1]^d} \Prob_R \big(  B(\vec t, \rho)  \big) \, \ud \vec t \\
        &\qquad{}\leq \int_{[0,1]^d} v(d) \rho^d \, \ud \vec t \\
        &\qquad{}= v(d) \rho^d,
   \end{align*}
\end{minipage}
\hspace{0.02\linewidth}
\begin{minipage}{0.55\linewidth}
    \includegraphics[scale=0.78]{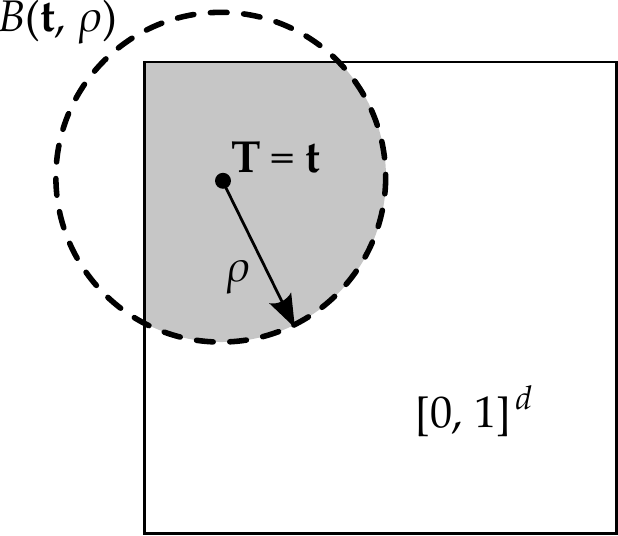}
\end{minipage}   
%
\mbox{}
\medskip

\noindent where $v(d)$ is the volume of the $d$-dimensional unit ball $B(\vec 0,1)$.  Taking $\csep = v(d)$, $\beta =d$ gives the result.
\end{proof}


The standard dense network has been the subject of much research (see the review paper
of Xue and Kumar \cite{XueKumar} and references therein). However, we emphasise that our result holds for a wider range
of models.

\section{Jafar network}

We now review in more detail Jafar's important result.  We call a network with
fixed equal $\snr$s and IID $\INR$s a \defn{Jafar network}. (Note that the Jafar
network cannot be written as an IID network with power-law attenuation.)

\medskip

\noindent \textbf{Theorem 4.1 restated.}
  \emph{Suppose direct $\SNR$s are fixed and identical, so $\SNR_i = \snr$ for all $i$,
  and suppose that all $\INR$s are IID random and supported on some neighbourhood
  of $\snr$.  Then the average per-user capacity $\Csum/n$ tends in probability
  to $\frac12 \log(1+2\snr)$ as $n\to\infty$.}
\medskip

Proving the direct part of this result is simple: the central result of ergodic interference
alignment (Theorem 2.15) tells us that the rates $\frac12 \log(1+2\snr)$ are simultaneously achievable by all users.

For the converse part, Jafar defines \cite[proof of Theorem 5]{Jafar} the crosslink $i \link j$ as being
a \defn{$\epsilon$-bottleneck link} if the sum-capacity of the two-user network with links $i\link i$ and $j \link j$
with crosslink $i \link j$ has sum-rate bounded by 
  \begin{equation} r_i + r_j \leq \log(1+2\snr) + \epsilon . \label{bneck} \end{equation}
for some fixed $\epsilon > 0$.
Analysis of these bottleneck links then
gives the converse result.

We will use a similar method to this to prove our main result (Theorem 4.3). We will
need to alter the definition of bottleneck links slightly to fit our needs. Also, the entire
converse part is made more complicated: while the Jafar network has lots of convenient independences
between the $\snr$ and $\INR$s, we are not so lucky.  Therefore extra care must be taken.

We also give here an alternative proof of Theorem 4.1. This proof uses techniques from
graph theory, and gives a faster rate of convergence than Jafar's own proof
-- exponential, rather than $O(n^{-2})$.

\begin{proof}[Alternative proof of Theorem 4.1]
We need to show that every user is involved in a bottleneck
link.  We will first review some facts from random graph theory.

Form a random bipartite graph by taking 
two sets $\mathcal V$, $\mathcal W$ of vertices each of size $K$, and
making each edge from $\mathcal V$ to $\mathcal W$ present independently with probability
$\delta$ (there are no edges within either $\mathcal V$ or $\mathcal W$).
A \defn{matching} is a set of $k$ of the edges such that every vertex in
the graph is adjacent to one edge -- so each vertex $v \in \mathcal V$
is matched to a vertex $w \in \mathcal W$ by an edge $vw$.

  \begin{center}
    \includegraphics[scale=0.89]{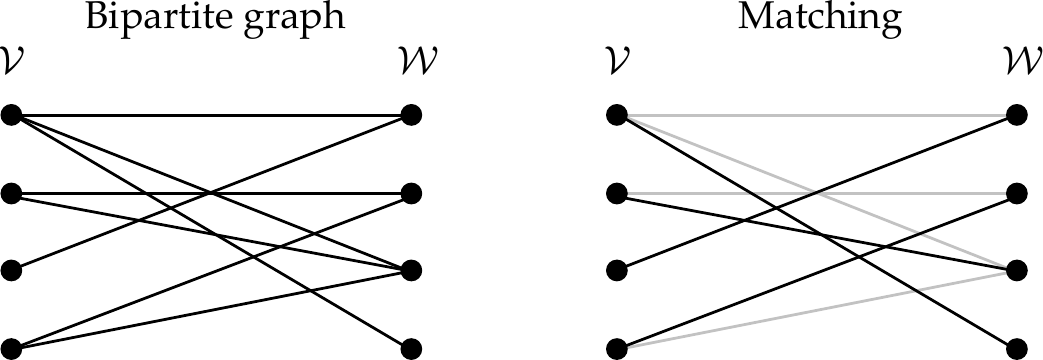}
  \end{center}

A classical result due to Erd\"{o}s and R\'{e}nyi \cite[Theorem 2]{ErdosRenyi}
(originally stated in terms of random matrices) states that
the probability that a matching 
fails to exist tends to $0$ for any $\delta = \delta(K) = (\log K + \omega(1))/K$,
where $\omega(1)$ is, using Bachmann--Landau asymptotic notation, a term that tends to $\infty$. 

We recall the argument where $\delta$ is a fixed constant, so that we
can be precise about the bounds, rather than just working asymptotically.
Following Walkup \cite[Definition 1]{Walkup}, we say that a subset $\bar{\mathcal V} \subseteq \mathcal V$ of
size $k$ and a subset $\bar{\mathcal W} \subseteq \mathcal W$ of size $K-k+1$ form a \defn{blocking pair} of 
size $k$ if no
edge of the graph connects $\bar{\mathcal V}$ to $\bar{\mathcal W}$. Walkup \cite[Section 3]{Walkup} uses
K\"{o}nig's theorem (equivalently one can use Hall's marriage theorem) to deduce that 
  \begin{align*} 
    \Prob( \text{no matching from $\mathcal V$ to $\mathcal W$} )  
      &\leq \sum_{k=1}^K \mathop{\sum_{|\bar{\mathcal V}| = k}}_{|\bar{\mathcal W}| = K-k+1}
                            \Prob \left( (\bar{\mathcal V}, \bar{\mathcal W}) \text{ a blocking pair} \right) \\
      &= \sum_{k=1}^{K} \binom{K}{k} \binom{K}{K-k+1} (1-\delta)^{k(K-k+1)} \\                  
      &= 2 \sum_{k=1}^{(K+1)/2} \binom{K}{k} \binom{K}{k-1} (1-\delta)^{k(K-k+1)} .
  \end{align*}
We split this sum into the terms where $k \leq \sqrt{K}$ and those where $k > \sqrt{K}$.
Using the bounds
  \[ (1-\delta)^{k(K-k+1)} \leq \begin{cases}
                     \displaystyle \exp \left( - \delta\, \frac{K+1}{2} \right) & \text{for $k \leq \sqrt{K}$,} \\
                     \displaystyle \exp \left( - \delta\, \frac{K^{3/2}}{2} \right) & \text{for $k > \sqrt{K}$,}
                                \end{cases} \]
we get the bound
  \[ \Prob( \text{no matching}) \leq
        2 \sqrt{K} K^{2 \sqrt{K}}  \exp \left( - \delta\, \frac{K+1}{2} \right)
            + 2^{2K} \exp \left( - \delta\, \frac{K^{3/2}}{2} \right) .\]
We deduce that the probability
of a complete matching failing to exist decays at an exponential rate in $K$.

We can now prove Theorem 4.1.

Construct a random bipartite graph by dividing the receiver-transmitter links 
into two groups $\mathcal V$ and $\mathcal W$ of 
size $K = n/2$.
Choose $\epsilon > 0$ and, for $i \in \mathcal V$ and $j \in\mathcal W$ include
the edge $ij$ if either of the crosslinks $i\link j$ or $j \link i$ is an
$\epsilon$-bottleneck link, which in the Jafar network occurs independently with some probability
$\delta$.

We seek a matching on this graph.
For each pair $(i,j)$ that is successfully matched up, the corresponding two-user channel is an 
$\epsilon$-bottleneck, and
hence has the bound $r_i + r_j \leq \log (1+ 2 \snr) + \epsilon$ by \eqref{bneck}.
for any achievable rates.
If every edge is matched this way, we have
  \[ r_\Sigma = \sum_{i=1}^n r_i = \sum_{\text{matches $(i,j)$}} r_i + r_j
       \leq \frac{n}{2} \left(\log (1+ 2 \snr) + \epsilon\right) \]
The high probability of a matching existing implies exponential decay of
  \[ \Prob \left(\left| \frac{\csum}{n} - \frac12 \log(1+2 \snr)\right| > \epsilon \right) \to 0\]
proving the theorem with exponential convergence.
\end{proof}

For our IID networks, extra dependencies between links make the picture
much more difficult, so we have been unable to find a proof that extends
this random bipartite graph method.  Whether or not exponential convergence holds
for IID networks with power-law attenuation is an open problem.

\section{Proof: achievability}

We can now prove our main theorem, Theorem 4.3, by breaking the probability into two terms, which we deal
with separately. So
\begin{equation} \label{eq:unionbd}
 \PP \left( \left| \frac{\Csum}{n} - E \right| > \epsilon \right) 
 = \PP \left( \frac{\Csum}{n} - E  < -\epsilon \right) +
\PP \left( \frac{\Csum}{n} - E  > \epsilon \right).\end{equation}

Bounding the first term of (\ref{eq:unionbd}) corresponds to the achievability part of the proof.
Bounding the second term of (\ref{eq:unionbd}) corresponds to the converse part, and
represents our major contribution.

We prove the direct part using ergodic interference alignment.

\begin{proof}
The first term of (\ref{eq:unionbd}) 
can be bounded relatively simply, using ergodic interference alignment. 
A theorem of Nazer, Gastpar, Jafar, and Vishwanath (Theorem 2.15 of this thesis) 
implies that the rates
  \[ R_i = \frac12 \log(1 + 2 \SNR_i) = S_{ii} \]
are simultaneously achievable.

This implies that
$\Csum \geq \sum_{i=1}^n R_i = \sum_{i=1}^n S_{ii}$. This allows us to bound the first term in
(\ref{eq:unionbd}) as 
  \[ \PP \left( \frac{\Csum}{n} - E  < -\epsilon \right)
       \leq \PP \left( \frac{\sum_{i=1}^n S_{ii}}{n}   < E-\epsilon \right).
\label{eq:lln}
\]
But $E = \EE S_{ii}$, so this probability tends to $0$ by the weak law of large numbers.
\end{proof}

\section{Proof: converse}

We now need to show that the second term of (\ref{eq:unionbd}) tends to $0$ too.
Specifically, we must prove the following: for all $\epsilon > 0$
  \begin{equation} \label{converse}
    \Prob \left( \frac{C_\Sigma}{n} \geq E + \epsilon \right)
                                                              \to 0  \qquad \text{as $n \to \infty$.}
  \end{equation}

The proof of the converse part is the major new part of this chapter.
First, bottleneck links are introduced, and we prove a tight information-theoretic
bound on the capacity of such links.  Second, a probabilistic counting argument
ensures there are (with high probability) sufficiently many bottleneck links to bound
the sum-capacity of the entire network.

\subsection{Bottleneck links}

The important concept is that of the bottleneck link, an idea first used
by Jafar \cite{Jafar} and later adapted by Johnson, Aldridge, and Piechocki \cite{JAP} in the following form:

\addcontentsline{lof}{figure}{\numberline{\textbf{Definition 4.10}} Bottleneck link}
\begin{definition}
  We say the link $i\link j$, $i\neq j$, is an \emph{$\epsilon$-bottleneck
  link}, if the following three conditions hold:
    \begin{description}
      \item[B1] $S_{ii} \leq E + \epsilon/2$,
      \item[B2] $S_{ji} \leq E + \epsilon/2$,
      \item[B3] $S_{jj} \leq S_{uj}$.
    \end{description}
  
  We let $B_{ij}$ be the indicator function that the crosslink $i\link j$ is a $\epsilon$-bottleneck
  link.  We also define the \emph{bottleneck probability} $\beta :=
  \Ex B_{ij}$ to be the probability that a given link is an $\epsilon$-bottleneck which is
  independent of $i$ and $j$ for an IID network.  (We suppress the $\epsilon$ dependence for simplicity.)
\end{definition}

\begin{wrapfigure}{r}{0.41\textwidth}
  \vspace{-0.2cm}
  \begin{center}
    \includegraphics[scale=0.79]{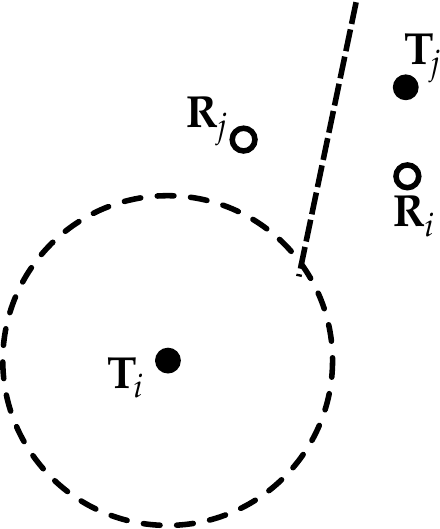}
  \end{center}
  \vspace{-0.2cm}
\end{wrapfigure}
This definition does have a physical interpretation (although the only
motivation for it is that it allows us to prove the bound in Lemma
\ref{links}).

The physical interpretation is this: fix the position
$\vec T_i$ of transmitter $i$.  Condition B1 requires that receiver
$i$ is sufficiently far away from $\vec T_i$, and condition B2 requires that receiver
$j$ is sufficiently far away too.  Condition B3 requires that transmitter
$j$ is closer to receiver $j$ than to receiver $i$.


The crucial point about bottleneck links is the constraints they place
on achievable rates in a network.

\addcontentsline{lof}{figure}{\numberline{\textbf{Lemma 4.11}} Sum-capacity of bottleneck link}
\begin{lemma} \label{links}
  Consider a crosslink $i\link j$ in a $n$-user Gaussian interference
  network.  If $i\link j$ is a $\epsilon$-bottleneck link, then the
  sum of their achievable rates is bounded by
  $r_i + r_j \leq 2E + \epsilon$.
\end{lemma}

\begin{proof}
  First, note that we make things no worse by considering the two-user
  subnetwork:
    \begin{align*}
      Y_i &= \exp (\ii\Theta_{ii}) \sqrt{\SNR_i}X_i
               + \exp (\ii\Theta_{ij}) \sqrt{\INR_{ij}}X_j + Z_i \\
      Y_j &= \exp (\ii\Theta_{ij}) \sqrt{\INR_{ji}}X_i
               + \exp (\ii\Theta_{ji}) \sqrt{\SNR_j}X_j + Z_j
    \end{align*}
  where receiver $i$ needs to determine message $m_i$, and receiver $j$
  message $m_j$.  (The time index is omitted for clarity.)
  
  From bottleneck conditions B1 and B2 we have
    \[ 1 + 2\SNR_i \leq \exp (2E + \epsilon) , \quad
       1 + 2\INR_{ij} \leq \exp (2E + \epsilon) .     \]
  Summing and taking logs gives
    \begin{equation} \label{logs}
      \log (1+\SNR_i+\INR_{ij}) \leq 2E + \epsilon .
    \end{equation}
    
  We combine this with the argument given by Jafar \cite{Jafar}, which we discussed earlier (Section 2.3).
  Let $r_i$ and $r_j$ be jointly achievable rates for the subnetwork.
  In particular, receiver $i$ can determine message $m_i$ with an arbitrarily
  low probability of error.  

  We certainly do no worse if a genie presents message $m_i$ to receiver
  $j$, so we assume receiver $j$ can indeed recover $m_i$.  But condition B3
  ensures that it is easier for receiver $i$ to determine $m_j$ than
  it is for receiver $j$ (since the weighting is larger in the first
  case).  So since receiver $j$ can recover $m_j$ (as $r_j$ is achievable),
  receiver $i$ can recover $m_j$ also.
  
  Because receiver $i$ can determine both $m_i$ and $m_j$, these
  two signals must have been transmitted at a sum-rate no higher
  than the sum-capacity of the Gaussian multiple-access channel
  (Theorem 2.8).
  Hence,
    \[ r_i + r_j \leq \log(1+\SNR_i+\INR_{ij}) \leq 2E + \epsilon, \]
  where the second inequality comes from \eqref{logs}.
\end{proof}

\subsection{Three technical lemmas}

A few technical lemmas are required in order to prove 
\eqref{converse}.

First, we need to ensure that very high $\SNR$s are very rare (Lemma \ref{SNRs}).
Second, we need to show that bottleneck links will actually occur (Lemma \ref{beta}).
Last, we must show that the number of bottleneck links cannot vary
too much (Lemma \ref{Var}).

Under any network model where these three lemmas are true, our theorem will hold.
We emphasise that our model of IID networks with power-law attenuation is one such
model; we believe the result holds more widely.

\addcontentsline{lof}{figure}{\numberline{\textbf{Lemma 4.12}} Very high $\SNR$ unlikely}
\begin{lemma} \label{SNRs}
  Consider a spatially separated IID network, with power-law attenuation.  Then for any $\eta > 0$,
    \[ \Prob \bigg(\max_{1\leq i\leq n} S_{ii} > n^{\eta/2}\bigg)  = O(n^{-1})
                                     \qquad \text{as $n \to \infty$.} \]
\end{lemma}

In fact, in our case convergence to $0$ is considerably quicker than
$O(n^{-1})$, but this is sufficient.

\begin{proof}
First, we have by the union bound
  \begin{equation} \Prob (\max S_{ii} > n^{\eta/2}) \leq n \Prob ( S_{11} > n^{\eta/2}) .  \label{ubd} \end{equation}

Now we apply the definition of $S_{11} := \frac12 \log(1 + 2\SNR)$ and
recall that $\SNR_1 = a(\|\vec R_1 - \vec T_1\|)$ (Definition 4.4) to get
  \begin{align*}
    \Prob ( S_{11} > n^{\eta/2})
      &= \Prob \left( \SNR_{1} > \frac12 (2^{2n^{\eta/2}}-1) \right) \\ 
      &= \Prob \left( a(\| \vec R_1 - \vec T_1 \|) > \frac12 (2^{2n^{\eta/2}}-1) \right) .
  \end{align*}
Since $a$ is a power-law attenuation function, we have
  \begin{align*}
    \Prob ( S_{11} > n^{\eta/2})
      &\leq \Prob \left( \cdec \| \vec R_1 - \vec T_1 \|^{-\alpha} > \frac12 (2^{2n^{\eta/2}}-1) \right) \\ 
      &= \Prob \left( \| \vec R_1 - \vec T_1 \| < \left(\frac{1}{2\cdec} (2^{2n^{\eta/2}}-1) \right)^{-1/\alpha} \right);
   \end{align*}
and since the network is spatially separated, we have
  \[
    \Prob ( S_{11} > n^{\eta/2})
      \leq \csep \left(\frac{1}{2\cdec} (2^{2n^{\eta/2}}-1) \right)^{-\beta/\alpha} = O(n^{-2})
   \]   
(and obviously much tighter than $O(n^{-2})$).
Together with \eqref{ubd}, this gives the result.
\end{proof}

(It is worth noting that this fast decay in the tails of $S_{ii}$ ensures that the
expectation $E = \EE S_{ii}$ does indeed exist and is finite.)

We will often condition off this event; that is, condition on the complementary event
$\{ \max S_{ii} \leq n^{\eta/2} \}$.  We use $\PP_n$, $\EE_n$ and $\Var_n$ to
denote such conditionality, and write $\beta_n=\EE_nB_{ij}$ for the
conditional bottleneck probability.

The next two lemmas concern showing that conditional probabilities are
nonzero.  However,  we have for any event $A$,
	\begin{align*}
	  \Prob(A) &= \Prob(A \mid \max S_{ii} \leq n^{\eta/2}) \Prob(\max S_{ii} \leq n^{\eta/2}) \\
                   &\qquad\qquad\qquad\qquad\qquad {}+ \Prob(A \mid \max S_{ii} > n^{\eta/2}) \Prob(\max S_{ii} > n^{\eta/2}) .
  \end{align*}
and hence by Lemma \ref{SNRs} we have the two bounds
  \begin{align*}
    \Prob(A) &\leq \Prob(A \mid \max S_{ii} \leq n^{\eta/2}) + \Prob(\max S_{ii} > n^{\eta/2}) 
             = \Prob_n(A) + O(n^{-1})  \\
	  \Prob(A) &\geq\Prob(A \mid \max S_{ii} \leq n^{\eta/2}) \Prob(\max S_{ii} \leq n^{\eta/2}) 
             = \Prob_n(A) \big(1-O(n^{-1})\big),
  \end{align*}
and so
  \begin{equation} \Prob(A) = \Prob_n(A) + O(n^{-1}) . \label{bigo} \end{equation}
This will be useful in the next two proofs.

\addcontentsline{lof}{figure}{\numberline{\textbf{Lemma 4.13}} The bottleneck probability is bounded away from zero}
\begin{lemma} \label{beta}
  Consider a spatially separated IID network, with power-law attenuation.  Then the conditional
  bottleneck probability $\beta_n$ is bounded away from $0$ for
  all $n$ sufficiently large.
\end{lemma}

\begin{proof}
  First note that by \eqref{bigo}, we need only show that the
  unconditional bottleneck probability $\beta$ is nonzero.

  Second, note that by the exchangeability of $\vec R_i$ and $\vec R_j$,
  we have
    \[ \Prob(\text{B1 and B2 and B3}) \geq \frac12 \Prob(\text{B1 and B2}) .\]
  It is left to show that $\Prob(\text{B1 and B2})$ is non-zero.
  
  Note that B1 requires $S_{ii}$ to be less than its expectation plus $\epsilon$.
  So $\vec R_i$ must be situated such that this has nonzero probability.  So $\vec T_i$
  has a nonzero probability of being positioned such that B1 occurs.  But
  $\vec T_i$ and $\vec T_j$ are also exchangeable, so we are done.
\end{proof}

\addcontentsline{lof}{figure}{\numberline{\textbf{Lemma 4.14}} A bound on the variance of the\\ number of bottleneck links}
\begin{lemma} \label{Var}
  Consider a spatially separated IID network with power-law attenuation.
  Then, conditional on $\{\max_i S_{ii} < n^{\eta/2} \}$,
  \[ \Var_n (\# \text{\emph{ bottleneck links}})
       = \Var_n \left( \sum_{i\neq j} B_{ij} \right)
       = O(n^3) ,                                  \]
  where the sum is over all crosslink pairs $(i,j)$, $i\neq j$.
\end{lemma}

In general, one might assume that $\Var_n (\# \text{ bottleneck links})$
would be proportional to the square of the total number of links, and thus be $O(n^4)$.
However, because of the independences in the IID network, the variance is in fact
much lower.

\begin{proof}
  First consider the unconditional version.  We have
    \[ \Var \left( \sum_{i\neq j} B_{ij} \right)
         = \sum_{i\neq j} \sum_{k\neq l} \Cov(B_{ij}, B_{kl}) . \]
  The important observation is that for $i,j,k,l$ all distinct,
  $B_{ij}$ and $B_{kl}$ are independent giving $\Cov(B_{ij}, B_{kl}) = 0$.
  (This is because they depend only on the position of distinct and
  independently-positioned nodes.) Hence there are only $O(n^3)$ non-zero
  terms in the sum. Each non-zero covariance term in the sum is bounded by
  variance of the indicator function, so
   \[ \Cov(B_{ij}, B_{kl}) \leq \Var B_{ij} = \beta(1-\beta) \leq \frac14 . \]
  
  But by \eqref{bigo}, if $\Cov(B_{ij}, B_{kl}) = 0$, then the
  conditional covariance is very small $\Cov_n(B_{ij}, B_{kl}) = O(n^{-1})$.
  Hence, 
    \[ \Var_n \left( \sum_{i\neq j} B_{ij} \right) \leq O(n^3)\frac14 + O(n^4)O(n^{-1}) = O(n^3), \]
  as desired.
\end{proof}

\subsection{Completing the proof}

We are now in a position to prove \eqref{converse}, and hence prove
Theorem 4.3.

\begin{proof}
We need to show
  \[
    \forall\, \epsilon > 0 \quad  \forall\, \delta>0 \quad \exists\, N
       \quad  \forall\, n\geq N \quad
         \Prob\left( \frac{C_\Sigma}{n} \geq E + \epsilon \right)
           \leq \delta.      
  \]
So choose $\epsilon>0$, $\delta>0$, fix $n \geq N$ (where $N$ will
be determined later), and pick a fixed rate vector $\mathbf{r}
\in \bR_+^n$ with sum-rate
  \begin{equation} \label{sumrate}
    \frac{r_\Sigma}{n} > E + \epsilon ;
  \end{equation}
we need to show that $\Prob(\mathbf r\text{ is achievable}) \leq \delta$.
(Here, we are writing $r_\Sigma := \sum_{i=1}^n r_i$ for the sum-rate.)

We divide into two cases: when there is a very high $\SNR$, which is unlikely
to happen; and when there is not, in which case $\mathbf r$ is unlikely to be
achievable.  Formally,
  \begin{align}
    &\Prob(\mathbf r \text{ achievable}) \notag \\ 
       &\qquad\qquad = \Prob(\mathbf r\text{ achievable} \given \max S_{ii} \leq n^{\eta/2})
           \Prob(\max S_{ii} \leq n^{\eta/2}) \notag \\
         &\qquad\qquad \qquad\qquad\qquad {}+ \Prob(\mathbf r\text{ achievable} \given \max S_{ii} > n^{\eta/2})
             \Prob(\max S_{ii} > n^{\eta/2}) \notag \\
       &\qquad\qquad \leq \Prob(\mathbf r\text{ achievable} \given \max S_{ii} \leq n^{\eta/2}) 
       				+ \Prob(\max S_{ii} > n^{\eta/2}) \notag \\
       &\qquad\qquad  \leq \Prob_n (\mathbf r\text{ achievable} )
           + \frac{\delta}{2} , \label{bound}
  \end{align}
for $n$ sufficiently large, by Lemma \ref{SNRs}.  We need to bound the first term in 
\eqref{bound}.

First, note that our assumption on $\max_i S_{ii}$ means that if $r_i > 2n^{\eta/2}$,
then we break the single-user capacity bound, since we would have
  \[
     r_i > 2n^{\eta/2} \geq 2 \max_j S_{jj} 
       \geq 2 S_{ii} = \log (1+2\SNR_i) > \log (1+\SNR_i)
  \]
meaning $\mathbf r$ is not achievable, and we are done.  Thus we assume this does
not hold; that
  \begin{equation} \label{boundR}
    r_i \leq 2n^{\eta/2} \quad \text{for all $i$}.
  \end{equation}

(The rest of our argument closely follows Jafar \cite[proof of Theorem 5]{Jafar}.)

Now, if $\mathbf r$ is achievable, it must at least satisfy the constraints on the
$\epsilon$-bottleneck links $i\link j$ from Lemma \ref{links}, and hence also the sum of those constraints.  So
  \begin{align}
    \Prob_n(\mathbf r \text{ achievable})  
      &\leq \Prob_n ( r_i + r_j \leq 2E + \epsilon
          \text{ on bottleneck links $i\link j$} ) \notag \\
      &\leq \Prob_n \left( \sum_{i\neq j} B_{ij} (r_i + r_j)
          \leq \bigg( \sum_{i\neq j} B_{ij} \bigg) (2E + \epsilon) \right) \notag \\
      &= \Prob_n (U \leq V) , \label{UV}
  \end{align}
where we have defined
  \begin{align*}
    U &:= \frac{1}{n(n-1)} \sum_{i\neq j} B_{ij} (r_i + r_j) , \\
    V &:= \frac{1}{n(n-1)} \bigg( \sum_{i\neq j} B_{ij} \bigg) (2E + \epsilon) .
  \end{align*}

The conditional expectations of $U$ and $V$ are
  \[ \Ex_n U = 2\beta_n \frac{r_\Sigma}{n}, \quad
     \Ex_n V = \beta_n (2E + \epsilon)
                 = 2\beta_n \left(E + \frac{\epsilon}{2} \right). \]
Note that since $\beta_n > 0$ by Lemma 6, we can rewrite
\eqref{sumrate} as
  \[ \Ex_n U > \Ex_n V + \beta_n\epsilon , \]
or equivalently,
  \[ \Ex_n U - \frac{\beta_n\epsilon}{2} > \Ex_n V + \frac{\beta_n\epsilon}{2} . \]
The proof is completed by formalising the following idea: since 
the expectations are ordered $\Ex_n U > \Ex_n V$, we can only rarely
have the opposite ordering $U < V$.  Hence the expression in \eqref{UV} is small.

Formally, by (the conditional version of) Chebyshev's inequality and the union bound, we have
  \begin{align}
    \Prob_n ( U \leq V) 
      &\leq \Prob_n \bigg( U \leq \Ex_n U - \frac{\beta_n\epsilon}{2}
                    \text{ or } V \geq \Ex_n V + \frac{\beta_n\epsilon}{2} \bigg) \notag \\
      &\leq \Prob_n \bigg( |U - \Ex_n U| \geq \frac{\beta_n\epsilon}{2} \bigg)
              + \Prob_n \bigg( |V - \Ex_n V| \geq \frac{\beta_n\epsilon}{2} \bigg) \notag \\
      &\leq \bigg(\frac{2}{\beta_n\epsilon}\bigg)^2 \Var_n U + \bigg(\frac{2}{\beta_n\epsilon}\bigg)^2 \Var_n V 
                                                                                                     \notag \\
      &= \frac{4}{\beta_n^2\epsilon^2} ( \Var_n U + \Var_n V ) . \label{UV2}
  \end{align}

Using Lemma \ref{Var} we can bound these variances as
  \begin{align*}
    \Var_n U &= \frac{1}{n^2(n-1)^2} \Var_n \left( \sum_{i\neq j} B_{ij} (r_i + r_j) \right) \\
                 &\leq \frac{1}{n^2(n-1)^2} O(n^3) 16 n^\eta \\ &= O(n^{-(1-\eta)}). \\
    \Var_n V &= \frac{1}{n^2(n-1)^2} O(n^3) (2E+\epsilon)^2 = O(n^{-1}) ,
  \end{align*}
where we used \eqref{boundR} to bound $\Var_n U$.  Choosing $\eta$ to be
less than $1$, we can ensure $N$ is sufficiently large that for all
$n \geq N$,
  \[ \Var_n U + \Var_n V \leq \frac{\beta_n^2 \delta \epsilon^2}{8} . \]


\noindent This makes \eqref{UV2} into $\Prob_n (U \leq V) < \delta /2$.
Together with \eqref{UV} and \eqref{bound}, this yields the result.
\end{proof}

\section{Conclusion}

In this chapter we have defined IID interference networks with power-law attenuation.  We have shown
that this setup fulfills necessary properties for the average per-user capacity
$\Csum/n$ to tend in probability to $\frac12 \EE \log(1 + 2 \SNR)$.  We have
also noted that this result is not unique to our setup.  We briefly mention one more example.

Suppose, for example, that Rayleigh fading is added to our model.  That is, now let
$\SNR_i := |H_{ii}|^2 a(\|\vec T_i-\vec R_i\|)$ and
$\INR_{ji} := |H_{ji}|^2 a(\|\vec T_i-\vec R_j\|)$, where the $H_{ji}$
are IID standard complex Gaussian random variables.

Because ergodic interference alignment still works with Rayleigh fading \cite[Section IV]{Nazer},
the direct part of the theorem still holds.  But also, because the fading coefficients
are IID, the independence structure from the non-fading case remains, ensuring Lemmas
4.12--4.14 hold.  Hence, the theorem is still true.

Characterising \emph{all} networks for which such a limit for average per-user capacity
exists is an open problem.

At the moment, Theorem 4.3 should perhaps be regarded as being of
theoretical interest. That is, our major contribution is to provide a sharp
upper bound on the performance of interference networks. However, the
lower bound relies on an ergodic interference alignment which, while
rigorously proved, may not be feasible to implement in practice for large
number of users. Examination of the proof of the effectiveness of ergodic
interference alignment \cite[Theorem 1]{Nazer} shows that,
even for a model with alphabet size $q$, the channel needs to be used
$O( (q-1)^{n^2})$ times. Even for $n \sim 10$, this is a prohibitive
requirement.  We approach this problem in the next chapter.

\section*{Notes}
\addcontentsline{toc}{section}{Notes}

The new work in this chapter -- specifically Theorem 4.3 and its proof, and the definition of the IID network --
is joint work with Oliver Johnson and Robert Piechocki.  This chapter is
based on two papers we wrote \cite{JAP,AJP}.  This work benefited from
the advice of anonymous reviewers for from \emph{IEEE Transactions on Information Theory} and
the 2010 IEEE International Symposium on Information Theory.

The ideas in this chapter were first studied by Jafar \cite{Jafar} --
in particular, the concept (although not the exact definition) of bottleneck
links comes from that paper.

The earlier of our two papers \cite{JAP} (which actually appeared in publication later)
considered only the standard dense network, but defined many of the important
concepts in this chapter.  The full general proof was first published in the later
of our two papers \cite{AJP} (which appeared earlier).

With the exception of our new proof of Theorem 4.1 from our earlier paper \cite{JAP}, the material on Jafar networks
(Section 4.3) is due to Jafar.

\addtocontents{lof}{\protect\addvspace{20 pt}}

%% file: chapters/delay.tex
\chapter{Delay--rate tradeoff for ergodic interference alignment}

\section{Introduction}

Earlier, in Section 2.8, we discussed the ergodic interference alignment scheme
of Nazer, Gastpar, Jafar, and Viswanath (which we hereafter to refer
to as the NGJV scheme).

Recall that
we considered a model of communication over a finite field $\F_q$ of size 
$q$. Since the NGJV scheme (see Section 2.8) requires a particular $n \times n$ channel matrix
with entries in $\F_q \setminus \{ 0 \}$ to occur, the expected delay for a particular 
message is $(q-1)^{n^2}$ (which is roughly $q^{n^2}$ for large $q$).
It is clear that even for $n$ and $q$ relatively small,
this is not a practical delay.  (For $n=6$ and $q=3$, for example,
the delay is $2^{36} \approx 7 \times 10^{10}$.)

There are five questions we would like to try to answer:
  \begin{enumerate}
    \item Can we find a scheme that, like NGJV, achieves half the single-user
      rate, but at a lower time delay?
    \item Can we find schemes that have lower time delays than NGJV, even
      at some cost to the rate achieved?
    \item Specifically, which schemes from Question 2 perform well for
      situations where we have few users ($n$ small)?
    \item Specifically, which schemes from Question 2 perform well for
      situations where we have many users ($n \to \infty$)?
    \item What is a lower bound on the best time delay possible for any
      scheme achieving a given rate for a given number of users?
  \end{enumerate}
  
In Sections \ref{sec:schemes} and 5.4, we define a new set of schemes, called
JAP (Subsection 5.3.1), a beamforming extension JAP-B (Subsection 5.3.4),
and child schemes derived from them (Section 5.4) 
that have lower time delays than the NGJV scheme, for
a variety of different rates, answering Question 2.  As a special case,
examined in Subsection 5.3.5, the $\JAPB([n])$ schemes
achieve half the single-user rate, like NGJV, while reducing the time delay from
$q^{n^2}$ to $q^{(n-1)(n-2)}$, answering Question 1.  In Section
5.5, we answer Questions 3 and 4, by finding and analysing the JAP schemes
that perform the best for small and large $n$; the table on page \pageref{table} 
and the graphs on page \pageref{graphs} illustrate the best schemes for small $n$, and Theorems
5.6 and 5.7 give the asymptotic behaviour of the schemes.  Question 5 remains an
open problem (although we do give a lower bound on the delay achievable for the
schemes listed above).

Koo, Wu, and Gill \cite{Koo} have previously attempted to answer Questions 2 and 3.
We briefly outline their work at the end of Section 5.2.

\section{Model} \label{sec:model}

We give our results in the context of the finite-field interference network
with fast-fading.

Recall that the single-user capacity of this channel is $\log q - \HH(Z) =: \D(Z)$.
Extending our previous definition of degrees of freedom (Definition 1.13) to multiple users, we
have the following:

\addcontentsline{lof}{figure}{\numberline{\textbf{Definition 5.1}} Degrees of freedom (finite field network)}
\begin{definition}
  Given an achievable symmetric rate point $(r,r,\dots,r)$, we
  define the \emph{symmetric per-user degrees of freedom} to be
  $\DOF = r/\D(Z)$.
\end{definition}

In particular, it's clear that a single user can achieve $1$ degree of freedom.

We define the expected time delay for the NGJV scheme to be the average number of time slots we must wait
after seeing a channel matrix $\mt H$ until we see the corresponding matrix $\mt I-\mt H$.  The
time delay is geometrically distributed with parameter $p$, where $p$ is
the probability that the random channel matrix takes the value $\mt I - \mt H$.
The mean of this random variable is $1/p$; hence the problem of finding the average
time delay is reduced to a problem of finding the probability that a desired
matrix appears in the next time slot.  Since a channel matrix has $n^2$ entries,
each of which needs to take the correct one value of $q-1$ possible values, the
average time delay is
  \begin{equation} D = \frac{1}{\big(\frac{1}{q-1}\big)^{n^2}} = (q-1)^{n^2} \sim q^{n^2}. \label{eqdelay} \end{equation}
(Here and elsewhere, we write $f(q) \sim g(q)$ if $f(q)/g(q) \to 1$
as $q \to \infty$.)

As we mentioned before, this expected delay will be quite large even for modest values
of $q$ and $n$.
For this reason, we will concentrate on the delay exponent. 

\addcontentsline{lof}{figure}{\numberline{\textbf{Definition 5.2}} Delay exponent, delay coefficient}
\begin{definition} An interference alignment scheme with expected
delay $D \sim k q^T$ for some $k$ and $T$ has \emph{delay exponent}
$T$ and \emph{delay coefficient} $k$.  More specifically, we have
$T := \lim_{q\to\infty} \log D/\log q$.
\end{definition}

We regard reduction of the delay exponent as the key aim,
with the delay coefficient playing a secondary role.
In particular, the finite field model is in some sense an abstraction
of the model where channel coefficients are Gaussians quantized into a set of
size $q$, where $q$ is chosen large enough to reduce quantization
error.  When $q$ is large, the delay exponent $T$ dominates the
delay coefficient $k$ in determining size of the expected delay $D$.

From Theorem 2.14, we know that the NGJV schemes achieve $\DOF=1/2$,
and we have just shown in \eqref{eqdelay} it requires a delay
exponent of $n^2$.

%

We also mention some new schemes outlined in a recent paper by Koo, Wu, and
Gill \cite{Koo}.  They attempted to answer our Questions 2
and 3, by finding schemes -- we call them KWG schemes -- with lower delay than the NGJV scheme.
The KWG schemes  suggest matching
a larger class of matrices than simply $\mt{H}$ and $\mt{I} - \mt{H}$.
By analysing the hitting probability of an associated Markov chain,
they were able to reduce the expected delay, at the cost of a reduction
in rate (and hence degrees of freedom). However, their
schemes only affect the delay by a constant multiple, with
the shortest-delay scheme only reducing the delay to $0.64 (q-1)^{n^2} \sim 0.64 q^{n^2}$
with a sum-rate of $0.79 \D(Z)$  \cite[page 5]{Koo}.
That is, the KWG schemes only reduce the delay coefficient $k$, leaving
the delay exponent as $T = n^2$.  For modest $q$ and $n$ (say $q=3$, $n=6$, again),
we regard this delay as still impractical.
Since the KWG schemes achieve a lower rate than the NGJV scheme for the
same delay exponent, we shall only compare our results with the NGJV scheme.

\section{New alignment schemes: JAP and JAP-B} \label{sec:schemes}

\subsection{Three important observations}

In the NGJV scheme, all receivers were able to decode their message
by summing their two pseudomessages
  \[ \sum_{i=1}^n h_{ji}[t_0] \vc m_i + \sum_{i=1}^n h_{ji}[t_1] \vc m_i = \vc m_j \qquad \text{for $j=1,\dots,n$.} \]
In other words, the NGJV scheme relies on the linear dependence
  \[ \mt H[t_0] + \mt H[t_1] = \mt I . \]

This scheme has a large delay, because, given $\mt H[t_0]$, there is only
one matrix, $\mt H[t_1] = \mt I - \mt H[t_0]$, that can complete the linear dependence.
If there were a large collection of matrices that could complete the
dependence, then the delay would be lower.
  
We make three observations to this end.

First, while NGJV matches two channel states $\mt H[t_0]$ and
$\mt H[t_1]$ to form this linear dependence, we could use more than
two.  That is, if we have $K+1$ channel matrices
$\mt H[t_0], \mt H[t_1], \dots, \mt H[t_K]$ such that
  \[ \mt H[t_0] + \mt H[t_1] + \cdots + \mt H[t_K] = \mt I , \]
than receivers $j$ can sum the $K+1$ pseudomessages to
recover their message,
  \[ \sum_{i=1}^n h_{ji}[t_0] \vc m_i + \sum_{i=1}^n h_{ji}[t_1] \vc m_i +\cdots + \sum_{i=1}^n h_{ji}[t_K] \vc m_i = \vc m_j . \]
Note that the transmission of a single message is now split among
$K+1$ channel states, rather than $2$ as in NGJV.  This means
that the degrees of freedom of this scheme is reduced to $1/(K+1)$ from NGJV's
$1/2$.

Second, any linear combination of channel state matrices that
sums to $\mt I$ is sufficient.  That is, if there exist
scalars $\lambda_0, \lambda_1 \in \F_q$ such that
  \[ \lambda_0 \mt H[t_0] + \lambda_1 \mt H[t_1] = \mt I , \]
then all receivers can recover their message by forming the linear
combination of pseudocodewords
  \[ \lambda_0 \sum_{i=1}^n h_{ji}[t_0] \vc m_i + \lambda_1 \sum_{i=1}^n h_{ji}[t_1] \vc m_i = \vc m_j \qquad \text{for $j=1,\dots,n$.} \]
  
Third, NGJV requires all receivers to be able to decode their messages
at the same time.  However, receiver $j$ can decode its message if
  \[ \sum_{i=1}^n h_{ji}[t_0] \vc m_i + \sum_{i=1}^n h_{ji}[t_1] \vc m_i = \vc m_j \]
regardless of whether this equality holds for other receivers as well.  In other words,
receiver $j$ can decode its message if
  \begin{align*}
    h_{jj}[t_0] + h_{jj}[t_1] &= \one  \\
    h_{ji}[t_0] + h_{ji}[t_1] &= \zero \qquad \text{for $i\neq j$}.
  \end{align*}

Putting these three observations together, we make the following conclusion:
Let $\mt H[t_0], \mt H[t_1], \dots, \mt H[t_K]$ be a sequence of
$K+1$ channel state matrices.  If there exist scalars
$\lambda_0, \lambda_1, \dots, \lambda_K$ such that for some $j$
  \begin{align}
    \lambda_0 h_{jj}[t_0] + \lambda_1 h_{jj}[t_1] + \cdots + \lambda_K h_{jj}[t_K] &= \one  \\
    \lambda_0 h_{ji}[t_0] + \lambda_1 h_{ji}[t_1] + \cdots + \lambda_K h_{ji}[t_K] &= \zero \qquad \text{$i\neq j$},
  \end{align}
then receiver $j$ can recover its message by forming the linear combination
of pseudocodewords
  \[ \lambda_0 \sum_{i=1}^n h_{ji}[t_0] \vc m_i + \cdots + \lambda_K \sum_{i=1}^n h_{ji}[t_K] \vc m_i = \vc m_j . \]

In fact, we only require 
  \begin{align*}
    \lambda_0 h_{jj}[t_0] + \lambda_1 h_{jj}[t_1] + \cdots + \lambda_K h_{jj}[t_K] &\neq \zero  \\
    \lambda_0 h_{ji}[t_0] + \lambda_1 h_{ji}[t_1] + \cdots + \lambda_K h_{ji}[t_K] &= \zero \qquad \text{for $i\neq j$}.
  \end{align*}
since the coefficients $\lambda_k$ can be rescaled to make the top equation (5.2)
equal to $\one$ without breaking the bottom equation (5.3). 

Or, writing $\vc h_j^{\text{int}}$ for the interference vector
  \[ \vc h_j^{\text{int}} = (h_{j1}, \dots, h_{j\,j-1}, h_{j\,j+1}, \dots, h_{jn}) , \]
we can again rewrite the requirement as
  \begin{align}
    \lambda_0 h_{jj}[t_0] + \lambda_1 h_{jj}[t_1] + \cdots + \lambda_K h_{jj}[t_K] &\neq \zero  \label{one} \\
    \lambda_0 \vc h_j^{\text{int}}[t_0] + \lambda_1 \vc h_j^{\text{int}}[t_1] + \cdots + \lambda_K \vc h_j^{\text{int}}[t_K] &= \vc 0 . \label{two}
  \end{align}

If $n$ equalities like the \eqref{one} and \eqref{two} above hold, we say that ``receiver $j$
can recover its message from $\mt H[t_0$], $\mt H[t_1], \dots,$ $\mt H[t_K]$.''

The time delay of this scheme is $t_K - t_0$.

Recall that the average delay is the reciprocal of the probability that a
random matrix allows a receiver to recover its message.  Thus it will be useful
to note the following lemma.

\addcontentsline{lof}{figure}{\numberline{\textbf{Lemma 5.3}} Probability of message recovery}
\begin{lemma}
  Conditional on the interference vectors
  $\vc H_j^{\text{int}}[t_0],\dots, \vc H_j^{\text{int}}[t_K]$ being linearly
  dependent, the probability that receiver $j$ can recover its message
  is $1-O(q^{-1})$.
\end{lemma}

Note that we only use a matrix when we know for certain that it fulfills our desired criteria
-- we merely need to know what the probability is that the next matrix will do, in order
to calculate the expected delay.

\begin{proof}
  Since the interference vectors are linearly dependent, there exists
  a linear combination
    \[ \lambda_0 \vc H_j^{\text{int}}[t_0] + \lambda_1 \vc H_j^{\text{int}}[t_1] + \cdots + \lambda_K \vc H_j^{\text{int}}[t_K] = \vc 0 \]
  where $L > 0$ of the $\lambda_k$ are nonzero.  Thus, receiver $j$ can
  recover its message provided that the corresponding linear combination
    \begin{equation}
      \lambda_0 H_{jj}[t_0] + \lambda_1 H_{jj}[t_1] + \cdots + \lambda_K H_{jj}[t_K] \label{comb}
    \end{equation}
  is nonzero; call the probability that this happens $p$.
  
  When $\lambda_k \neq \zero$, then $\lambda_k H_{jj}[t_k] =: V_k$ is uniform on $\mathbb F_q \setminus \{\zero\}$,
  and when $\lambda_k = \zero$, then $\lambda_k H_{jj}[t_k] = \zero$ too.
  So \eqref{comb} is the sum of $L$ random variables $V_k$ IID uniform on $\mathbb F_q \setminus \{\zero\}$.
  We can write the mass function of each $V_k$ as
  $(1+\rho)U - \rho \delta_\zero$,  where $U$
  is uniform on $\mathbb F_q$, $\delta_\zero$ is a point mass on $\zero$, and
  $\rho = 1/(q-1)$.  Then the mass function of the $L$-fold convolution is
    \[ (1-(-\rho)^L)U + (-\rho)^L\delta_\zero . \]
    
  Hence, the probability that \eqref{comb} is zero
  is
    \[ 1 - p = \big( 1 - (-\rho)^L) \frac{1}{q} + (-\rho)^L = \frac{1}{q} + \frac{1}{q(q-1)^{L-1}} = O(q^{-1}). \]
  The result follows.
\end{proof}

\subsection{The scheme $\JAP(\vc a)$} \label{subsec:NGJVesque}

We now present our new scheme.

The idea behind the scheme is as follows: We start by seeing
some channel state $\mt H[t_0]$.  We then set $t_1$ to be the first
time slot that allows receivers $1$ to $a_1$ to recover their message
(where $a_1$ is decided on in advance).  Next, we set $t_2$ to be the first
time slot that allows receivers the next $a_2$ receivers to recover their message.
And so on, until all $n$ receivers have recovered their message.

Specifically, fix $K\leq n$ and a sequence $[a_1, a_2, \dots, a_K] =: \vc a$ of length
$K$ and weight $n$; that is, in the set
  \[ \mathcal A(n,K) := \left\{ \vc a \in \bZ_+^K : \sum_{k=1}^K a_k = n \right\}. \]
We write $A_k$ for the partial sums
$A_k := a_1 + a_2 + \cdots + a_k$ (so in particular $A_1 = a_1$
and $A_K = n$).

Then we define the scheme $\JAP(\vc a)$
as consisting of the following $K+1$ steps:
  \begin{description}
    \item[Step 0:] Start with a matrix $\mt H[t_0]$.
    \item[Step 1:] Set $t_1$ to be the first
      timeslot that allows the first $a_1$ receivers
      $1, 2, \dots, A_1$ to recover their message from
      $\mt H[t_0], \mt H[t_1]$.
    \item \mbox{} $\vdots$
    \item[Step \emph{k}:] Set $t_k$ to be the first
      timeslot that allows the next $a_k$ receivers
      $A_{k-1}+1, A_{k-1}+2, \dots, A_{k}$ to recover their message from
      $\mt H[t_0], \mt H[t_1], \dots, \mt H[t_k]$.
    \item \mbox{} $\vdots$
    \item[Step \emph{K}:] Set $t_K$ to be the first
      timeslot that allows the final $a_K$ receivers
      $A_{K-1}+1, A_{K-1}+2, \dots, A_{K}$ to recover their message from
      $\mt H[t_0], \mt H[t_1], \dots, \mt H[t_K]$.
  \end{description}
By the end of this process, all $n=A_K$ receivers have recovered their
message.

Since the message was split over $K+1$ time slots, the common rate of communication
is $\D(Z)/(K+1)$, which corresponds to $\DOF = 1/(K+1)$.

\subsection{Delay exponent of  JAP schemes}

We now examine the delay exponent for our new schemes.

\addcontentsline{lof}{figure}{\numberline{\textbf{Theorem 5.4}} Delay exponent of JAP scheme}
  \begin{theorem}
    Consider the $n$-user finite field interference network.
    Fix $K$ and $\vc a \in \mathcal A(n,K)$.  We use the scheme $\JAP(\vc a)$ as outlined
    above.  Then
      \begin{enumerate}
        \item the expected time for the $k$th round to take place
          is $D \sim q^{T_k(\vc a)}$, where 
          $T_k(\vc a) = a_k (n-k-1)$;
        \item the delay exponent for the whole scheme is
          \[ T(\vc a) := \max_{1\leq k\leq K} T_k(\vc a) = \max_{1\leq k\leq K} a_k (n-k-1) . \]
     \end{enumerate}
  \end{theorem}
  
\begin{proof}
  Recall that the expected delay is the reciprocal of the probability
  the desired match can be made.

  Suppose we are about to begin stage $k$ of a scheme $\JAP(\vc a)$.
  By Lemma 5.3, the probability we can complete the stage is $1-O(q^{-1})$
  multiplied by the probability that the interference vectors
  for the next $a_k$ receivers 
    \[ \vc H_j^{\text{int}}[t_0], \vc H_j^{\text{int}}[t_1],\dots, \vc H_j^{\text{int}}[t_k] \]
  are linearly dependent.
  
  If the first $k-1$ interference vectors are already linearly dependent,
  then we are done (with high probability, by Lemma 5.3).  Assume they are not.
  
  Write $\mathcal S$ for the span of the first $k-1$ interference
  vectors for one of the desired $a_k$ receivers $j$, so
    \[ \mathcal S := \text{span} \{ \vc H_j^{\text{int}}[t_0],\dots, \vc H_j^{\text{int}}[t_{k-1}] \} . \]
  Since all possible interference vectors in $(\F_q\setminus\{\zero\})^{n-1}$
  are equally likely, the probability that the next matrix completes
  a linear dependence is
    \[ \frac{| \mathcal S \cap (\F_q\setminus\{\zero\})^k|}{|(\F_q\setminus\{\zero\})^{n-1}|}
         = \frac{q^k s}{(q-1)^{n-1}} , \]
  where $s$ is the proportion of vectors in $\mathcal S$ with no
  zero entries. By counting the possible coefficients in $\mathbb F_q$ used
  in the span, the inclusion--exclusion formula gives us
    \[ s = 1 - (K-1)\frac{1}{q} + O\left(\frac{1}{q^2}\right) = 1 - O(q^{-1}) . \]
  Hence, the desired probability is
    \[
      \frac{q^k s}{(q-1)^{n-1}} \big(1-O(q^{-1}) \big) = \frac{q^k \big( 1 - O(q^{-1}) \big)}{(q-1)^{n-1}} \big(1-O(q^{-1})\big) \sim q^{-(n-k-1)}  \]
  (where the $1-O(q^{-1})$ term comes from Lemma 5.3).
    
  This property that a linear dependence is completed must hold for all $a_k$ receivers, which happens with probability
  $(q^{-(n-k-1)})^{a_k} = q^{-a_k(n-k-1)}$, hence the first result.
  
  For the second result, note that, as $q \to \infty$, the delay
  is dominated by the delay for the slowest round.
\end{proof}

\subsection{Improving delay with beamforming: $\JAPB$}

Beamforming slightly improves the performance of $\JAP(\vc a)$ schemes,
combining ideas from the original Cadambe--Jafar interference alignment
\cite{CadambeJafar} with the JAP scheme.

In round $k$ we can guarantee that the interference matches up for receiver
$l := A_{k-1}+1$.
Each transmitter $i$, instead of repeating their message $w_i$, rather encodes
$( h_{li}[t_k]] )^{-1} h_{li}[t_0] \vec m_i$.  (Since the coefficient $h_{li}$ cannot be $\zero$
and $q$ is prime,
the inverse term certainly exists.)  The total received interferences
at receiver $l$ at times $t_0$ and $t_k$ are
both equal to $\sum_{i \neq l} h_{li}[t_0] \vec m_i$, so can be estimated and
cancelled.

We refer to such schemes that take advantage of beamforming
as $\JAPB(\vc a)$ schemes.

\addcontentsline{lof}{figure}{\numberline{\textbf{Theorem 5.5}} Delay exponent of JAP-B scheme}
\begin{theorem}
The delay exponent of a $\JAPB(\vc a)$ scheme indexed
by sequence $\vc{a}$ is
  \[ T_B(\vc a) := \max_{1\leq k\leq K} (a_k-1)(n-k-1) . \]
\end{theorem}

\begin{proof}
At each round, receiver $l = A_{k-1}+1$ will automatically be able
to recover its message, leaving the JAP scheme to align interference for
the other $a_k -1$ users.  (Independence of the coefficients $h_{ji}$ ensures
that the scheme still has the same problem to solve.)
\end{proof}

In particular, the $\JAPB$ scheme will always outperform the $\JAP$ scheme
with the same sequence $\vc a$.

\subsection{An interesting special case: $\JAPB([n])$}

An interesting special case of the JAP-B schemes is the case when
$K=1$ and $a_1 = n$; we call this scheme $\JAPB([n])$.

In this case we have $1/(K+1) = 1/2$ degrees of freedom for a rate of $\D(Z)/2$.
From Theorem 5, we see that the delay exponent is
  \[ T_B([n]) = (a_1-1)(n-1-1) = (n-1)(n-2) . \]

Effectively, the $\JAPB([n])$ scheme works by using beamforming to
automatically cancel transmitter $1$'s interference, then for users
$2,3,\dots,n$ requiring the existence of diagonal matrices $\mt D_0,
\mt D_1$ such that $\mt D_0 \mt H[t_0] + \mt D_1 \mt H[t_1] = \mt I$.

Note that this is the same rate as is achieved by the original NGJV
scheme, but that the delay exponent has been reduced from NGJV's
$n^2$ to 
  \[ (n-1)(n-2) = n^2 - (3n-2) . \]
For small $n$ in particular,
this is a worthwhile improvement (see figure, p.\ \pageref{graphs}).

\section{Child schemes: using time-sharing}

Another way to generate new alignment schemes is by time-sharing schemes
designed for a smaller number of users.

Call the NGJV, KWG, JAP and JAP-B schemes `parent schemes'.
Given a parent scheme for the $m$-user network, we can modify
for the any $n$-user network with $n>m$, giving what we call a
`child scheme'.

Specifically, we use resource division by time (see Subsection 2.6.1) to split the network into $\binom nm$ subnetworks,
each of which contains a unique collection of just $m < n$ of the users.
Within each of these $m$-user subnetworks, a parent scheme is used, while
the other $n-m$ transmitters remain silent.

Resource division by time is often known as \defn{time-division multiple access} or TDMA -- we
use that abbreviation in the rest of this chapter.

Such a child scheme clearly has the same delay exponent as the parent scheme,
with the rate -- and thus the degrees of freedom -- reduced by a factor of
$m/n$.  So an $m$-user JAP-B scheme shared between
$n$ users gives $\DOF = m/n(K+1)$.

In particular, time-sharing the NGJV schemes for smaller networks gives a
collection of schemes with a lower delay exponent $m^2 < n^2$ than the
main NGJV scheme for a given number of users, reducing the degrees of freedom
from $1/2$ to $m/2n$.  

(We are not aware that the idea of time-sharing NGJV schemes has previously
appeared in the literature.  However, the idea seems simple enough that we regard
this as the `current benchmark' against which we should compare our new schemes.)

Interestingly, it seems that child schemes derived from time-sharing
an NGJV-like $\JAPB([n])$ parent scheme are particularly effective, and
very often performs better than other $\JAPB$ schemes.  We discuss this
point further in the next section.

\section{Best schemes}\label{sec:best}

\subsection{General case}

Given a number of users $n$ and a desired number of degrees of freedom
$\DOF = 1/(K+1)$, we wish to find a scheme with the lowest delay exponent.

For $K=n-1$ or $n$, when $\DOF = 1/n$ or $1/(n+1)$, the best JAP-B schemes have delay
exponent $T_B([1,\dots,1,2]) = T_B([1,\dots,1,1,1]) = 0$. This is the same
delay exponent as TDMA, which has $\DOF = 1/n$ also.  Thus we need not
consider schemes with $K=n-1$ or $n$.

For $K \leq n-2$ the best parent scheme will be a
JAP-B scheme with parameter vector $\vc a \in \mathcal A(n,K)$.  We write $T(n,K)$
for this best delay exponent, that is
  \[
    T(n,K) := \min_{\vc a \in \mathcal A(n,K)} T_B(\vc a) 
            = \min_{\vc a \in \mathcal A(n,K)} \ \max_{1\leq k\leq K} (a_k - 1)(n-k-1) .
  \]

We can bound $T(n,K)$ as follows.

\addcontentsline{lof}{figure}{\numberline{\textbf{Theorem 5.6}} Bounds on best delay exponents}
\begin{theorem} \label{thm:genbounds}
  Fix $n$ and $K \leq n-2$.
  For $T(n,K)$ as defined above, we have the following bounds:
    \[ \frac{n}{K} (n-2) - (2n-K-2) \leq T(n,K) \leq \frac{n}{K}  (n-2) \]
\end{theorem}

The gap between the upper and lower bounds grows linearly with $n$.

The following lemma on partial harmonic sums will be useful.

\addcontentsline{lof}{figure}{\numberline{\textbf{Lemma 5.7}} Bounds on partial harmonic sums}
\begin{lemma}\label{lem:harm}
  Let $S(n,K)$ be the partial harmonic sum
    \[ S(n,K) := \sum_{k=1}^K \frac{1}{n - k - 1} = \frac{1}{n-2} +  \cdots + \frac{1}{n-K-1} . \]
  Then we have the bounds
    \[ \frac{K}{n-2} \leq S(n,K) \leq \frac{K}{n-K-2} . \]
\end{lemma}

Of course, tighter bounds are available by comparing $\sum 1/k$ to $\int 1/x\, \ud x$,
but this suffices for our needs.

\begin{proof}
  There are $K$ terms in the sum, the largest of which is $1/(n-K-2)$ and the
  smallest of which is $1/(n-2)$.
\end{proof}

We can now prove Theorem \ref{thm:genbounds}.

\begin{proof}[Proof of Theorem \ref{thm:genbounds}] 
  The value of $T(n,K)$ is lower-bounded by the value of the same
  minimisation problem relaxed to allow the $a_k$ to be real.
  That is,
    \begin{align*}
      T(n,K) &= \min_{\vc a \in \bZ_+^K : \sum_k a_k = n} \ \max_{1\leq k\leq K} (a_k - 1)(n-k-1) \\
          &\geq \min_{\vc a \in \bR_+^K : \sum_k a_k = n} \ \max_{1\leq k\leq K} (a_k - 1)(n-k-1) .
    \end{align*}
  The relaxed problem is solved by waterfilling, setting $a_k - 1 = c/(n-k-1)$.
  Requiring the weight of $\vc a$ to be $n$ forces
    \[c = \frac{n-K}{S(n,K)} \geq \frac{(n-K)(n-K-2)}{K} , \]
  where we have used Lemma \ref{lem:harm}.  Rearrangement gives the lower bound.
  
  An upper bound is obtained by using the same $c$ and taking
   \[ a_k - 1 = \left\lceil \frac{c}{n-k-1} \right\rceil \leq \frac{c}{n-k-1} + 1 . \]
  This gives
    \begin{align*}
      T_B(\vc a) &\leq c + \max_k (n-k-1) \\
                  &= \frac{n-K}{S(n,K)} + (n-2) \\
               &\leq \frac{(n-K)(n-2)}{K} + (n-2) ,
    \end{align*}
  where we have used Lemma \ref{lem:harm}.
  Rearrangement gives the upper bound.
\end{proof}

\subsection{Few users: small $n$}

For small values of $n$, we can find the best parent JAP-B schemes
by hand.  (The task
is simplified by noting that the optimal $a_k$ will be nonzero
and increasing in $k$.)
The table below gives the delay exponents of the
best JAP-B schemes for $n = 3, \dots, 8$ and $K\leq n-2$. 

\setlength{\tabcolsep}{4pt}

\begin{table}[hb]\label{table}
\begin{center}

Best $\JAPB(\vc a)$ schemes for small values of\\
$n$ and $K$, and their delay exponents. 
\bigskip
\begin{small}
\begin{tabular}{c|cccccc}
\toprule
             & $n=3$ & $n = 4$ & $n= 5$ & $n = 6$ & $n=7$ & $n=8$ \\
\hline 
$K = 1$      & $2$   & $6$     & $12$  & $20$  & $30$  & $42$  \\[-0.15cm]
$\DOF = 1/2$ & $[3]$ & $[4]$   & $[5]$ & $[6]$ & $[7]$ & $[8]$  \\[0.1cm]
$K = 2$      & 0     & 2       & 4  & 8 & 12        & 18 \\[-0.15cm]
$\DOF = 1/3$ & TDMA  & $[1,3]$ & $[2,3]$     & $[3,3]$ & $[3,4]$ & $[4,4]$ \\[0.1cm]
$K = 3$      &       & $0$     & $2$ & $4$   & $6$ & $8$ \\[-0.15cm]
$\DOF = 1/4$ &       & TDMA    & $[1,1,3]^*$ & $[1,2,3]^*$ & $[2,2,3]$ & $[2,3,3]$ \\[0.1cm]
$K = 4$      & & & $0$ & $2$ & $4$ & $6$ \\[-0.15cm]
$\DOF = 1/6$ & & & TDMA & $[1,1,1,3]^*$ & $[1,1,2,3]^*$ & $[1,2,2,3]^*$ \\[0.1cm]
$K = 5$ & & & & $0$ & $2$ & $4$ \\[-0.15cm]
$\DOF = 1/6$ & & & & TDMA & $[1,1,1,1,3]^*$ & $[1,1,1,2,3]^*$ \\[0.1cm]
$K = 6$ & & & & & $0$ & $2$ \\[-0.15cm]
$\DOF = 1/7$ & & & & & TDMA & $[1,1,1,1,1,3]^*$ \\[0.1cm]
$K = 7$ & & & & & & $0$ \\[-0.15cm]
$\DOF = 1/8$ & & & & & & TDMA \\
\bottomrule
\end{tabular}
$^*$Asterisks mean that the choice of $\vc a$ achieving 
this delay exponent is non-unique.
\end{small}
\end{center}
\end{table}
\setlength{\tabcolsep}{6pt}
\begin{figure*}[pth] \label{graphs}
	\begin{center}
	  \noindent Delay-rate tradeoff for TDMA, NGJV and JAP-B \\
	  schemes for $n$-user interference networks
	  
	  \vspace{0.5cm}
		\includegraphics[width=\textwidth]{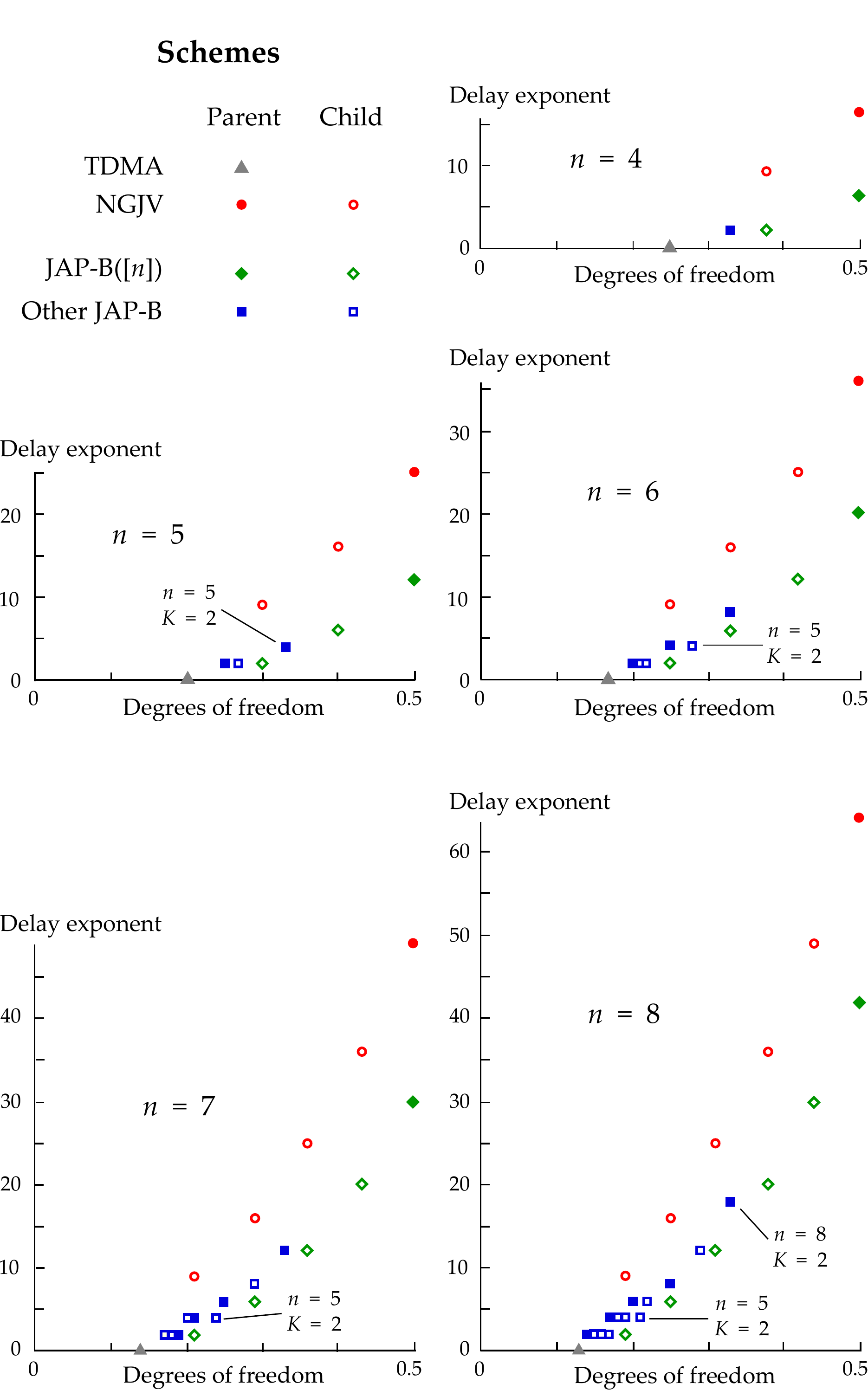}
  \end{center}
\end{figure*}

We can also consider child schemes based on parent JAP-B schemes.
The figure on page \pageref{graphs} plots the performance of NGJV and all JAP-B schemes, as well
as child schemes derived from them, for $n=3, \dots, 7$.  Note
that for many values of $n$ and $\DOF$, the scheme with the lowest
delay exponent is $\JAPB([n])$ or one of the child schemes derived
from it.  (Note however, that the the parent schemes with $n=5, K=2$ and
$n=8, K=2$, as well as child schemes derived from them, outperform 
$\JAPB([n])$ for some degrees of freedom.)

\subsection{Many users: $n \to \infty$}

We now consider the performance of schemes in the many-user limit
$n \to \infty$.

In particular, we are interested in two limiting regimes, specifying
how the degrees of freedom $\DOF(n)$ should scale with the number of users $n$.
In regime I the per-user rate is held constant; in regime II the sum-rate is kept
constant, so each user's individual rate falls like $1/n$.
  \begin{itemize}
    \item \textbf{Regime I}, where we hold the degrees of freedom
      constant as $n \to \infty$.  That is, we want to communicate at fixed
      fraction of the single-user rate, as in the NGJV scheme.
      In this regime I, we take $\DOF(n) = \alpha$ for some
      $\alpha \in (0,1/2]$.  (The NGJV scheme corresponds to
      $\alpha = 1/2$.)
    \item \textbf{Regime II}, where we allow the degrees of freedom to
      fall as the number of users increases, scaling like $1/n$.
      That is, we want to communicate at a fixed multiple of
      the rate allowed by resource division schemes like TDMA.
      In regime II, we take $\DOF(n) = \beta/n$ for some
      $\beta \geq 1$. (TDMA corresponds to $\beta =1$.)
  \end{itemize}

First, we consider how parent JAP-B schemes perform in the
many-user limit.

\addcontentsline{lof}{figure}{\numberline{\textbf{Theorem 5.8}} Scaling of $\JAPB$ parent scheme delay exponents}
\begin{theorem}
  For regimes I and II, as above, and as $n \to \infty$, we have the following
  results for the delay exponent $T(n)$ of parent JAP-B schemes:
  \begin{itemize}
    \item \textbf{Regime I:} Fix $\alpha\in (0,1/2]$.  Then the delay exponent
      for $\DOF(n) = \alpha$ scales quadratically like
        \[ T(n) \sim \frac{1}{\lfloor 1/\alpha \rfloor - 1} n^2 . \]
    \item \textbf{Regime II:} Fix $\beta > 1$.  Then the delay exponent
      for $\DOF(n) = \beta/n$ scales linearly, in that $T(n) = O(n)$, or more specifically,  
       \[ \left( \beta  + \frac{1}{\beta} - 2 \right)n - o(n) \leq T(n) \leq \beta n + o(n). \]
  \end{itemize}
\end{theorem}

\begin{proof}
  For regime I, we need $\DOF = 1/(K+1) \geq \alpha$, so we 
  take
    \[ K = \left\lfloor \frac{1}{\alpha} - 1 \right\rfloor = \left\lfloor \frac{1}{\alpha} \right\rfloor - 1 . \]
  But the
  general bounds on delay exponents from Theorem \ref{thm:genbounds} tell
  us that for fixed $K$ we have $T(n,K) \sim \frac{1}{K} n^2$.
  The result follows.
  
  For regime II, $1/(K+1) = \DOF = \beta/n$, so we need
    \[ K = \left\lfloor \frac{n}{\beta} - 1 \right\rfloor = \frac{n}{\beta} + O(1) . \]
  Hence
    \[ \frac nK = \frac{n}{n/\beta + O(1)} = \beta + o(1). \]
  Putting this into the bounds from Theorem 5.6 gives
    \[
      (\beta + o(1))(n-2) - \left(2n - \frac{n}{\beta} + O(1) \right) \leq T(n) 
           \leq (\beta + o(1))(n-2)   . \]
  Rearranging gives the result.
\end{proof}

Note that in regime I with $\alpha = 1/2$, we get $T(n) \sim n^2$, the
same as NGJV.
  
We noted previously that child schemes produced by sharing the parent
scheme $\JAPB([m])$ were particularly effective.  The following theorem
shows this.

\addcontentsline{lof}{figure}{\numberline{\textbf{Theorem 5.9}} Scaling of $\JAPB([m])$ child scheme delay exponents}
\begin{theorem}
  For regimes I and II, as above, and as $n \to \infty$, we have the following
  results for the delay exponent $T(n)$ of child schemes based on $\JAPB([m])$ parent schemes:
  \begin{itemize}
    \item \textbf{Regime I:} Fix $\alpha\in (0,1/2]$.  Then the delay exponent
      for $\DOF(n) = \alpha$ scales quadratically, in that
        \[ T(n) = 4\alpha^2n^2 - 6\alpha n + O(1) \sim 4 \alpha^2 n^2 . \]
    \item \textbf{Regime II:} Fix $\beta > 1$.  Then the delay exponent
      for $\DOF(n) = \beta/n$ is constant, in that
       \[ T(n) = (\lfloor 2 \beta \rfloor - 1)(\lfloor 2 \beta \rfloor - 2) . \]
  \end{itemize}
\end{theorem}

\begin{proof}
  Recall from Section 5.4 that sharing the scheme $\JAPB([m])$ amongst
  $n$ users gives $\DOF = m/2n$ for delay exponent $T = (m-1)(m-2)$.
  
  For regime I, note that $m/2n = \DOF(n) = \alpha$, so we need to
  take $m = \lfloor 2 \alpha n \rfloor$, giving
  $T(n) = (\lfloor 2 \alpha n \rfloor - 1)(\lfloor 2 \alpha n \rfloor - 2)$.
  The result follows.
  
  For regime II, note that $m/2n = \DOF(n) = \beta/n$, so we need to
  take $m = \lfloor 2 \beta \rfloor $, giving
  $T(n) = (\lfloor 2 \beta \rfloor - 1)(\lfloor 2 \beta \rfloor - 2)$.
\end{proof}

Note that asymptotically, this means that in both regimes child schemes
from $\JAPB([m])$ parent schemes are asymptotically more effective than any other
parent scheme.  This is because
  \[ 4\alpha^2n^2 \leq \frac{1}{\lfloor 1/\alpha \rfloor - 1} n^2 \]
(with inequality unless $\alpha = 1/2$, when no child scheme will achieve
the desired degrees of freedom) and any constant is less than $(\beta - 2)n$ 
for $n$ sufficiently large.

Note also that by the same argument as the above proof, sharing the
NGJV parent scheme gives $T(n) = 4\alpha^2 n^2$ in regime I, which is less
good than sharing $\JAPB([m])$, but the same to first-order terms.

\section{Conclusion}

In the Secion 5.1, the questions we attempted to answer were:
  \begin{enumerate}
    \item Can we find a scheme that, like NGJV, achieves half the single-user
      rate, but at a lower time delay?
    \item Can we find schemes that have lower time delays than NGJV, even
      at some cost to the rate achieved?
    \item Specifically, which schemes from Question 2 perform well for
      situations where we have few users ($n$ small)?
    \item Specifically, which schemes from Question 2 perform well for
      situations where we have many users ($n \to \infty$)?
    \item What is a lower bound on the best time delay possible for any
      scheme achieving a given rate for a given number of users?
  \end{enumerate}

In answer to question 2, we defined the new sets of parent schemes JAP and
the even more effective JAP-B, and also derived child schemes from them. We noted
that these had lower time delays -- and sometimes significantly lower -- at the costs
of some loss in rate (or equivalently degrees of freedom). We saw that the
child schemes from $\JAPB([n])$ schemes were often particularly effective.

In answer to question 1, we noted that the $\JAPB([n])$ schemes keep
the degrees of freedom to $1/2$ while reducing the delay exponent
from $n^2$ to $(n-1)(n-2) = n^2 - (3n - 2)$.
  
In answer to Questions 3 and 4, we explicitly found the best schemes
$\JAPB$ schemes for $n \leq 8$, and analysed the asymptotic behaviour
of our schemes as $n \to \infty$.

Question 5 remains an open problem.

\section*{Notes}
\addcontentsline{toc}{section}{Notes}

This chapter is joint work with Oliver Johnson and Robert Piechocki,
and is based on our paper \cite{JAPdelay}.

The NJGV scheme is due to Nazer, Gastpar, Jafar, and Vishwanath \cite{Nazer}.

The delay--rate tradeoff problem was first studied by Koo, Wu, and Gill \cite{Koo}.

\addtocontents{lof}{\protect\addvspace{20 pt}}

%% file: chapters/grouptesting.tex
\chapter[Interference, group testing, and channel coding][Interference, group testing, and channel coding]{Interference, group testing, and channel coding}

\section{Building the interference graph}

In Section 2.6, we looked at resource division schemes such as resource division by time (Subsection 2.6.1).  We noted that for a user to communicate without interference, it was necessary for all of the other users to stand idle.  However, if not every receiver gets interference for every transmitter, than it might be possible for more than one user to communicate through the network at once.

For concreteness, consider the $N$-user finite field interference network with fixed fading, so
  \[ Y_{jt} = \sum_{i=1}^N h_{ji} x_{it} + Z_t \pmod{q} . \]
For the moment, so that we can concentrate on the interference, we will assume that the noise $Z$ is $\zero$ with probability $1$, so
  \[ Y_{jt} = \sum_{i=1}^N h_{ji} x_{it} \pmod{q} . \]

If we have for some $i\neq j$ that $h_{ji} = h_{ij} = \zero$, then both users $i$ and $j$ can communicate simultaneously and interference-free.  (Recall that we use the word ``user'' to mean a transmitter--receiver pair.  So we mean that if $h_{ji} = h_{ij} = \zero$, then transmitter $i$ can communicate to receiver $i$ and simultaneously transmitter $j$ can  communicate to receiver $j$, both links without interference.)

More generally, we can build the \defn{interference graph} to show which users interfere with which.

\addcontentsline{lof}{figure}{\numberline{\textbf{Definition 6.1}} Interference graph}
\begin{definition}
  The \defn{interference graph} of an $N$-user interference network with
  fixed fading has vertex set
    \[ \mathcal V := \{ 1,2,\dots, N\} \]
  and edge set
    \[ \mathcal E := \{ ij : h_{ji} \neq \zero
                               \text{ or } h_{ij} \neq \zero \} . \]
\end{definition}

The figure below shows the nonzero links in a network (transmitters on the left; receivers on the right), and
the interference graph derived from it.

  \begin{center}
    \includegraphics[scale=0.89]{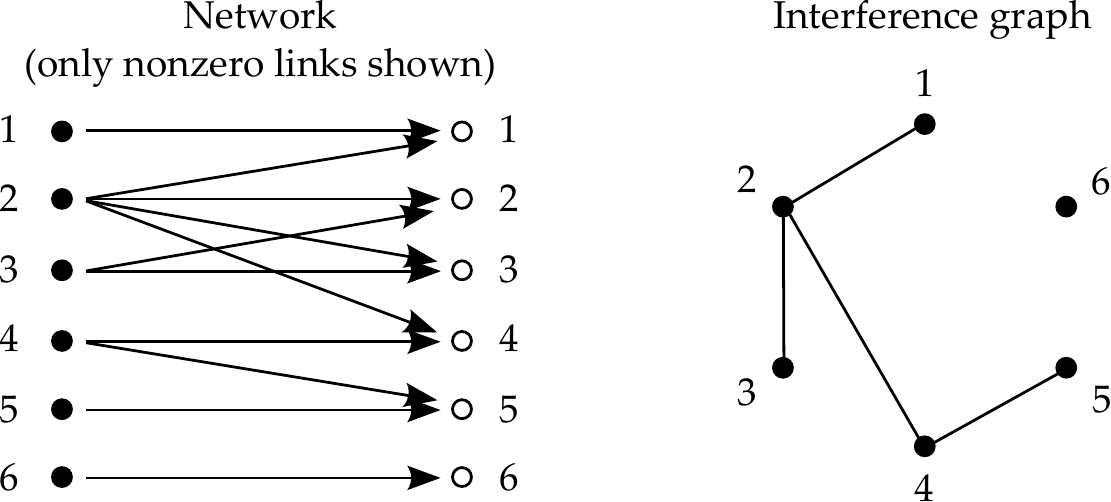}
  \end{center}

So two users $i$ and $j$ can communicate simultaneously and interference-free if they are not joined by an edge in the interference graph.  More generally, any independent set of the interference graph can communicate at the same time.  (Recall that a set of vertices $\mathcal U \subseteq \mathcal V$ is called \defn{independent} if no edge in the graph joins one vertex in $\mathcal U$ to another.)

Optimal use of such a resource division strategy requires operating using only maximal independent sets.  In particular, the maximum degrees of freedom achievable by a resource division scheme is $\dof = \alpha$, where the \defn{interference number} $\alpha$ of the graph is the size of the largest independent set.

  \begin{center}
    \includegraphics[scale=0.89]{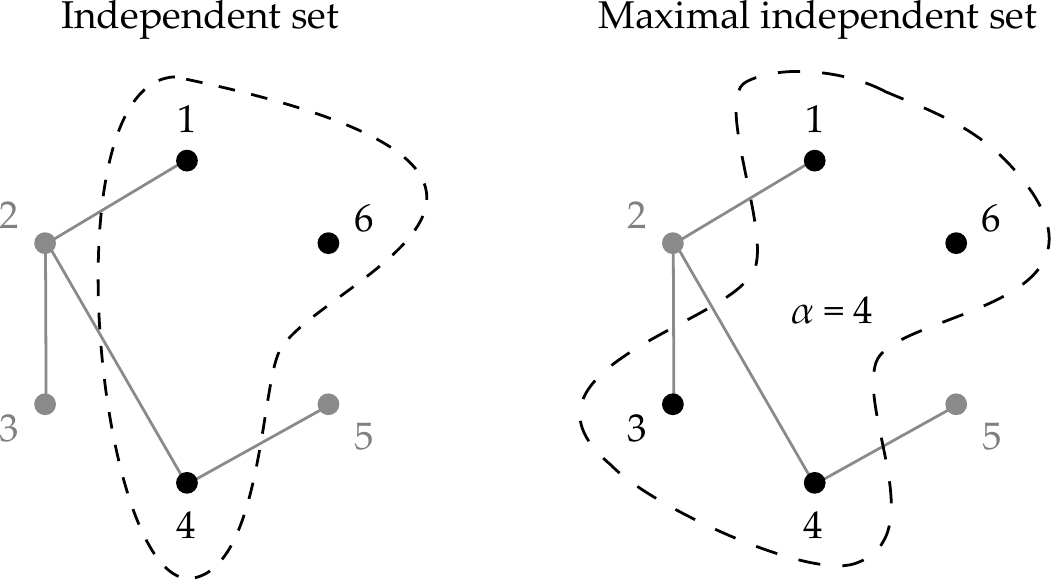}
  \end{center}
  
Given such a network, the users need to find out which other users they interfere with.  Here is one method they could use to do so: Let each transmitter $i$ choose a random nonzero message $m_i \in \mathbb{F}_q \setminus \{\zero\}$.  Then, for each of $T$ timeslots, each transmitter $i$ either sends $m_i$ or just sends the empty message $\zero$.  In other words, let $x_{it} = \one$ denote that transmitter $i$ communicates in timeslot $t$, and $x_{it} = \zero$ that she does not.  Then each receiver $j$ receives the signal
  \[ y_{jt} = \sum_{i=1}^N h_{ji} x_{it}m_i . \]

For large $q$, the probability that $y_{jt} = \zero$ when at least one transmitter has $x_{it} = \one$ is small -- for the moment we neglect this.  Then if $y_{jt} = \zero$, receiver $j$ knows that $h_{ji} = \zero$ for all $i$ with $x_{it} = \one$; conversely, if $y_{jt} \neq \zero$, receiver $j$ knows that $h_{ji} \neq \zero$ for at least one $i$ with $x_{it} = \one$.

Receiver $j$ wants to discover which $h_{ji}$ are nonzero (but doesn't need to know their actual values) in as few tests $T$ as possible -- this $T$ will depend on the number of users $N$, the number $K$ of transmitters $i$ for which $h_{ji} \neq \zero$, and the acceptable error probability $\epsilon$.  This is equivalent to the problem of \defn{group testing}, which we will outline more fully in the next section. 

Our model would be more accurate if we included a noise term $Z$ that wasn't always $\zero$ and if we did not neglect the possibility that signals cancel each other out.  Thus we would like our group testing protocols to be robust to this noise and to still have a small probability of error $\epsilon$.  In this chapter, we investigate how a channel coding approach to group testing can help with this.

Similar problems in multiuser networks have also been studied from a group testing perspective.  Berger and coauthors \cite{Berger} and Capetanakis \cite{Capetanakis} studied this problem using a model where more than one interfering message results in a collision where all messages are lost (rather than our model where signals are superposed at the receiver).  Zhang, Luo, and Guo \cite{ZhangLuoGuo} studied the Gaussian network, where low-interference links $h_{ji} \approx 0$ are assumed to be zero and those signals are treated as noise.

\bigskip

\noindent In the rest of this chapter, we outline the problem of group testing, and explore a new approach to it
using techniques from channel coding (as outlined in Chapter 1).

We define for the first time \defn{group testing channels}, which operate much like communications channels, and identify an important property where `only defects matter' that allows us to prove a theorem on the number of tests needed for accurate group testing to be possible.
We also give the first information theoretic bound on adaptive group testing, by drawing an analogy to channel coding
with feedback.

\section{Group testing: a very short introduction}

The problem of \defn{group testing} concerns detecting the defective members of a set of items through the means of pooled tests.  (In our previous example, for receiver $j$, think of the `defective items' as the interfering transmitters with $h_{ji} \neq \zero$.)  Group testing as a subject dates back to the work of Dorfman \cite{Dorfman} in 1940s studying practical ways of testing soldiers' blood for syphilis, and has received much attention from combinatorialists and probabilists since.

The setup is as follows: Suppose we have of a set $N$ items, of which a subset $\K$ of size $K$ is defective.
To identify $\K$, we could test each of the $N$ items individually for defectiveness.
However, when $K$ is small compared to $N$, most of the tests will give negative results.
A less wasteful method is to test \defn{pools} of numerous items together at the same time.
After a number $T$ of such pooled tests, it should be possible to deduce which items were defective.

Let $x_{it} = \one$ denote that item $i$ is included in test $t$.  In the so-called \defn{deterministic case}, a test of $n_t := |\{ i : x_{it}=\one\}|$ items of which $k_t := |\{ i \in \K : x_{it}=\one\}|$ are defective gives a negative result $y_t = \zero$ if no defects are tested ($k_t = 0$) and a positive result if at least one defect is pooled into the test ($k_t \geq 1$).

After $T$ tests, we make an estimate $\Khat$ of the defective set, with some average probability of error $\epsilon$.  We want to choose our tests in such a way that $\epsilon$ is small, while keeping $T$ as  low as possible.

Traditionally, this has been seen as a combinatorial problem: given $N$ and $K$, one aims to find an $N \times T$ \defn{testing matrix} $\mat X = (x_{it})$ such that all $\binom NK$ possible defective sets $\K$ give a different sequence of test results $y_1, \dots, y_T$.  This gives a zero error probability $\epsilon=0$, and one is interested in how small $T$ can be made. (See, for example, the textbook of Du and Hwang \cite{DuHwang} for more details on the combinatorial approach to group testing.)  

However, an alternative approach is to use random pools.  That is, we set $X_{it}$ to be random $\zero$s or $\one$s -- typically, $X_{it} = \one$ with some probability $p$ IID across $i$ and $t$, where $p$ may depend on $K$ and $N$. One then investigates how big $T$ must be compared to $N$ and $K$ in order to keep the average error probability $\epsilon$ low.

Recent progress has been made on this channel coding approach by Atia and Saligrama \cite{AtiaSaligrama}, by comparing the problem to the classical problem of channel coding, first studied by Shannon \cite{Shannon}. (See Chapter 1 for more details on channel coding.)

The two figures below show an interesting similarity between the two problems.  Note that our goals are slightly different, though -- in channel coding we wish to maximise the number of messages $M$ for large blocklengths $T$; whereas in group testing we wish to minimise the number of tests $T$ for a large number of items $N$.

\newpage

	\begin{center}
	
	Channel coding:

	\smallskip	
	
		\begin{picture}(335,40)(0,0)

		  \put(0,0){\framebox(55,30)}
		  \put(9,18){Message}
		  \put(23,3){$m$}
		  \put(24,33){}
		  
   	      \put(55,15){\vector(1,0){35}}
		  \put(71,18){}
		  
		  \put(90,0){\framebox(55,30)}
		  \put(94,18){Codeword}
		  \put(108,3){$\vec x(m)$}
		  
		  \put(145,15){\vector(1,0){45}}
		  \put(151,20){$p(y\given x)$}
		  \put(108,33){}
		  \put(151,4){$T$ times}
		  
		  \put(190,0){\framebox(55,30)}
		  \put(197,18){Received}
		  \put(200,3){signal $\vec y$}
		  \put(214,34){}
		  
		  \put(245,15){\vector(1,0){35}}
		  \put(259,18){}
		  
		  \put(280,0){\framebox(55,30)}
		  \put(289,18){Message}
		  \put(284,3){estimate $\hat m$}
		  \put(300,33){}
		  
    \end{picture}

  \bigskip

	Group testing:
	
	\smallskip
	
		\begin{picture}(335,40)(0,0)

		  \put(0,0){\framebox(55,30)}
		  \put(7,18){Defective}
		  \put(17,3){set $\K$}
		  \put(24,33){}
		  
     	  \put(55,15){\vector(1,0){35}}
		  \put(71,18){}
		  
		  \put(90,0){\framebox(55,30)}
		  \put(108,18){Test}
		  \put(98,3){design $\mat X$}
		  
		  \put(145,15){\vector(1,0){45}}
		  \put(146,20){$p(y\given n,k)$}
		  \put(108,33){}
		  \put(153,4){$T$ tests}
		  
		  \put(190,0){\framebox(55,30)}
		  \put(208,18){Test}
		  \put(192,3){outcomes $\vec y$}
		  \put(214,34){}
		  
		  \put(245,15){\vector(1,0){35}}
		  \put(259,18){}
		  
		  \put(280,0){\framebox(55,30)}
		  \put(287,18){Defective}
		  \put(284,3){set estm.\ $\Khat$}
		  \put(300,33){}
		  
    \end{picture}

\end{center}

In single-user point-to-point channel coding, $M$ messages are encoded into codewords $(x_{m1}, \dots, x_{mT})$ of length $T$.  After being sent through some channel, the codewords are received as $(y_1, \dots, y_T)$, and the original message is estimated as $\hat m$, with the hope that the average error probability $\epsilon$ will be small. 

Shannon's celebrated channel coding theorem \cite[Theorem 11]{Shannon} (see Theorem \ref{Shannons2} of this thesis) tells us how large we can make $M$ compared to $T$, while still being sure that the error probability stays small.  Shannon's breakthrough was to study a random coding scheme, were the $X_{mt}$ are all IID according to some distribution $X$.  One way to phrase the achievability part of Shannon's theorem -- Shannon offered a similar phrasing as an alternative in his original paper \cite[Theorem 12]{Shannon} -- is the following:

\addcontentsline{lof}{figure}{\numberline{\textbf{Theorem 6.2}} Shannon's channel coding theorem}
\begin{theorem}[Shannon's channel coding theorem]
  Consider a communications channel $(\mathcal X, \mathcal Y, p(y \given x))$.
  Let $M^* = M^*(T,\epsilon)$ be the maximum number of messages that can be sent
  through the channel with blocklength $T$ and error probability at most $\epsilon \in (0,1)$.
  Then
    \[ M^* \geq 2^{T \max_X \I(X:Y) + o(T)} \qquad \text{as $T\to\infty$.} \]
\end{theorem}

Atia and Saligrama \cite[Theorem III.1]{AtiaSaligrama} adapted Gallager's  proof \cite{Gallager} of the achievability part of Shannon's channel coding theorem  to give a similar result. This time, we're interested in how many tests $T$ are required to keep the error probability arbitrarily low.

\addcontentsline{lof}{figure}{\numberline{\textbf{Theorem 6.3}} Group testing theorem: deterministic case}
\begin{theorem}  
%
  Consider group testing in the deterministic case.
  
  Let $T^* = T^*(N,K,\epsilon)$ be the minimum number of tests necessary
  to identify $K$ defects among $N$ items with error probability at most $\epsilon \in (0,1)$.
  Then $T^* \leq \overline T + o(\log N)$ as $N\to\infty$, where
    \[ \overline T =  \min_p \max_{\L\subset\K}
                    \frac{\log_2 \binom{N-K}{|\L|} \binom{K}{|\L|}}
                         {\I (\vec X_{\K\setminus\L} : \vec X_{\L}, Y)} . \]
\end{theorem}


(It's worth noting that the term inside the maximisation depends only on the cardinality $|\L|$ of $\L$, not on $\L$ itself.)

In channel coding, the main interest is in finding the value of $M^*$ for different channels -- the limit  
  \[ \lim_{T\to\infty} \frac{\log M^*}{T} = \max_X \I(X:Y) =:c \]
exists and is called the channel capacity (see Definitions 1.6 and 1.10).

In group testing, for the deterministic case, Atia and Saligrama
\cite[Theorem V.1]{AtiaSaligrama} showed that we have
  \[ \overline T = O(K \log N) \qquad \text{as $K \to\infty$ and $N\to\infty$.} \]
a bound that Sejdinovic and Johnson \cite[Theorem 2]{Sejdinovic} improved to
  \[ \overline T \sim \e K \frac{ \log (K(N-K)) }{\log K} =
              \e K \left(1 + \frac{\log(N-K)}{\log K} \right) \quad \text{as $K \to\infty$ and $N\to\infty$.} \]
(Here, as elsewhere, we use $f(x) \sim g(x)$ to mean that $f(x)/g(x) \to 1$.)  
  
Here, we will attempt to find to what range of \defn{group testing channels} the result of Atia and Saligrama can be extended, and some bounds on $\overline T$ -- and hence $T^*$ -- for those channels. We also investigate further insights that channel coding can give to group testing.

\section{Channels}

In channel coding, many different types of communication can be modelled by using different channels.  Recall from Definition 1.1 that a communication channel is defined by stating what inputs $x \in \mathcal X$ and outputs $y \in \mathcal Y$ the channel can have, and what the probability $p(y \given x)$of each output is given each input.


We want to do the same for group testing.  The input is already constrained: there are $N$ items each of which can be in ($x_i = \one$) or not in ($x_i = \zero$) the pool, so the input alphabet is $\zeroone^N$.  So we must define the output alphabet and the probability function.

We will assume that in a testing pool there is no `order' to the items, nor will any elements not placed in the pool affect the outcome of the test, nor can we distinguish between the items other than whether or not they are defective.  Hence, the outcome can only depend on two things, the number of items in the test pool $n$, and the number of those items that are defective $k$.

\addcontentsline{lof}{figure}{\numberline{\textbf{Definition 6.4}} Group testing channel}
\begin{definition} \label{channel}
  A \defn{group testing channel} consists of
    \begin{itemize}
      \item an output alphabet $\mathcal Y$,
      \item a probability transition function $p(y \given n, k)$ relating
        the number of items $n$ and defectives $k$ in a testing pool to the
        test outcome $y$.
    \end{itemize}
\end{definition}

The problem of group testing is then to come up with test designs that have a low
error probability for as few tests as possible.

\addcontentsline{lof}{figure}{\numberline{\textbf{Definition 6.5}} Testing pool, test design}
\begin{definition}
  A \defn{testing pool} for $N$ items consists of a vector $\vec x = (x_i) \in \zeroone^N$,
  where $x_i = \one$ denotes that item $i$ is included in the pool and $x_i = \zero$
  denotes that item $i$ is not included in the pool.  We define $n := |\{ i : x_{i}=\one\}|$
  to be the total number of items in the pool and $k := |\{ i \in \K : x_{i}=\one\}|$ to be the
  number of defective items in the pool.

  A \defn{test design} of $T$ tests for $N$ items consists of
    \begin{itemize}
      \item a sequence $(\vec x_1, \vec x_2, \dots, \vec x_T)$
         of $T$ \defn{testing pools} (which can be summarized by the \defn{testing matrix}
         $\mat X = (x_{it}) \in \zeroone^{N \times T}$);
      \item a \defn{defective set detection function} $\Khat \colon \mathcal Y^T \to [N]^{(K)}$,
        where $[N]^{(K)}$ is the collection of subsets of $\{1,2,\dots, N\}$ of size $K$.
    \end{itemize}
\end{definition}

We can now describe the deterministic case discussed earlier
as an example of a group testing channel under Defintion \ref{channel}.

\addcontentsline{lof}{figure}{\numberline{\textbf{Definition 6.6}} Deterministic channel}
\begin{definition}
  The \defn{deterministic channel} has output alphabet
  $\mathcal Y = \{ \zero, \one \}$ and probability transition function
    \[ p(\one \given n,k) = \begin{cases} 0 & \text{if $k = 0$,} \\
                                        1 & \text{if $k \geq 1$,} \end{cases} 
       \qquad
       p(\zero \given n,k) = \begin{cases} 1 & \text{if $k = 0$,} \\
                                        0 & \text{if $k \geq 1$.} \end{cases} \]
\end{definition}

Atia and Saligrama \cite[Subsection II-C]{AtiaSaligrama} also studied two ways in which error could be introduced into group testing,
which they called the additive and dilution models.  They showed that their main result (Theorem 6.3) also holds true for these two channels.

In the additive model, a negative pool can actually return a false positive result, with some fixed probability $q$.
This could happen in our interference graph example from Section 6.1 if we included a noise term $Z$.  This model
can be recast as an example of a group testing channel.

\addcontentsline{lof}{figure}{\numberline{\textbf{Definition 6.7}} Addition channel}
\begin{definition}
  The \defn{addition channel} with addition probability $q>0$ has output alphabet
  $\mathcal Y = \{ \zero, \one \}$ and probability transition function
    \[ p(\one \given n,k) = \begin{cases} q & \text{if $k = 0$,} \\
                                        1 & \text{if $k \geq 1$,} \end{cases} 
       \qquad
       p(\zero \given n,k) = \begin{cases} 1-q & \text{if $k = 0$,} \\
                                        0 & \text{if $k \geq 1$.} \end{cases} \]
\end{definition}

Atia and Saligrama \cite[Table 1]{AtiaSaligrama} calculated that for the addition channel,
as $K\to\infty$ and $N\to\infty$, we have
  \[ \overline T = O \left( \frac{K \log N}{1-q} \right) , \]
Using the work of Sejdinovic and Johnson \cite[Theorem 6]{Sejdinovic} we can improve this to
  \[ \overline T \sim\ e  K \frac{ \log (K(N-K)) }{\log \frac 1q} , \]
by setting $u = 0$ and optimising $\alpha = 1$ in their equation (22) and rearranging.
(Interestingly, this is discontinuous with the deterministic channel at $q = 0$.)

The dilution model describes the case where a very small number of defective items in a testing pool might be `drowned out' by the nondefective items. Specifically, any defective items in the test may each evade the test independently with some probability $u$, with the potential to cause false negative results. This could happen in our interference graph example from Section 6.1 if we did not neglect the possibility that superposed interfering signals can cancel each other out.

\addcontentsline{lof}{figure}{\numberline{\textbf{Definition 6.8}} Dilution channel}
\begin{definition}
  The \defn{dilution channel} with dilution probabilty $u\geq 0$ has output alphabet
  $\mathcal Y = \{ \zero, \one \}$ and probability transition function
    \[ p(\one \given n,k) = \begin{cases} 0 & \text{if $k = 0$,} \\
                                       1-u^k & \text{if $k \geq 1$,} \end{cases} 
       \qquad
       p(\zero \given n,k) = \begin{cases} 1 & \text{if $k = 0$,} \\
                                        u^k & \text{if $k \geq 1$.} \end{cases} \]
\end{definition}

Atia and Saligrama \cite[Table 1]{AtiaSaligrama} calculated that for the addition channel we have
  \[ \overline T = O \left( \frac{K \log N}{(1-u)^2} \right) , \]
Again, using the work of Sejdinovic and Johnson \cite[Theorem 6]{Sejdinovic} we can improve this to
  \[ \overline T \sim \e^{1+u+f(u)} K \left(1 + \frac{\log(N-K)}{\log K} \right) \]
where
  \[ 0 \leq f(u) \leq \frac{u^2}{1-u} = u^2 + u^3 + \cdots ; \]
we do this by optimising $\alpha = 1/(1-u)$ in their equation (24) and rearranging.  In other words, the addition channel requires about $e^u \approx 1+u$ times as many tests as the deterministic channel, for small $u$.

Sejdinovic and Johnson also considered a channel that combines the additive and dilutive noise models.

\addcontentsline{lof}{figure}{\numberline{\textbf{Definition 6.9}} Addition/dilution channel}
\begin{definition}
  The \defn{addition/dilution channel} with addition probability $q>0$ and dilution probabilty $u\geq 0$ has output alphabet
  $\mathcal Y = \{ \zero, \one \}$ and probability transition function
    \[ p(\one \given n,k) = \begin{cases} q & \text{if $k = 0$,} \\
                                       1-u^k & \text{if $k \geq 1$,} \end{cases} 
       \qquad
       p(\zero \given n,k) = \begin{cases} 1-q & \text{if $k = 0$,} \\
                                        u^k & \text{if $k \geq 1$.} \end{cases} \]
\end{definition}

Working on the unproven assumption that Atia and Saligrama's result (Theorem 6.3) also holds for the addition/dilution channel, Sejdinovic and Johnson \cite[Theorem 3]{Sejdinovic} prove
a similar result for the addition/dilution channel.  (Although note that Sejdinovic and Johnson \cite[reference 1]{Sejdinovic} were working from an earlier preprint of Atia and Saligrama's paper \cite[version 2]{AtiaSaligrama}.  This had a different, and less rigorous, proof based on typical set decoding, rather than Gallager's maximum likelihood approach as in the most recent version \cite[version 4]{AtiaSaligrama}.)

Atia and Saligrama and Sejdinovic and Johnson defined their channels in terms of complicated Boolean sums and products of random vectors and matrices with the testing matrix $\mat X$. But our definition in terms of probability transition functions makes the behaviour of such channels clearer, and should allow the proof of more universal theorems.  It also makes it much easier to define new channels to model testing behaviour -- and many other existing models can be reformulated as group testing channels.

\addcontentsline{lof}{figure}{\numberline{\textbf{Definition 6.10}} More group testing channels}
\begin{definition}
The \defn{erasure channel} is a model that works like the deterministic channel, but fails to produce a result with some fixed erasure probability $\epsilon$.  That is, $\mathcal Y = \{ \zero, \quest, \one \}$ and
  \[ p(\one \given n,k) = \begin{cases} 0 & \text{if $k = 0$,} \\
                                 1-\epsilon & \text{if $k \geq 1$,} \end{cases} 
       \ 
       p( \quest \given n,k) = \epsilon,
       \ 
       p(\zero \given n,k) = \begin{cases} 1-\epsilon & \text{if $k = 0$,} \\
                                        0 & \text{if $k \geq 1$.} \end{cases} \]
                                        
The \defn{dilution threshold} only gives a positive result if a sufficient proportion of the tested items are defective, above some threshold $\theta \in (0,1)$.  That is, $\mathcal Y = \{ \zero, \one \}$ and
 \[ p(\one \given n,k) = \begin{cases} 0 & \text{if $k/n < \theta$,} \\
                                    1 & \text{if $k/n \geq \theta$,} \end{cases} 
       \qquad
       p(\zero \given n,k) = \begin{cases} 1 & \text{if $k/n < \theta$,} \\
                                 0 & \text{if $k/n \geq \theta$.} \end{cases} \]

The \defn{counting channel} gives as the output the number of defective items in the set.  That is, the probability transition function $p(y \given n,k)$ defined implicitly be the relation
  \[ Y = k . \]
It's worth noting that group testing under the counting channel model is equivalent to $\zero$--$\one$ compressed sensing with sparsity exactly $k$.  (A model equivalent to the count channel has previously been studied by Shapiro and Fine \cite{Fine}, Erd\H os and R\'enyi \cite{ErdosRenyi2}, and others -- for more details see the textbook of
Du and Hwang \cite[Section 11.2]{DuHwang}.)
                                 
The \defn{overflow channel} gives as a result exactly how many defective items were in the test, up to some limit $l$.  That is, $\mathcal Y = \{ 0,1,2,\dots, l\}$ and probability function definied implicitly by the relation
  \[ Y = \max\{k,l\} . \]
(When $l=n$, this is equivalent to the counting channel; when $l=1$ this is equivalent to the deterministic channel; when $l=2$ this is the message collision model
of Capetanakis \cite{Capetanakis}.)
  
The \defn{symmetric channel} gives a negative result if all items are nondefective, a positive result if all items are defective, and an uncertain result if there is a mixture of defective and nondefective items. That is, $\mathcal Y = \{ \zero, \quest, \one \}$ and relation
  \[ Y = \begin{cases} \zero  & \text{if $k=0$} \\
                       \quest & \text{if $0<k<n$} \\
                       \one   & \text{if $k=n$.} \end{cases} \]
(A model equivalent to the symmetric channel has previously been studied by Sobel, Kumar, and Blumenthal \cite{Sobel} and Hwang \cite{Hwang}.)
\end{definition}                  

These are just a few examples -- many more realistic error models for group testing can be formulated
as channels, and perhaps wider use can be made of this new concept.

\section{When only defects matter}

Note that for many of the group testing channels we have mentioned -- including all those studied by Atia and Saligrama and by Sejdinovic and Johnson -- the output depends only on $k$, the number of defects in the test, and not on $n$, the total number of items in the test.  In other words, `only defects matter', and the number of nondefects in the test is irrelevant.

\addcontentsline{lof}{figure}{\numberline{\textbf{Definition 6.11}} Only-defects-matter property}
\begin{definition}
  A channel $(\mathcal Y, p)$ whose probability function $p(y \given n,k) = p(y \given k)$ is
  dependent only on $k$ and not on $n$ is said to have the
  \defn{only-defects-matter} property.
\end{definition}

In the examples from Defintion 6.11 and earlier, the deterministic, addition, dilution, addition/dilution, erasure, counting,
and overflow channels have the only-defects-matter property.  The dilution threshold and symmetric channels do not have the
only-defects-matter property 

Another way to state the only-defects-matter property is that for a channel where only defects matter, we can make the simplification
  \[ \Prob (Y \given \vec X) = \Prob (Y \given \vec X_{\mathcal K }) . \]
Making this simplification is crucial to the proof of Atia and Saligrama \cite[for example Section III.A]{AtiaSaligrama}.  
Indeed, this is the only specific point about the deterministic, additive, and dilution channels
that Atia and Saligrama use in their proof.  Hence, we have the following:

\addcontentsline{lof}{figure}{\numberline{\textbf{Theorem 6.12}} Group testing theorem: only-defects-matter case}
\begin{theorem}  \label{ODMthm}
  Consider a group testing channel where only defects matter.
  Let $T^* = T^*(N,K,\epsilon)$ be the minimum number of tests necessary
  to identify $K$ defects among $N$ items with error probability at most $\epsilon \in (0,1)$.
  Then $T^* \leq \overline T + o(\log N)$ as $N\to\infty$, where
    \[ \overline T =  \min_p \max_{\L\subset\K}
                    \frac{\log_2 \binom{N-K}{|\L|} \binom{K}{|\L|}}
                         {\I (\vec X_{\K\setminus\L} : \vec X_{\L}, Y)} . \]
\end{theorem}
The proof is the same as that of Atia and Saligrama's theorem \cite[Theorem III.1]{AtiaSaligrama}.
We briefly outline the proof here.

\begin{proof}[Sketch proof]
  We need to analyse the error probability of a test design.  To do this, we will analyse the probability
  our estimated defective set $\Khat$ of cardinality $K$ overlaps with the true defective set $\K$ on a set
  $\mathcal L$.
  
  Given a set $\L$ There are $\binom{N-K}{K-|\L|}$ such sets, and $\binom{K}{|\L|}$ sets $\L$ of each
  possible cardinality.
  Using a technique similar to Gallager's proof of Shannon's coding theorem, we can bound the
  probability that we make an error in $|\L|$ places as
    \[ \epsilon \leq \binom{N-K}{K-|\L|} \binom{K}{|\L|} 2^{-T\I (\vec X_{\K\setminus\L} : \vec X_{\L}, Y)} . \]
  Hence we require
    \[ T > \frac{\log \binom{N-K}{K-|\L|} \binom{K}{|\L|}}
                         {\I (\vec X_{\K\setminus\L} : \vec X_{\L}, Y)} . \]
                         
  We need this to be true for every $\L \subset \K$, and can optimise the result over the test design parameter $p$.
\end{proof}

Identifying the only-defects-matter property as the crucial factor for proving Theorem \ref{ODMthm} means
that Sejdinovic and Johnson's bound on $\overline T$ for the addition/dilution channel is now rigorously proven.

Whether Theorem \ref{ODMthm} -- or a similar theorem -- holds for channels without the only-defects-matter property is an open problem.

Calculating, or finding good bounds for, the value of $\overline T$ or $T^*$ for the erasure, counting, and overflow channels is also an open problem.

\section{Converse part and adaptive testing}

Atia and Saligrama \cite[Theorem IV.1]{AtiaSaligrama} also provide a lower bound on the number of tests needed in group
testing.  The proof is along the lines of Shannon's converse to the channel coding theorem,
and uses Fano's inequality.  As before, the Atia--Saligrama proof in fact applies to all channels where only defects matter.

\addcontentsline{lof}{figure}{\numberline{\textbf{Theorem 6.13}} Group testing theorem: converse part}
\begin{theorem}\label{conversegt}
    Consider a group testing channel where only defects matter.
  Let $T^* = T^*(N,K,\epsilon)$ be the minimum number of tests necessary
  to identify $K$ defects among $N$ items with error probability at most $\epsilon \in (0,1)$.
  Then $T^* \geq \underline T - o(\log N)$ as $N\to\infty$, where
    \[ \underline T =  \min_p \max_{\L\subset\K}
                    \frac{\log_2 \binom{N-|\L|}{K-|\L|} }
                         {\I (\vec X_{\K\setminus\L} : \vec X_{\L}, Y)}. \]
\end{theorem}
  
Note that, unlike for channel coding, the bounds $\overline T$ and $\underline T$ do not coincide.  Therefore, the exact number of tests needed for group testing to work is not known.

We now present a proof of the converse that is slightly simpler than Atia and Saligrama's.  Our proof is based on theirs \cite[Theorem IV.1]{AtiaSaligrama}, with some simplifications based on the standard proof of the converse of Shannon's coding theorem, as exposited by Cover and Thomas \cite[Section 7.9]{CoverThomas}.

\begin{proof}
  Suppose a genie reveals to us some subset $\mathcal L \subset \K$ of the defective set, leaving us
  to work out the remaining $K-|\mathcal L|$ defective items.  Given $\mathcal L$, let $\mathcal K$ 
  be the random defective set chosen uniformly among the possible
  $\binom{N-|\mathcal L|}{K-|\mathcal L|}$ sets of size $K$ of which $\mathcal L$ is a subset.
  
  Then we have the following:
    \begin{align}
      \log \binom{N-|\L|}{K-|\L|} 
        &= \HH(\K \mid \L)  \tag{definition of entropy} \\
        &= \HH(\K \mid \Khat, \L) + \I(\K : \Khat \mid \L)  \tag{HB7} \\
        &\leq 1 + \epsilon \log \binom{N-|\L|}{K-|\L|} + \I(\K : \Khat \mid \L)
             \tag{Fano's inequality} \\
        &\leq 1 + \epsilon \log \binom{N-|\L|}{K-|\L|} + \I(\mat X_{\K\setminus\L} : \vec Y \mid \mat X_\L) ,  
             \label{hello}
    \end{align}
  where the final step is uses the data-processing inequality and the fact that
  only defects matter.
    
  We can bound the mutual information term in \eqref{hello} as
    \begin{align}
      \I(&\mat X_{\K\setminus\L} : \vec Y \mid \mat X_\L) \notag \\
        &\quad {}= \HH(\vec Y \mid \mat X_\L) - \HH(\vec Y \mid \mat X_\K)  \tag{HB7} \\
        &\quad {}= \sum_{t=1}^T \left( \HH(Y_t \mid Y_1, \dots, Y_{t-1}, \mat X_\L)
             - \HH(Y_t \mid Y_1, \dots, Y_{t-1}, \mat X_{\K}) \right)
             \tag{chain rule} \\
        &\quad {}= \sum_{t=1}^T \left( \HH(Y_t \mid Y_1, \dots, Y_{t-1}, \mat X_\L)
             - \HH(Y_t \mid \vec X_{\K t}) \right)
             \tag{memorylessness} \\
        &\quad {}\leq \sum_{t=1}^T \left( \HH(Y_t \mid \vec X_{\L t})
             - \HH(Y_t \mid \vec X_{\K t}) \right)
             \tag{conditioning reduces entropy} \\
        &\quad {}= \sum_{t=1}^T \I(\vec X_{\K\setminus\L \  t} : Y_t \mid \vec X_{\L t}) 
             \tag{HB7} \\
        &\quad {}= T \I(\vec X_{\K\setminus\L} : Y \mid \vec X_{\L}) .
             \label{helloagain}
    \end{align}
  (Here, we write $\vec X_{\K t} := (X_{it} : i \in \K)$  for fixed $t$.
  
  But we can rewrite this mutual information as
    \begin{equation}
      \I(\vec X_{\K\setminus\L} : Y \mid \vec X_{\L})
        = \I (\vec X_{\K\setminus\L} : \vec X_{\L}, Y) - \I (\vec X_{\K\setminus\L} : \vec X_{\L})
        = \I (\vec X_{\K\setminus\L} : \vec X_{\L}, Y) , \label{oncemore}
    \end{equation}
  since $\vec X_{\K\setminus\L}$ and $\vec X_{\L}$ are independent.
  
  Putting together \eqref{oncemore}, \eqref{helloagain}, and \eqref{hello}, we get
    \[ \log \binom{N-|\L|}{K-|\L|}
         \leq 1 + \epsilon \log \binom{N-|\L|}{K-|\L|} + T \I (\vec X_{\K\setminus\L} : \vec X_{\L}, Y) . \]
  We can rearrange this to get
    \[ \epsilon \geq 1 - T \frac{\I (\vec X_{\K\setminus\L} : \vec X_{\L}, Y)}{\log \binom{N-|\L|}{K-|\L|}}
         - \frac{1}{\log \binom{N-|\L|}{K-|\L|}} . \]
  Sending $N \to \infty$, it is clear that we require
    \[ T \geq \frac{\log \binom{N-|\L|}{K-|\L|}}{\I (\vec X_{\K\setminus\L} : \vec X_{\L}, Y)} \]
  to get the error probability arbitrarily low.
  
  This has to be true for all $\L \subset \K$, and we can optimise over the test inclusion
  parameter $p$.  This gives the result.
\end{proof}

\bigskip

\noindent So far, we have been looking at \defn{nonadaptive group testing}, where all the test pools are decided on ahead of time.

Instead, we could consider \defn{adaptive group testing} where tests are performed sequentially, and the makeup of a testing pool can depend on the results of previous tests.  That is, the $x_{it}$ are functions of the previous test outcomes $(y_1, y_2, \dots, y_{t-1})$.

\addcontentsline{lof}{figure}{\numberline{\textbf{Definition 6.14}} Adaptive test design}
\begin{definition}
  An \defn{adaptive test design} of $T$ tests for $N$ items consists of a sequence
  $(\vec x_1, \vec x_2, \dots, \vec x_T)$ of $T$ testing pools, where the
  testing pool for the $t$th test $\vec x_t = \vec x_t(y_1, y_2, \dots, y_{t-1})$
  can depend on earlier test outcomes.
\end{definition}

This is analogous to channel coding with feedback, where the encoding function $x_t$ at time $t$ can depend on previous channel outputs $y_1, y_2, \dots, y_{t-1}$.

As Shannon \cite[page 15]{Shannonfeedback} showed, for discrete memoryless channels, feedback does not increase the channel capacity.  (Although it can help in simplifying encoding and decoding \cite[Section 7.12]{CoverThomas}.)

Due to the non-tightness of the bounds on testing in the nonadaptive case, we will not be able to show that adaptive group testing requires the \emph{same} number of tests as nonadaptive testing, but we will be able to show that it obeys the same lower bound and requires no more tests than the nonadaptive case.

This is (as far as we are aware) the first application of information theoretic techniques to
adaptive group testing.

\addcontentsline{lof}{figure}{\numberline{\textbf{Theorem 6.15}} Group testing theorem: adaptive case}
\begin{theorem}
  Let $T^*_{NA}$ and $T^*_A$ (dependent on $N$,$K$, and $\epsilon$) be the minimum number of tests necessary
  to identify $K$ defects among $N$ items with error probability at most $\epsilon \in (0,1)$ for
  nonadaptive and adaptive group testing respectively.
  
  Then, as $N\to\infty$, we have the inequalities
    \[ \underline T - o(\log N) \leq T^*_A \leq T^*_{NA} \leq \overline T + o(\log N) \]
  where $\underline T$ and $\overline T$ are as in Theorems 6.3 and 6.13.
\end{theorem}

\begin{proof}
  The third inequality was proven in Theorem 6.3.  The second inequality is trivial, as nonadaptive
  group testing is merely a special case of adaptive group testing where the tester chooses
  to ignore the information of previous test results.
  
  To prove the first inequality, we adapt the proof of
  Theorem \ref{conversegt}, and Shannon's proof that feedback fails to improve channel
  capacity \cite[Theorem 6]{Shannonfeedback}, as exposited by Cover and Thomas \cite[Theorem 7.12.1]{CoverThomas}.
  
  We begin exactly the same  way as the proof of Theorem \ref{conversegt}, to get
    \begin{equation} 
      \log \binom{N-|\L|}{K-|\L|} 
        \leq 1 + \epsilon \log \binom{N-|\L|}{K-|\L|} + \I(\K : \Khat \mid \L) . \label{first}
    \end{equation}
  We again use the data processing inequality (but in a slightly different way)
  and the only-defects-matter property on the mutual information term in \eqref{first},
  to write
    \begin{align}
      \I(&\K : \Khat \mid \L) \notag \\
        &\leq \I(\K\setminus\L : \vec Y \mid \L) \tag{data processing and only-defects-matter} \\
        &= \HH(\vec Y \mid \L) - \HH(\vec Y \mid \K) \tag{HB7} \\
        &= \sum_{t=1}^T \big( \HH(Y_t \mid Y_1, \dots, Y_{t-1}, \K)
             - \HH(Y_t \mid Y_1, \dots, Y_{t-1}, \L) \big) \tag{chain rule} \\
        &= \sum_{t=1}^T \big( \HH(Y_t \mid Y_1, \dots, Y_{t-1}, \L, \vec X_{\L t})
             - \HH(Y_t \mid Y_1, \dots, Y_{t-1}, \K, \vec X_{\K t}) \big)
             \tag{$\vec X_{\J t}$ a function of $Y_1, \dots, Y_{t-1}, \J$} \\
        &\leq \sum_{t=1}^T \big( \HH(Y_t \mid \vec X_{\L t})
             - \HH(Y_t \mid Y_1, \dots, Y_{t-1}, \K, \vec X_{\K t}) \big)
             \tag{cond.\ reduces entr.} \\     
        &= \sum_{t=1}^T \big( \HH(Y_t \mid \vec X_{\L t}
             - \HH(Y_t \mid \vec X_{\K t}) \big) , \label{third}
    \end{align}
  where \eqref{third} is because, conditional on $\vec X_{\K t}$, we know that $Y_t$ is independent
  of previous $Y$s and the defective set $\K$.
  
  We can now pick back up with the proof of Theorem \ref{conversegt}, two lines above \eqref{helloagain}, to complete the proof.
\end{proof}

It is tempting to wonder if in fact non-adaptive and adaptive group testing require exactly the same number of tests $T^*_A = T^*_{NA}$.
This would be in contrast to results for the deterministic model in the zero-error case, where it is known that adaptive group testing can be performed in strictly fewer tests than nonadaptive (in the $K\to \infty, N\to\infty$ regime) \cite[page 139]{DuHwang}.

Similarly, Shannon showed that feedback for channel coding does help in the zero-error case, but not the arbitrarily-small-error case \cite{Shannonfeedback}.

\section{Further work}

The investigation of applying information theoretic ideas to group testing is at a very early
stage and there are many open questions.

\begin{description}
\item[Can we find an analogue of Theorem \ref{ODMthm} where nondefects also matter?] \mbox{} \newline
That is, can we drop the only-defects-matter property.
The proof of Theorem \ref{ODMthm} relies on the fact that, given $\K \cap \Khat =: \L$, the random variables
$\Prob (Y \given \vec X, \K \text{ defective})$ and $\Prob (Y \given \vec X, \Khat \text{ defective})$ are
conditionally independent.  While this is no longer true when nondefects matter too, they are still conditionally
independent given $\L$ and the total number of items $n$ in the test.  Perhaps this promises a way forwards.

\item[What is $T^*$ for different channels?] Is there a reliable method to calculate, exactly or approximately, $T^*$ for different
channels.  Is there an easy way to do so? What are the minimising choices of $p$?

\item[Can we close the gap between $\overline T$ and $\underline T$] to tighten our bounds?

\item[Does nonadaptive group testing require more tests] than adaptive? That is, does $T^*_A = T^*_{NA}$ asymptotically.

\item[The wider view:] What other information-theoretic techniques can be useful?
Sejdinovic and Johnson \cite{Sejdinovic} have tried using message-passing decoding algorithms
for computationally feasible defective set detection. Cheraghchi and coauthors \cite{Cheraghchi} have
used information theoretic techniques to study group testing on graphs. What else is there to try?
\end{description}

That's just a few brief suggestions.

\section*{Notes}
\addcontentsline{toc}{section}{Notes}

This chapter has benefited from discussions with Oliver Johnson, Dino Sejdinovic, and Robert Piechocki.

Group testing was first studied by Dorfman \cite{Dorfman} in the context of testing soldiers' blood for syphilis.

The information theoretic approach to group testing is due to Atia and Saligrama \cite{AtiaSaligrama}. A recent paper
of Sejdinovic and Johnson was also useful \cite{Sejdinovic}.

The textbook of Du and Hwang \cite{DuHwang} was useful for material on group testing. The textbook of Cover and Thomas
\cite{CoverThomas} was useful for material on channel coding.

\chapter{Conclusions and further work}

In this thesis, we've examined a number of different ways of combating interference in wireless networks.

\begin{itemize}
\item In Chapter 3 we saw how the simple interference-as-noise technique can be effective over short
hops in well-structured networks.

\item In Chapter 4, we saw how interference alignment can give tight bounds on the performance
of large random networks.

\item In Chapter 5, we examined the tradeoff between delay and communication rate when using
ergodic interference alignment.

\item In Chapter 6, we examined how group testing can help with network performance, and how
channel coding techniques can help with group testing.
\end{itemize}

As we've gone through, we have left a number of pointers to open questions and further
work (in particular, see Sections 3.4, 4.6, 5.6, and 6.6).  We now have the opportunity to
take a brief look at the wider picture.

\begin{itemize}
\item While interference alignment has been an important theoretical breakthrough, it
has had few practical benefits to date.  The work in Chapter 5 of this thesis on delay--rate
tradeoff is a step in the right direction, but more work is needed.  How can we reduce the
complexity of schemes?  Can we reduce the amount of channel state information required?
Can we reduce blocklengths further?  Are the schemes robust to errors in channel estimation?

\item There is no Shannon's channel coding theorem for networks.  That is, there is no general
result that tells one how to calculate the capacity region (or even just the sum-capacity) of
a network.  (The best current result is the cutset bound of Cover and Thomas \cite[Theorem 15.10.1]{CoverThomas}.)
Even a result for the Gaussian case seems far off -- Jafar refers to a result in the Gaussian case for only
interference networks as ``the holy grail of network information theory'' \cite[page 1]{Jafar}.

\item Most capacity theorems for networks -- including those in this thesis -- are nonconstructive.
That is, no explicit codes are given.  Are some coding schemes particularly well suited to
interference alignment schemes, or can standard codes (like low-density parity-check codes)
be adapted to this new area?

\item Theorems about sum-capacity (such as ours in Chapter 4) do not give everything we want to know
about a network.  After all, a network operating solely at its optimal sum-rate may provide very poor
performance for some unfortunate users.  Concepts of fairness and cooperation also need to be taken into account.

\item As wireless devices become ever more ubiquitous, many engineering problems in this area
become ever more severe.  Networks must be set up `ad hoc' with little prior knowledge (the example
in Section 6.1 shows one way of approaching this); also, interference
will become a much bigger problem, and `green' technologies with low power consumption will be important.

\item Work on the connection between group testing and channel coding is at a very early stage --
in Section 6.6 we listed a number of future directions for research.  Are other mathematical problems
that involve inference about unknown quantities approachable from an information theoretic
point of view?  How far can Shannon's work at a 1940s telephone company take us?
\end{itemize}

%% file: chapters/endmatter.tex
\addtocontents{toc}{\protect\addvspace{10 pt}}
\addcontentsline{toc}{section}{\protect\hspace{-1.5 em} \emph{References}}
\nobibintoc



%